\documentclass[algo2e, 10pt,twocolumn]{IEEEtran}
\usepackage{soul,xcolor,lipsum,enumitem}

\usepackage{wrapfig,dirtytalk}
\usepackage{floatrow}

\usepackage{algorithm}
\usepackage{algcompatible}
\usepackage{algpseudocode}
\usepackage{comment}

\usepackage{tablefootnote}
\usepackage{amsfonts}
\usepackage{amsmath}
\usepackage{upgreek}
\usepackage{amssymb}
\usepackage{amsthm}
\usepackage[ansinew]{inputenc} 
\usepackage[colorlinks=true, allcolors=blue]{hyperref}
\usepackage{dsfont}
\usepackage{mathtools}
\usepackage{graphicx}
\usepackage[skip=1pt,font=scriptsize]{subcaption}
\usepackage{import}
\usepackage{multirow}
\usepackage{cite}
\usepackage[export]{adjustbox}
\usepackage{breqn}
\usepackage{mathrsfs}
\usepackage{acronym}
\usepackage{balance}
\usepackage{setspace}
\usepackage{stackengine}
\usepackage{steinmetz}
\usepackage[skip=2pt,font=footnotesize]{caption}

\DeclareMathOperator*{\argmin}{arg\,min}

\newcommand{\jj}{\mathrm{j}}

\newcommand{\proposed}{NCOTA-DGD}

\newtheorem{theorem}{Theorem}
\newtheorem{lemma}[theorem]{Lemma}

\newtheorem{assumption}{Assumption}
\theoremstyle{definition}
\newtheorem{defi}{Definition}
\newtheorem{example}{Example}
\newtheorem{remark}{Remark}
\theoremstyle{remark}

\usepackage{cancel}
\usepackage{bbm}

\graphicspath{{./figures/}}
\newcommand{\secref}[1]{Sec.~\ref{#1}}

\newcommand{\e}{\boldsymbol{\epsilon}}
\newcommand{\ba}[2]{\begin{align}\label{#2}#1\end{align}}

\newcommand{\LL}{\hat{\boldsymbol{L}}}

\newcommand\bstheta{\boldsymbol{\theta}}

\newcommand\CS{Cauchy-Schwarz inequality}

\title{Non-Coherent Over-the-Air\\ Decentralized  Gradient Descent}

\author{Nicol\`o Michelusi
\thanks{N. Michelusi is with the School of Electrical, Computer and Energy Engineering, Arizona State University.
Part of this work appeared at IEEE ICC'23 \cite{ICC23} and at IEEE ICASSP'24 \cite{icassp24}.}
\thanks{This research has been funded in part by NSF under grant CNS-2129615.}
\vspace{-5mm}
}

\begin{document}

\thispagestyle{plain}
\pagestyle{plain} 
\setulcolor{red}
\setul{red}{2pt}
\setstcolor{red}

\maketitle

\begin{abstract}
Implementing Decentralized Gradient Descent (DGD) in wireless systems is challenging due to noise, fading, and limited bandwidth, necessitating topology awareness, transmission scheduling, and the acquisition of channel state information (CSI) to mitigate interference and maintain reliable communications. These operations may result in substantial signaling overhead and scalability challenges in large networks lacking central coordination. This paper introduces a scalable DGD algorithm that eliminates the need for scheduling, topology information, or CSI (both average and instantaneous). At its core is a Non-Coherent Over-The-Air (NCOTA) consensus scheme that exploits a noisy energy superposition property of wireless channels. Nodes encode their local optimization signals into energy levels within an OFDM frame and transmit simultaneously, without coordination. The key insight is that the received energy equals, \emph{on average}, the sum of the energies of the transmitted signals, scaled by their respective average channel gains, akin to a consensus step. This property enables unbiased consensus estimation, utilizing average channel gains as mixing weights, thereby removing the need for their explicit design or for CSI. Introducing a consensus stepsize mitigates consensus estimation errors due to energy fluctuations around their expected values. For strongly-convex problems, it is shown that the expected squared  distance between the local and globally optimum models vanishes at a rate of $\mathcal O(1/\sqrt{k})$ after $k$ iterations, with suitable decreasing learning and consensus stepsizes. Extensions accommodate a broad class of fading models and frequency-selective channels. Numerical experiments on image classification demonstrate faster convergence in terms of running time compared to state-of-the-art schemes, especially in dense network scenarios.
\end{abstract}
\vspace{-3mm}
\section{Introduction}
\label{intro}
Distributed optimization algorithms are used 
 across various domains such as remote sensing, distributed inference \cite{6494683},  estimation \cite{9224135}, multi-agent coordination \cite{Nedic2018}, and machine learning (ML)  \cite{YANG2019278}.  Typically, in these applications, $N$ devices with sensing, computation and wireless communication capabilities aim to find a global ${\mathbf w}^*\in\mathbb R^d$, solution of
\begin{align}
\label{global}
{\mathbf w}^*{=}\smash{\argmin_{\mathbf w\in\mathbb R^d}}
\ F(\mathbf{w}),
\text{ where }\vphantom{\sum}
\smash{F(\mathbf{w})\triangleq\frac{1}{N}\sum_{i=1}^Nf_i(\mathbf w),} \tag*{\{P\}}  
\end{align}
$f_i(\mathbf w)$ is a local function only known to node $i$,
 $\mathbf w$ is a $d$-dimensional parameter, so that $F(\mathbf w)$ is the network-wide objective.
For instance, in ML applications, $f_i(\mathbf w){=}\frac{1}{|\mathcal D_i|}\sum_{\boldsymbol{\xi}\in\mathcal D_i}\phi(\boldsymbol{\xi};\mathbf w)$ is the empirical loss over the local dataset $\mathcal D_i$, with loss function  $\phi(\boldsymbol{\xi};\mathbf w)$ on the datapoint $\boldsymbol{\xi}$.
In distributed linear regression, each node measures $\mathbf w^*$
via $\mathbf y_i{=}\mathbf A_i\mathbf w^*{+}\mathbf n_i$ corrupted by noise $\mathbf n_i$,
corresponding to the error metric $f_i(\mathbf w){=}\frac{1}{2}\Vert\mathbf y_i{-}\mathbf A_i\mathbf w\Vert_2^2$ for node $i$ \cite{10103556,9224135}.

Federated learning (FL) solves \ref{global} based on a client-server architecture,
in which the $N$  devices (the clients) interact
over multiple rounds
 with a parameter server (PS, such as a base station), acting as a model aggregator.
Nevertheless, the FL architecture encounters several challenges in wireless scenarios:
1) devices far away from the PS (e.g., in rural areas), may suffer from severe path loss conditions and blockages \cite{8870236};
2) uplink communications to the PS may be a severe bottleneck when $N$ is large, due to limited bandwidth  \cite{Lian17};
3) if the PS fails, e.g., due to a natural disaster, the whole system may break down.
Therefore, in crucial scenarios like swarms of unmanned aerial vehicles (UAVs) in remote areas \cite{9475989} or disaster response operations, solving \ref{global} in a fully decentralized way, without centralized coordination, is paramount \cite{8950073}.

\label{rev1resp2}
  A renowned algorithm solving \ref{global} in a fully decentralized fashion is \emph{Decentralized Gradient Descent} (DGD) \cite{Nedic2009grad,Yuan2016} and its stochastic gradient variant \cite{9241497}: each node ($i$) 
  owns a local copy $\mathbf w_{i}\in\mathbb R^d$ of $\mathbf w$ (its \emph{state}), and
transmits it to its neighbors in the network; upon receiving its neighbors' states, it then updates $\mathbf w_{i}$ as a weighted average of the received signals (the consensus signal $\mathbf c_i$), followed by a local gradient step,  illustrated here as non-stochastic for clarity,
\begin{align}
\mathbf w_{i}
\gets\mathbf c_i
-\eta\nabla f_i(\mathbf w_{i}),\text{ with }\vphantom{\sum}
\smash{\mathbf c_i=\sum_{j=1}^N\omega_{ij}\mathbf w_{j}.}
\tag*{\{DGD\}}  
\label{DGD}
\end{align}
Here,  $\omega_{ij}$ denotes a set of non-negative, symmetric ($\omega_{ij}{=}\omega_{ji}$) mixing weights, such that
$\sum_{j=1}^N\omega_{ij}{=}1$, defined over a mesh network ($\omega_{ij}{=}0$ if $i$ and $j$ are not direct neighbors).
These steps are repeated across the $N$ nodes and over multiple iterations, until a desired convergence criterion or deadline are met. 
Yet, in wireless systems affected by noise, fading, and interference, reliably communicating
the states
 $\mathbf w_j$ necessary to compute the consensus signal $\mathbf c_i$ can introduce delays. For example, with a TDMA scheme, the $N$
nodes broadcast their local state to their neighbors in pre-assigned time slots in a round-robin fashion. This strategy protects transmissions from mutual interference, thus facilitating reliable communications during the consensus phase. However, the time required to complete one DGD iteration scales proportionally with $N$.
Consequently, when $N$ is large, only a few iterations of DGD may be completed within a given deadline, resulting in poor convergence performance.
 Communication delays may be reduced by
quantizing $\mathbf w_j$ below machine precision \cite{9782148, Kajiyama2020, Magnusson2020, 9562482,9224135}, or via non orthogonal transmissions \cite{9563232}, at the cost of errors introduced in the computation of $\mathbf c_i$.

Consequently, there is a tension between accurately computing $\mathbf c_i$ and minimizing the communication delay of DGD iterations.
In fully decentralized wireless systems lacking centralized coordination, this tension is further exacerbated by the signaling overhead associated with tasks such as scheduling, channel estimation, link monitoring, and topology awareness:
 \begin{enumerate}[leftmargin=10pt]
\item \underline{Scheduling \& signaling overhead}: Mitigating the impact of unreliable wireless links requires scheduling of transmissions to manage interference, and the acquisition of channel state information (CSI, average or instantaneous) to compensate for signal fluctuations and link outages caused by fading. 
Executing tasks such as graph-coloring (used for scheduling in \cite{9563232,9322286,9517780}), pilot assignment to mitigate pilot contamination during channel estimation \cite{9205230}, monitoring large-scale propagation conditions, and reporting CSI feedback for power/rate adaptation  may require inter-agent coordination and entail severe signaling overhead.
\item \underline{Topology information \& weight design}:
The consensus step in \ref{DGD} requires local
topology awareness (set of neighbors in the mesh network),
necessitating continuous monitoring of link qualities and path loss conditions.
Furthermore, design of the mixing weights (e.g.,
using decentralized methods developed in 
 \cite{ROKADE2022110322,9222206}) incurs additional signaling overhead.
\end{enumerate}
The complexity of these tasks is further exacerbated in mobile scenarios, such as swarms of UAVs, where the connectivity structure changes frequently.
Therefore, achieving device-scalability of {DGD} in wireless systems, with a small signaling overhead footprint, remains an open challenge.

\vspace{-3mm}
\subsection{Novelty and Contributions}
\label{novelty}
 Addressing these challenges necessitates the design of decentralized optimization schemes tailored to wireless propagation environments.
To meet these goals, in this paper we present a \emph{Non-Coherent Over-the-Air} (NCOTA)-DGD algorithm.
Unlike conventional DGD approaches that rely on orthogonal transmissions to mitigate interference (e.g., via TDMA),
the core idea of \proposed\ is to allow simultaneous transmissions, thus achieving device-scalability. To do so,
\proposed\ 
exploits the superposition property of wireless channels to estimate the consensus signal $\mathbf c_i$, in line with over-the-air computing (AirComp) paradigms proposed for FL \cite{8952884,9014530,9515709,8870236,9382114,9042352}. 
Nevertheless, AirComp typically relies on \emph{coherent} alignment of signals at a common  base station, achieved via channel inversion at the transmitters. 
 In a decentralized setting, this condition cannot be met because signals are broadcast to multiple receivers, each receiving through a different channel.
 
To overcome this limitation, we develop a
\emph{non-coherent energy-based transmission technique} that operates without CSI, and
we exploit a \emph{noisy energy superposition property} of wireless channels at the receivers.
Specifically, nodes encode their local state $\mathbf{w}_i$ into transmitted energy levels and transmit simultaneously using a randomized scheme for half-duplex operation.
  Leveraging the circularly symmetric nature of wireless channels, the received signal energy is, on average, the sum of individual transmitted signal energies scaled by their respective channel gains, akin to the consensus signal $\mathbf{c}_i$, but with channel gains acting as mixing weights. Introducing a suitable consensus stepsize mitigates energy fluctuations in the received signal (errors in the consensus estimate).
The benefits of \proposed\ over conventional DGD schemes that rely on a mesh network \cite{9224135,9563232,9517780,9322286,9562482,9772390,9716792,9838891}  are two-fold:
\begin{enumerate}[leftmargin=10pt]
\item Concurrent transmissions obviate the need for scheduling or inter-agent coordination, resulting in fast (albeit noisy) iterations. Consequently, device-scalability (when $N$ is large) and scheduling-free operation are achieved.
\item
By leveraging intrinsic channel properties for signal mixing, explicit design and knowledge of the mixing weights or topology information are not required.
Furthermore, utilizing a non-coherent energy-based transmission and reception scheme obviates the need for CSI at transmitters or receivers.
Consequently, the  signaling overhead typically associated with these tasks is eliminated.
\end{enumerate}
 
Yet, these benefits come at the cost of consensus estimation errors due to random energy fluctuations, requiring careful  design of consensus and learning stepsizes to mitigate error propagation. Therefore, we develop a novel analysis of DGD with noisy consensus and noisy gradients, for the class of strongly-convex problems. 
 We prove that, by choosing the learning stepsize as $\eta_k\propto 1/k$ and the consensus stepsize as $\gamma_k\propto k^{-3/4}$ at the $k$th iteration,
the expected squared distance between the local and globally optimum models vanishes with rate $\mathcal O(1/\sqrt{k})$.
This result improves prior works \cite{9563232,9517780}, exhibiting non-vanishing errors,
and \cite{8786146} using fixed stepsizes.

Lastly, we propose extensions to cover a broad class of fading models and frequency-selective channels, including static channels as a special case. Numerical experiments on fashion-MNIST  \cite{fmnist} image classification,
 formulated as regularized cross-entropy loss minimization, demonstrate faster convergence compared to state-of-the-art algorithms, including quantized DGD \cite{8786146}, a device-to-device wireless implementation of DGD \cite{9563232}, and AirComp-based FL \cite{8870236}. Notably,
  our approach excels in densely deployed networks with large $N$, thus demonstrating its device-scalable design.

\subsection{Related work}
\label{relworks}
The study of distributed optimization algorithms to solve \ref{global} is vast and can be
traced back to the seminal work \cite{1104412}. 
Several studies enable decentralized algorithms to function over finite capacity channels,
 using methods such as
compression \cite{Taheri2020, Kovalev2020, Liao2021} or quantization  \cite{9224135,9782148, Kajiyama2020, Magnusson2020, 9562482}  of the transmitted signals. However, by assuming error-free communications, these works
 do not account for the impact of noise and fading in wireless channels.
The impact of unreliable communications is investigated in \cite{9562482,9772390,9716792,9838891}.
Yet, all these works 
 implicitly assume the use of orthogonal transmission scheduling to mitigate interference (e.g., via TDMA), hence may not be device-scalable (when $N$ is large), as observed earlier. 

 Recently, there has been significant interest in 
 achieving device-scalability by leveraging the waveform superposition properties of wireless channels via AirComp \cite{9562559}, first studied from an information-theoretic perspective in \cite{4305404}.
 However,
 most of these studies focus on a centralized FL architecture, where a base station functions as the PS  \cite{8952884,9014530,9515709,8870236,9382114,9042352,9815298}. As a result, they rely on assumptions that are difficult to satisfy in decentralized environments:
 \begin{itemize}[leftmargin=10pt]
 \item \underline{Coherent alignment}:
 The works \cite{8952884,9014530,9515709,8870236} rely on 
   channel inversion to achieve coherent alignment of the signals at a single receiver, the PS.
Instead, in decentralized settings, multiple receivers receive signals through different channels. This diversity of channels makes simultaneous coherent alignment across multiple receivers impossible.
 \item \underline{Signaling overhead}: Centralized FL schemes typically require meticulous power control and device scheduling to enforce power constraints and ensure device participation, requiring CSI acquisition and monitoring of average path loss conditions \cite{Faraz}.
 However, the ensuing signaling overhead, often presumed error-free or its impact not assessed, might not scale effectively to decentralized networks.
   \item  \underline{Massive-MIMO}: 
  The work \cite{9382114} overcomes the need for CSI at the transmitters by assuming a large number of antennas at the PS
  and by leveraging the channel hardening effect.
  However, this assumption is not applicable to edge devices (e.g., UAVs) with compact form factor.
 \item  \underline{Noise-free downlink}: All these papers (except \cite{9515709}) assume noise-free downlink. While reasonable for a PS with a typically larger power budget, 
  this assumption may not hold under stringent power constraints of edge devices.
 \end{itemize}
In contrast, \proposed\ does not require CSI (at either transmitters or receivers) and does not rely on a large number of antennas, thanks to the use of non-coherent energy-based signaling strategies. It  operates without  power control or scheduling of devices, thus achieving device-scalable and scheduling-free operation. These features make \proposed\ appealing to fully decentralized wireless systems.
 
 
 Notable exceptions in the AirComp literature with star-based topology include \cite{9815298,8974619} (no \emph{instantaneous} CSI) and \cite{9076343} (partial CSI).
 In \cite{9815298}, local gradients are first encoded using random vector quantization \cite{9740125}, and then transmitted via preamble-based random access
 over flat-fading channels. However, this scheme is not suitable for fully decentralized systems: 
1) The randomness introduced by random access increases the estimation variance at the receiver.
2) It relies on large antenna arrays at the PS to achieve vanishing optimality error, which may not be available, e.g., in small UAVs.
3) It hinges on average channel inversion to ensure unbiased gradient estimation at the PS. This is impractical in decentralized systems due to the signaling overhead associated with average CSI acquisition, and the diversity of channels across multiple receivers. 
4) Non-negligible propagation delays and multipath create frequency-selectivity.
To tackle 1-2), we design a \emph{deterministic} energy-based encoding strategy
and  stepsizes
 to mitigate errors at the receivers, supported by a novel convergence analysis of DGD with noisy consensus.
 To tackle 3), we design a signal processing scheme tailored to unbiased consensus estimation, with weights given by average path loss conditions, without relying on their explicit knowledge.
To tackle 4), we design mechanisms that enable operation over a broad class of frequency-selective channels, including static ones as a special case.
  In \cite{8974619}, blind over-the-air computation (BlairComp) of nomographic functions over multiple access channels is addressed using a Wirtinger flow method. However, \cite{8974619} BlairComp does not ensure unbiased estimates, crucial for proving convergence of FL and DGD algorithms. The paper \cite{9076343} uses channel phase correction at the transmitters to achieve coherent phase alignment at the PS. However, 
  such alignment condition cannot be achieved towards multiple receivers, calling for a non-coherent energy superposition technique, \proposed.

  
  Recent works have developed algorithms for decentralized FL that are robust to wireless propagation impairments \cite{9563232,9322286,9517780}.
  These works solve the optimization problem \ref{global}
  by implementing DGD over wireless channels, following \ref{DGD} or the equivalent formulation 
  in \eqref{ddd} with a consensus stepsize as in \cite{9563232}.
  These schemes mitigate interference by decomposing the network into smaller non-interfering subgraphs, via graph-coloring. In each subgraph, one device operates as the PS,  enabling the use of AirComp techniques, coupled with a suitable consensus enforcing step similar to \ref{DGD}. However, these schemes require topology awareness for graph-coloring, CSI feedback at the transmitters, power control, and scheduling of transmissions.  Moreover, they rely on design of the mixing weights $\omega$, using methods such as \cite{ROKADE2022110322,9222206}.
  Coordinating the network to effectively execute these tasks may entail severe signaling overhead, and their complexity may not scale well to large networks.
  In contrast, in this paper, we leverage average path loss as a proxy for mixing weights, eliminating the need for their explicit design, thus achieving device-scalability without CSI and inter-agent coordination.   

 Finally, the works \cite{9311931,9562522,9834707,9705093} proposed \emph{semi-decentralized} FL architectures, where
   edge devices collaborate with their neighbors unable to form a reliable direct link with the PS. However, these works inherit similar challenges outlined earlier under both centralized and decentralized architectures.
  
\subsection{Notation and Organization of the paper}
\noindent  For a set $\mathcal S{\subset}\mathbb R^d$, $\mathrm{int}(\mathcal S)$ is its interior, $\mathrm{bd}(\mathcal S)$ its boundary,
$\mathrm{conv}(\mathcal S)$ its convex hull.
For matrix $\mathbf A$,
$[\mathbf A]_{ij}$ is its $ij$th component,
$\mathbf A^\top$ its transpose, $\mathbf A^{\mathrm H}$ its complex conjugate transpose.
Vectors are  defined in column form.
 For (column) vector $\mathbf a$, 
 $[\mathbf a]_i$ is its $i$th component, $\Vert\mathbf a\Vert{=}\sqrt{\mathbf a^{\mathrm H}\mathbf a}$ its Euclidean norm.
 For random vector $\mathbf a$, we let $\Vert\mathbf a\Vert_{\mathbb E}{\triangleq}\sqrt{\mathbb E[\Vert\mathbf a\Vert^2]}$
 ($\Vert\mathbf a\Vert_{\mathbb E}{=}\Vert\mathbf a\Vert$ for deterministic $\mathbf a$),
 $\mathrm{var}(\mathbf a){\triangleq}\mathbb E[\Vert\mathbf a\Vert^2]{-}\Vert\mathbb E[\mathbf a]\Vert^2
$ (variance) and 
 $\mathrm{sdv}(\mathbf a){\triangleq}\sqrt{\mathrm{var}(\mathbf a)}$ (standard deviation).
 We write $\Vert\cdot\Vert_{\mathbb E|\mathcal F}$, $\mathrm{var}(\cdot|\mathcal F)$ and $\mathrm{sdv}(\cdot|\mathcal F)$ when the expectation is conditional on $\mathcal  F$.
  $\mathbf a{\odot}\mathbf b$ and $\mathbf a{\otimes}\mathbf b$ are the Hadamard (entry-wise) and Kronecker products of $\mathbf a,\mathbf b$.
 We define the following $n$-dimensional vectors/matrices (we omit $n$ when clear from context):
the $m$th standard basis vector $\mathbf e_{n,m}$,  with 
$m$th component ${=}1$, and $0$ otherwise; the all ones and all zeros vectors $\mathbf 1_n$, $\mathbf 0_n$;
the identity matrix $\mathbf I_n$. 
 $\mathbbm{1}[A]$ is the indicator of  event~$A$. 

 The rest of this paper is organized as follows. In Secs. \ref{sysmo}-\ref{M2}, we present the system model and \proposed. In \secref{convanalysis}, we delve into its convergence analysis, with proofs  provided in the Appendix.
In \secref{numres}, we present numerical results, followed by concluding remarks in \secref{conclu}.
Supplementary details, including additional proofs and extended numerical evaluations, are provided in the supplemental document.

\section{System Model and \proposed}
\label{sysmo}
Consider $N$ wirelessly-connected nodes solving \ref{global} in a  decentralized fashion.
We divide time into frames of duration $T$.
In this section, we consider a generic frame (iteration) $k$, and omit the dependence on $k$. 
In each iteration, we aim at emulating \ref{DGD}, expressed in the equivalent form
\begin{align}
\mathbf w_{i}
{\gets}
\mathbf w_{i}
{+}\gamma\mathbf d_i
{-}\eta\nabla f_i(\mathbf w_{i}),
\text{ with }
\vphantom{\sum}
\smash{\mathbf d_{i}{=}\sum_{j=1}^N\ell_{ij}(\mathbf w_{i}{-}\mathbf w_{j}),}
\label{ddd}
\end{align}
where
$\mathbf d_{i}$ is the \emph{disagreement signal} (so that $\mathbf c_i=\mathbf w_{i}
{+}\gamma\mathbf d_i$),
$\gamma\in(0,1/\max_{i}\ell_{ii}]$ and $\eta>0$ are (possibly, time-varying) \emph{consensus} and \emph{learning} stepsizes, and
 $\ell_{ij}$ are a set of Laplacian weights.\footnote{\eqref{ddd} maps to \ref{DGD} with $\omega_{ij}{=}-\gamma \ell_{ij}$ for $i\neq j$, and
$\omega_{ii}{=}1{-}\gamma\ell_{ii}$.
} 
These are symmetric 
($\ell_{ij}{=}\ell_{ji}{\leq}0$ for $i{\neq}j$) and add to zero at each node ($\ell_{ii}{=}{-}\sum_{j\neq i}\ell_{ij}{>}0$). The flexibility in the choice of $\ell$ will be exploited by designing a physical layer scheme that leverages the channel propagation conditions to mix the incoming signals -- a departure from prior studies that 
rely on a mesh network \cite{9224135,9563232,9517780,9322286,9562482,9772390,9716792,9838891} and
hinge on explicit design of $\ell$ or $\omega$ (see, for instance, 
\cite[Sec. III.B]{10103556} for a discussion on the design of these weights).
\begin{defi}[Laplacian matrix]
\label{lapdef}
We define the matrix $\boldsymbol{L}{\in}\mathbb R^{N\times N}$ with components $[\boldsymbol{L}]_{ij}{=}\ell_{ij}$.
$\boldsymbol{L}$
 is symmetric ($\boldsymbol{L}{=}\boldsymbol{L}^\top$) and satisfies $\boldsymbol{L}{\cdot}\mathbf 1{=}\mathbf 0$ (since $\ell_{ii}{=}{-}\sum_{j\neq i}\ell_{ij},\forall i$).
\end{defi}
  
  \begin{figure}
     \centering
         \includegraphics[width = .7\linewidth]{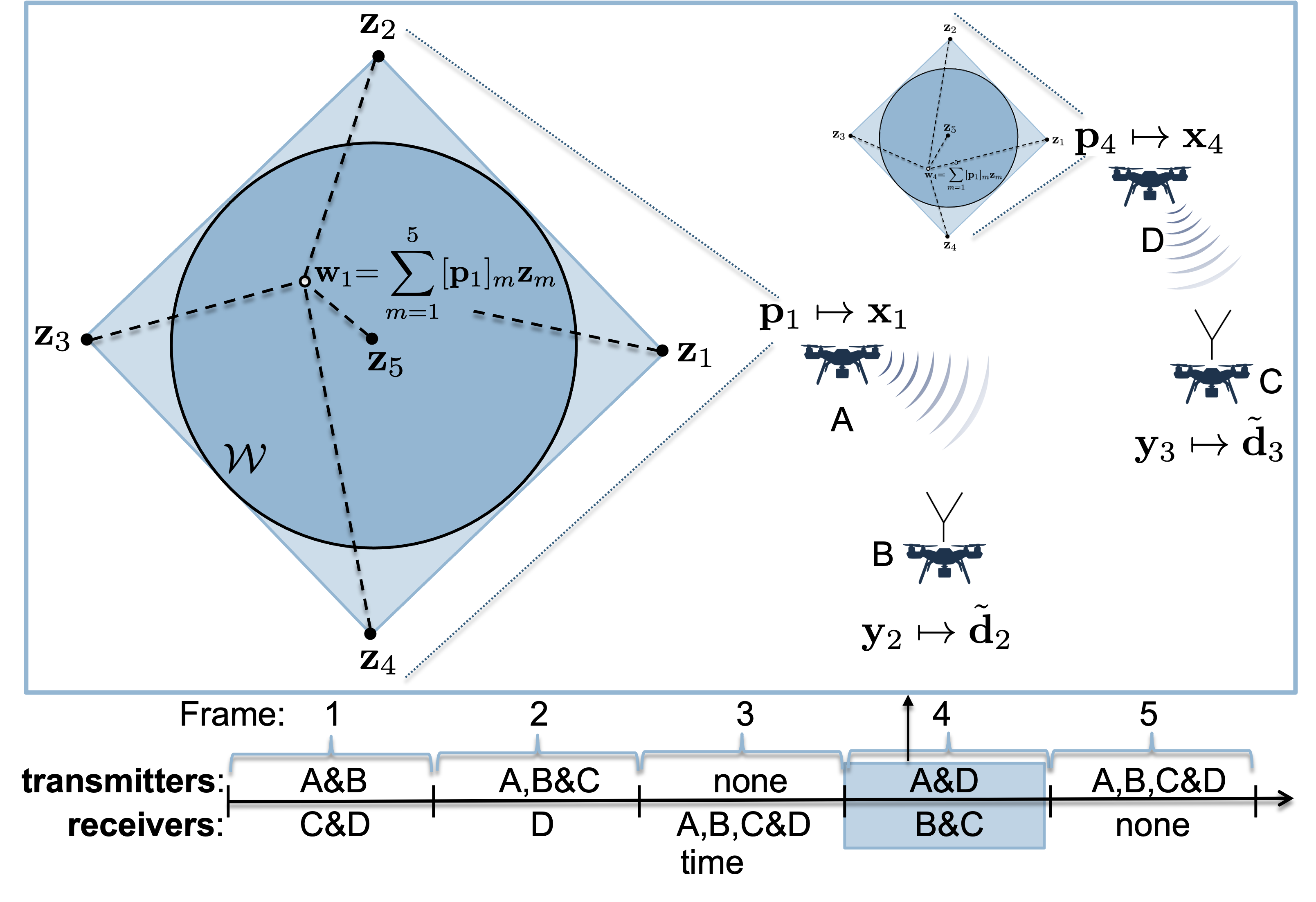}
         \vspace{-4mm}
\caption{
Example of $d{=}2$-dimensional problem with $N{=}4$ nodes.  The circle represents the set $\mathcal W$, whereas the diamond represents the convex hull of
$M{=}5$ codewords from the CP0 codebook (Example \ref{ex1}), $\mathbf z_1,\dots,\mathbf z_5$.
In frame $k=4$, nodes A and D encode their local state  $\mathbf w_i$ to a transmit signal $\mathbf x_i$ (Eq. \ref{xM1}),  
 and transmit simultaneously;
receiving nodes C and D estimate the disagreement signal $\tilde{\mathbf d}_i$ from their received signal $\mathbf y_i$
 (Eqs. \ref{yi}-\ref{dik}).\vspace{-2mm}}
\label{fig:sigmo}
\end{figure}

 We assume
 half-duplex operation, i.e., devices cannot simultaneously transmit and receive
   (Fig.~\ref{fig:sigmo}).
To accommodate this constraint, at the beginning of the frame, each node performs a random decision to either act as transmitter (denoted by the indicator variable $\chi_i{=}1$ for node $i$) or receiver ($\chi_i{=}0$) in the current frame. These decisions are i.i.d. over iterations $k$ and across nodes, and we let $p_{\mathrm{tx}}\triangleq\mathbb P(\chi_i=1)$ be the transmission probability (designed in Lemma \ref{L0}).

To achieve device-scalability, we allow the transmitting devices to transmit their local states $\mathbf w_j$ simultaneously (vs orthogonal transmissions) and we leverage the superposition property of the wireless channel to estimate $\mathbf{d}_i$. Unlike AirComp schemes developed for centralized FL that require coherent alignment of signals at the receiver, we employ \emph{non-coherent energy estimation} methods, where information is conveyed as \emph{signal energy}. As we will show, this method eliminates the need for instantaneous or average CSI at  transmitters and receivers,  but introduces errors in the estimation of $\mathbf{d}_i$,
calling for mitigation mechanisms developed in this work.

However, the local state $\mathbf w_i$
includes sign information unsuitable for energy representation.
To this end, transmitters employ an \emph{energy-based encoding procedure} to map
$\mathbf w_i$ to an "energy"
signal $\mathbf p_i\geq 0$ (\secref{enersig}).
  The latter controls energy levels across subcarriers of an OFDM symbol $\mathbf x_{i}$, then transmitted over the wireless channel (\secref{ptox}).
 Receivers, on the other hand, receive the signal $\mathbf y_i$ (\secref{rxsignal}) and
  estimate the \emph{disagreement signal} as $\tilde{\mathbf d}_i$ (\secref{disest}).
Finally, all nodes update their local state $\mathbf w_i$ 
 by combining it with $\tilde{\mathbf d}_i$, followed by  local gradient descent (\secref{optimization}).
  An example showcasing these steps is shown in Fig. \ref{fig:sigmo}.
  
 \subsection{Energy-based encoding:  $\mathbf w_i$ mapped to "energy" signal $\mathbf p_i$}
 \label{enersig}

We assume that the optimizer of \ref{global} is known to lie within a closed, convex and bounded set $\mathcal W$. 
Thus, $\mathbf w^*{\in}\mathcal W$, $\mathbf w_{i}{\in}\mathcal W,\forall i$, and \ref{global} is restricted within this set.
 For strongly-convex global functions $F(\cdot)$ with strong-convexity parameter $\mu$ (see Assumption \ref{fiassumption}),
 $\mathcal W$ may be chosen as the $d$-dimensional sphere 
  with radius
 $\frac{1}{\mu}\Vert\nabla F(\mathbf 0)\Vert$,\footnote{Strong convexity and the optimality condition $\nabla F(\mathbf w^*){=}\mathbf 0$ imply
$\Vert\nabla F(\mathbf 0)\Vert{=}\Vert\nabla F(\mathbf 0){-}\nabla F(\mathbf w^*)\Vert{\geq}\mu \Vert\mathbf 0{-}\mathbf w^*\Vert$,
hence 
$\Vert\mathbf w^*\Vert{\leq}\Vert\nabla F(\mathbf 0)\Vert/\mu$.
}
computed at initialization.

Since $\mathcal W$ is bounded, it is contained by the convex-hull defined by a finite set of codewords, as depicted in Fig. \ref{fig:sigmo}.
Concretely, let $\mathcal Z{=}\{\mathbf z_m{\in}\mathbb R^d:m=1,\dots,M\}$ be a codebook of $M$ codewords such that $\mathrm{conv}(\mathcal Z)\supseteq\mathcal W$. 
  $\mathcal Z$ is common knowledge among the nodes and  remains fixed over iterates.
Therefore, any $\mathbf w\in\mathcal W$ may be represented as a convex combination of $\mathcal Z$. In other words, there exists 
a probability vector $\mathbf p\in\mathbb R^M$ ($\mathbf 1^\top\cdot\mathbf p=1$, $\mathbf p\geq \mathbf 0$) such that
\begin{align}
\label{convcomb}
\mathbf w=\vphantom{\sum}
{\sum_{m=1}^M}
 [\mathbf p]_m\mathbf z_m.\end{align}
Given $\mathbf w_i$ and the codebook $\mathcal Z$, node $i$ then computes the associated probability vector
 $\mathbf p_i$ satisfying the convex combination condition of \eqref{convcomb}.
 Unlike $\mathbf w_i$ containing sign information, the non-negative vector $\mathbf p_i$ is suitable to encode energy levels.
 An example of this construction is provided next and depicted in Fig. \ref{fig:sigmo}, based on the cross-polytope codebook \cite{9740125}.
\begin{example}[Cross-polytope-$\phi$ (CP$\phi$) codebook]
\label{ex1}
 Let $\mathcal W\equiv\{\mathbf w:\Vert\mathbf w\Vert\leq r\}$ ($d$-dimensional sphere of radius $r$).
Consider $M{=}2d{+}1$ $d$-dimensional codewords defined as $\mathbf z_{2d+1}=\mathbf 0$,
\\\centerline{$\mathbf z_m=\sqrt{d}r\mathbf e_m,\ \mathbf z_{d+m}=-\sqrt{d}r\mathbf e_m,\ \text{for }m=1,\dots, d.$}
Then, for any $\mathbf w\in\mathcal W$, one can define the convex combination weights as
 $[\mathbf p]_{2d+1}{=}1{-}\frac{1}{\sqrt{d}r}\Vert\mathbf w\Vert_1{-}\phi$, and, for $m{=}1,\dots,d$,
\\\centerline{$[\mathbf p]_m{=}\frac{1}{\sqrt{d}r}([\mathbf w]_m)^+{+}\frac{\phi}{2d},\ [\mathbf p]_{d+m}{=}\frac{1}{\sqrt{d}r}(-[\mathbf w]_m)^+{+}\frac{\phi}{2d},$}
where $(\cdot)^+{=}\max\{\cdot,0\}$,  $\phi{\in}[0,1{-}\frac{1}{\sqrt{d}r}\Vert\mathbf w\Vert_1]$, and
$\Vert\cdot\Vert_1$ is the $\ell_1$ norm.
With these value of $\mathbf z$, $\mathbf p$ and $\phi$, it can be verified by inspection that 
$\mathbf p$ is a probability vector, i.e. $[\mathbf p]_m\geq 0,\forall m $ and $\sum_{m=1}^M[\mathbf p]_m=1$ (this step requires using $(x)^+{+}(-x)^+{=}|x|$),
and that \eqref{convcomb} holds (this step requires using $(x)^+{-}(-x)^+{=}x$).
We conclude that $\mathbf p$ defines the desired convex combination.
 The cross-polytope vector quantization of \cite{9740125} is a special case with $\phi{=}1{-}\frac{1}{\sqrt{d}r}\Vert\mathbf w\Vert_1$ and $[\mathbf p]_{2d+1}{=}0$.
\qed
\end{example}

\subsection{Transmission: $\mathbf p_i$ mapped to transmit signal $\mathbf x_i$}
\label{ptox}
To provide insights into the energy-based transmission mechanism and the subsequent disagreement signal estimation at the receiver, 
it is convenient to rewrite $\mathbf d_i$ in \eqref{ddd} in the form\footnote{Since $\ell_{ij}$ are non-positive for $i\neq j$, for convenience we absorb the sign into $-\ell_{ij}$, so that $-\ell_{ij}\geq 0$.}
\begin{align}
\label{direwritten}
\mathbf d_{i}{=}
\sum_{j\neq i}\ell_{ij}(\mathbf w_{i}{-}\mathbf w_{j}){=}\!\!
\sum_{m=1}^M\!\!\Big(\sum_{j\neq i}-\ell_{ij}[\mathbf p_j]_m\Big)(\mathbf z_m{-}\mathbf w_{i}),
\end{align}
where: in the first step, we restricted the sum to $j\neq i$, since $\ell_{ii}(\mathbf w_{i}{-}\mathbf w_{i})=0$;
in the second step, we replaced $\mathbf w_j$ with its  convex combination representation \eqref{convcomb}, 
and used the fact that $\mathbf p_j$ is a probability vector, adding to one.
This formulation suggests that, to estimate
$\mathbf d_{i}$, node $i$ can first estimate
\begin{align}\sum_{j\neq i}-\ell_{ij}[\mathbf p_j]_m,\ \forall m=1,\dots,M,\label{sump}\end{align}
since the terms $\mathbf z_m-\mathbf w_i$ are known at node $i$.
By interpreting $[\mathbf p_j]_m\geq 0$ as the energy level transmitted by device $j$ on the $m$th signal dimension, one can then interpret 
\eqref{sump} as a \emph{weighted sum} of these energies at the receiver side. Indeed, we will show that 
the wireless channel satisfies a \emph{noisy energy superposition property} similar to \eqref{sump}, allowing direct estimation of 
\eqref{sump} from received energies, with weights $-\ell_{ij}$ corresponding to the average channel gains.
  To leverage this property,
  node $i$ generates the transmit signal $\mathbf x_i$ as described next.
  
  We assume frequency-selective channels with non-negligible propagation delays between node pairs, hence the use of OFDM signaling. 
For clarity, each transmission is structured into one OFDM symbol with $M$ subcarriers (matching the codebook size $|\mathcal Z|$). 
We will extend this configuration to a generic number of subcarriers and/or OFDM symbols in \secref{M2}.
 Node $i$ then maps $\mathbf p_{i}$ to 
\begin{align}
\label{xM1}
[\mathbf x_{i}]_m=\sqrt{E\cdot M\cdot [\mathbf p_{i}]_m},\ \forall m=1,\dots, M,
\end{align}
 on the $m$th subcarrier,
with energy per sample $E$.
In other words, $[\mathbf p_{i}]_m$ controls the energy allocated on the $m$th subcarrier (up to a scaling factor $EM$).
In \eqref{xM1}, we opt for a deterministic map $\mathbf w_i{\mapsto}\mathbf p_i{\mapsto}\mathbf x_i$, in contrast to the random map
of $\mathbf w_i$ to $\mathbf z_m$ according to the probability vector $\mathbf p_i$,
 used in \cite{9815298,9740125}. In fact, this random map introduces an additional source of randomness that increases the variance of the disagreement signal estimation.
\begin{remark}
The transmitted energy per sample satisfies
$$
\frac{1}{M}\Vert\mathbf x_i\Vert^2=
\frac{1}{M}\sum_{m=1}^M|[\mathbf x_i]_m|^2
=
E\sum_{m=1}^M[\mathbf p_{i}]_m=E,
$$
since $\mathbf p_i$ is a probability vector.
Therefore, this  scheme satisfies a target power constraint, enforced via the parameter $E$.
\qed
\end{remark}

\subsection{Received signal and Noisy energy superposition property}
\label{rxsignal}
Focusing now on receiving nodes ($\chi_i=0$),
 let $\mathbf h_{ij}{\in}\mathbb C^{M}$ be the 
frequency-domain channel vector between transmitter $j$ and receiver $i$, so that $[\mathbf h_{ij}]_m$ is the channel in the $m$th subcarrier.
 In this section, we assume Rayleigh fading channels (generalized to a broad class of channel models in \secref{M2}), independent across
$i,j$ and iterations $k$, so that $[\mathbf h_{ij}]_m\sim\mathcal{CN}( 0,\Lambda_{ij}),\forall m$ with average channel gain $\Lambda_{ij}$
(${=}\Lambda_{ji}$ from channel reciprocity).
Then, node $i$ receives
\begin{align}
\label{yim}
[\mathbf y_{i}]_m=\sum_{j\neq i}\chi_j[\mathbf h_{ij}]_m[\mathbf x_{j}]_m+[\mathbf n_{i}]_m
\end{align}
on the $m$th subcarrier, where $[\mathbf n_{i}]_m$ is AWGN noise with variance $N_0$. 
In vector form,
\begin{align}
\label{yi}
\mathbf y_{i}=\sum_{j\neq i}\chi_j\mathbf h_{ij}\odot\mathbf x_{j}+\mathbf n_{i},
\end{align}
where $\mathbf n_{i}{\sim}\mathcal {CN}(\mathbf 0,N_0\mathbf I)$. 
This signal model is the discrete-time equivalent baseband representation of the communication system,
where $\mathbf x$ is the signal before the inverse DFT and the cyclic prefix appending at the transmitter,
and  $\mathbf y$ is the received signal after the DFT and cyclic prefix removal at the receiver. See, for instance, \cite[Chapter 12.4]{Goldsmith_2005}.
In the next lemma, we demonstrate the energy superposition property, 
which will be directly tied to the estimation of \eqref{sump}.
\begin{lemma}[Energy superposition property]
\label{ESP}
The energy received on the $m$th subcarrier of node $i$, $|[\mathbf y_{i}]_m|^2$, satisfies
\begin{align}
\label{Eym}
\mathbb E[|[\mathbf y_{i}]_m|^2]
{=}
E M\cdot p_{\mathrm{tx}}\sum_{j\neq i}\Lambda_{ij}[\mathbf p_{j}]_m{+}N_0,
\end{align}
 with $\mathbb E$ computed with respect to the noise, 
Rayleigh fading, and random transmission decisions of nodes $j\neq i$.\footnote{In this section, we implicitly assume that all expectations are conditional on
$\{\mathbf w_j,\forall j\}$, hence on $\{\mathbf p_j,\forall j\}$.}
\end{lemma}
\begin{proof}
Let $Y=\sum_{j} X_j$ be a sum of zero-mean ($\mathbb E[X_j]=0$) uncorrelated ($\mathbb E[X_j X_{j'}^*]=0$ for $j\neq j'$) complex-valued random variables. Then,
from the linearity of expectation,
\begin{align}
\label{uncorrsum}
\mathbb E[|Y|^2]=\mathbb E\Big[\sum_{j} |X_j|^2+\sum_{j,j'\neq j} X_j X_{j'}^*\Big]=\sum_{j}\mathbb E[|X_j|^2].
\end{align}
In \eqref{yim}, $[\mathbf y_{i}]_m$ is a sum of uncorrelated random variables, since
 channels $\mathbf h_{ij}$ and transmission decisions $\chi_j$ are independent across $j$, $[\mathbf n_{i}]_m{\sim}\mathcal {CN}(0,N_0)$ is independent.
 Furthermore, 
$\mathbb E[\chi_j[\mathbf h_{ij}]_m[\mathbf x_{j}]_m]=0$ 
and
 $\mathbb E[|\chi_j[\mathbf h_{ij}]_m[\mathbf x_{j}]_m|^2]=p_{\mathrm{tx}}\Lambda_{ij}|[\mathbf x_{j}]_m|^2$,
 since $[\mathbf h_{ij}]_m{\sim}\mathcal{CN}( 0,\Lambda_{ij})$ and $\mathbb P(\chi_j=1)=p_{\mathrm{tx}}$.
 Eq. \eqref{Eym} directly follows from \eqref{uncorrsum} and \eqref{xM1}.
\end{proof}
Lemma \ref{ESP} states that, on average, the received energy is a superposition of transmitted energies, scaled by the average channel gains, plus the noise energy.
It is thus apparent that $|[\mathbf y_{i}]_m|^2$ contains information about \eqref{sump}, in the form of an unbiased estimate of it,
with Laplacian weights corresponding to the average channel gains, $-\ell_{ij}=\Lambda_{ij},\forall i\neq j$.
This fact is exploited at the receiver to
compute an unbiased estimate of the disagreement signal, discussed next.
\subsection{Disagreement signal estimation, $\tilde{\mathbf d}_i$} 
\label{disest}
\noindent Lemma \ref{ESP} suggests that node $i$ can estimate $\sum_{j\neq i}\Lambda_{ij}[\mathbf p_{j}]_m$ as
\begin{align}
\label{fdgnhdsf}
&r_{im}
=
(1-\chi_i)\frac{|[\mathbf y_{i}]_{m}|^2-N_0}{p_{\mathrm{tx}}(1-p_{\mathrm{tx}})EM},\ \forall m=1,\dots,M.
\end{align}
Note that $r_{im}=0$ for transmitting nodes ($\chi_i=1$), consistent with the half-duplex constraint.
To see that this is an unbiased estimate of  $\sum_{j\neq i}\Lambda_{ij}[\mathbf p_{j}]_m$, we take the expectation of $r_{im}$, conditional on $\chi_i$, and use Lemma \ref{ESP},
yielding
$$
\mathbb E[r_{im}|\chi_i]
=
\frac{1-\chi_i}{1-p_{\mathrm{tx}}}\sum_{j\neq i}\Lambda_{ij}[\mathbf p_{j}]_m.
$$
Finally, we take the expectation with respect to the transmit decision $\chi_i\in\{0,1\}$ and use $\mathbb E[\chi_i]=p_{\mathrm{tx}}$, yielding
\begin{align}
\label{Erm}
\mathbb E[r_{im}]
=
\sum_{j\neq i}\Lambda_{ij}[\mathbf p_{j}]_m,\ \forall m=1,\dots,M.
\end{align}
Therefore,
$r_{im}$ is an unbiased estimate of \eqref{sump}, with Laplacian weights given by the average channel gains,
$\ell_{ij}{=}-\Lambda_{ij}$ and $\ell_{ii}{=}\sum_{j\neq i}\Lambda_{ij}{>}0$.
It is then  straightforward to generate an unbiased estimate of $\mathbf d_i$,  by simply replacing $\sum_{j\neq i}-\ell_{ij}[\mathbf p_j]_m$ in
\eqref{direwritten} with its unbiased estimate $r_{im}$, yielding
\ba{
&\tilde{\mathbf d}_{i}=
\sum_{m=1}^{M}r_{im}(\mathbf z_m-{\mathbf w}_{i}).
}{dik}
Since $\tilde{\mathbf d}_{i}$ is an unbiased estimate of $\mathbf d_i$, it follows that
\ba{
\mathbb E[\tilde{\mathbf d}_{i}]{=}
\sum_{m=1}^{M}\mathbb E[r_{im}](\mathbf z_m-{\mathbf w}_{i}){=}
\sum_{j\neq i}\Lambda_{ij}(\mathbf w_{j}{-}\mathbf w_{i})
{=}
{\mathbf d}_{i}.\!\!
}{Ed}
However, since this property holds only in expectation,
  $\tilde{\mathbf d}_{i}$ exhibits zero-mean deviations around its expected value,
due to noise, fading, random transmission decisions, and energy fluctuations in the wireless channel.
 The analysis of how these estimation errors impact convergence and are mitigated via a suitable stepsize design is conducted in \secref{convanalysis}.
 
The $\ell$ values are indeed Laplacian weights, since they satisfy $\ell_{ij}{=}-\Lambda_{ij}{\leq}0$ for $i{\neq}j$, $\ell_{ij}{=}\ell_{ji}$ (average channel reciprocity, $\Lambda_{ij}=\Lambda_{ji}$), and $\ell_{ii}{=}\sum_{j\neq i}\Lambda_{ij}{>}0$.
Accordingly, we define the Laplacian matrix $\boldsymbol{L}\in\mathbb R^{N\times N}$ as in Definition \ref{lapdef},
induced by the average channel gains $\Lambda_{ij}$.


\subsection{Local optimization state update}
 \label{optimization}
During the frame, node $i$ computes an unbiased stochastic gradient
$\mathbf g_{i}$, with $\mathbb E[\mathbf g_{i}|\mathbf w_i]{=}\nabla f_i(\mathbf w_{i})$.
It then updates $\mathbf w_{i}$ as
\begin{align}
\label{updateeq}
\mathbf w_{i}\gets\Pi[\mathbf w_{i}+\gamma\tilde{\mathbf d}_{i}-\eta\mathbf g_{i}],
\end{align}
where  
 $\Pi[\mathbf a]$ is a projection operator,
$$\Pi[\mathbf a]=\argmin_{\mathbf w\in\mathcal W}\Vert\mathbf w-\mathbf a\Vert,$$
restricting the algorithm within the set $\mathcal W$, so that  $\mathbf w_{i}\in\mathcal W,\forall i$.
$\gamma,\eta>0$ are (possibly, time-varying) \emph{consensus} and \emph{learning} stepsizes, respectively,:
the former controls information diffusion and mitigates signal fluctuations in the disagreement signal estimation (used also in
\cite{9563232} to mitigate wireless channel impairments and in 
\cite{8786146} to mitigate quantization errors);
the latter regulates the magnitude of gradient steps, hence the learning progress.
As detailed in \secref{convanalysis}, their design is crucial to balancing these competing objectives.



These steps   are repeated in frame $k{+}1$ with the new local state $\mathbf w_{i}$, and so on.
  Note the key differences with respect to the conventional formulation
\eqref{ddd}: 1) the disagreement signal $\mathbf d_i$ and gradient $\nabla f_i(\mathbf x_i)$ are replaced with unbiased estimates, $\tilde{\mathbf d}_i$ and $\mathbf g_i$, respectively; 2) the projection operator $\Pi$.
  Overall,  \proposed\ expressed in \eqref{updateeq} can thus be interpreted as a 
projected DGD with noisy consensus and noisy gradients,
whose convergence properties are studied in \secref{convanalysis}.
\section{Generalizations}
\label{M2}
\subsection{More general frame structure}

More generally, a frame is constituted of $O$ OFDM symbols, each with $\mathrm{SC}$ subcarriers,
defining $Q{=}O{\times}\mathrm{SC}{\geq}M$ resource units.
An example is depicted in Fig. \ref{fig:ofdmframe}.
Let $q{\in}\{1,{\dots},Q\}{\equiv}\mathcal Q$ be the $q$th resource unit, on a certain OFDM symbol and subcarrier.
Unlike the previous section where $Q=M$, hence each signal dimension could only be mapped to a single subcarrier,
we now have $Q\geq M$, which allows to map each component  of the vector $\mathbf p_{i}$ to multiple resources units.
We thus define a partition of $\mathcal Q$ into $M$ non-empty sets, $\mathcal R_1,\dots,\mathcal R_M$, with $\mathcal R_m\cap\mathcal R_{m'}\equiv\emptyset,\forall m\neq m'$ and $\cup_{m=1}^M\mathcal R_m\equiv\mathcal Q$.
For instance, in Fig. \ref{fig:ofdmframe},
$\mathcal R_1\equiv\{1,6,11,16,21\}$,
$\mathcal R_2\equiv\{2,7,12,17,22\}$,
$\mathcal R_3\equiv\{3,8,13,18,23\}$,
$\mathcal R_4\equiv\{4,9,14,19,24\}$,
containing 5 resource units, and
$\mathcal R_5\equiv\{5,10,15,20\}$ containing 4.
We then map the $m$th component $[\mathbf p_{i}]_m$ to the resource units $\mathcal R_m$ by repetition.
In vector notation, we define
\begin{align}
\label{preambles}
\mathbf u_m\triangleq \sqrt{\frac{Q}{R_m}}\sum_{q\in\mathcal R_m}\mathbf e_{q}\in\mathbb C^Q,
\end{align}
with $R_m\triangleq |\mathcal R_m|$,
 as the vector describing the resource units that are activated in the set $\mathcal R_m$,
 with $\Vert\mathbf u_m\Vert=\sqrt{Q}$.
 In other words, $\mathbf u_m$ has components equal to $\sqrt{Q/R_m}$ in the indices in the set $\mathcal R_m$, 
and equal to zero otherwise. Note also that
$\{\mathbf u_m\}$ are orthogonal to each other, $\mathbf u_m^{\mathrm H}\mathbf u_{m'}=0$ for $m\neq m'$, since $\mathcal R_m\cap\mathcal R_{m'}\equiv \emptyset$.
 With this definition, the allocation of $[\mathbf p_{i}]_m$ to its associated resource units by repetition can be simply 
described as $\sqrt{[\mathbf p_{i}]_m}\mathbf u_m$, so that $[\mathbf p_{i}]_m$ scales the energy allocated to $\mathbf u_m$, and
 the transmit signal is defined as
$$
\mathbf x_{i}=\sqrt{E}\cdot
\sum_{m=1}^{M}\sqrt{[\mathbf p_{i}]_m}\cdot\mathbf u_{m}.
$$
With $\mathbf x_i$ thus defined, the  signal received
by node $i$ is given as in \eqref{yi}.
Node $i$ then computes the $M$ energy signals
\begin{align}
\label{fdgnhdsf2}
&r_{im}
=
(1-\chi_i)\sum_{q\in\mathcal R_{m}}
\frac{|[\mathbf y_{i}]_{q}|^2-N_0}{p_{\mathrm{tx}}(1-p_{\mathrm{tx}})EQ}, 
\end{align}
yielding \eqref{xM1}-\eqref{fdgnhdsf} as a special case with $Q{=}M$ and $\mathcal R_m{\equiv}\{m\}$.
In \eqref{fdgnhdsf2}, energies are aggregated across
the resource units $\mathcal R_{m}$ allocated on the $m$th component, rather than a single subcarrier in \eqref{fdgnhdsf}.
To show that $r_{im}$ is an unbiased estimate of \eqref{sump}, note that $[\mathbf x_{j}]_q=\sqrt{EQ/R_m}\cdot\sqrt{[\mathbf p_{j}]_m}$ and 
\begin{align*}
[\mathbf y_{i}]_q{=}\sum_{j\neq i}\chi_j[\mathbf h_{ij}]_q[\mathbf x_{j}]_q+[\mathbf n_{i}]_q,\ \forall q\in\mathcal R_m.
\end{align*}
Then, we use similar steps as in Lemma \ref{ESP} to find
\begin{align}
\label{yiprexp}
\mathbb E[|[\mathbf y_{i}]_q|^2]{=}\frac{EQ}{R_m} p_{\mathrm{tx}}\sum_{j\neq i}\Lambda_{ij}[\mathbf p_{j}]_m+N_0,\ \forall q\in\mathcal R_m,
\end{align}
so that $\mathbb E[r_{im}]=\sum_{j\neq i}\Lambda_{ij}[\mathbf p_{j}]_m$ as in \eqref{Erm}, after 
replacing \eqref{yiprexp} into the expectation of \eqref{fdgnhdsf2}, and computing the expectation with respect to $\chi_i$.
Therefore, \eqref{Ed} still holds.

\begin{figure}
     \centering
         \includegraphics[width = \linewidth]{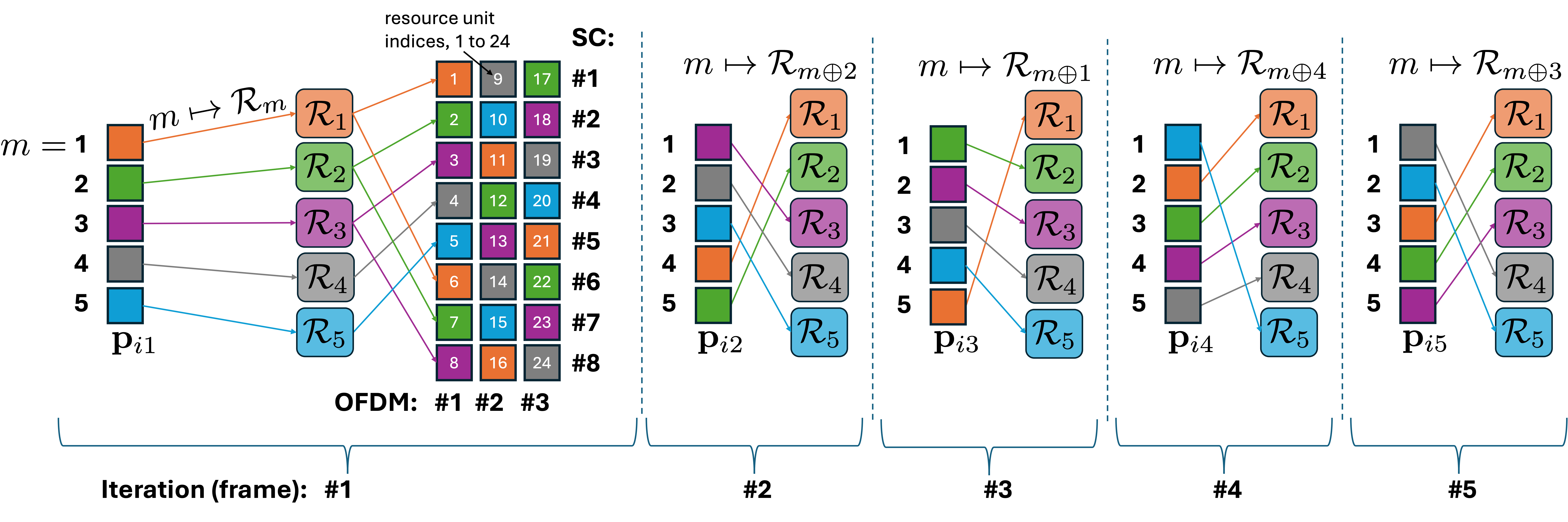}
         \vspace{-7mm}
\caption{
Example of $M{=}5$-dimensional $\mathbf p$ mapped to a frame containing $O=3$ OFDM symbols, each with $\mathrm{SC}=8$ subcarriers, over 5 iterations.
The color-coded sets $\mathcal R_{1},\dots,\mathcal R_5$ correspond to a partition of the $Q=3\times 8=24$ resource units (across subcarriers and OFDM symbols)
allocated to a certain signal dimension. For instance, at iteration $1$, $m=1$ is allocated to $\mathcal R_1$, containing the 
 resource units $\{1,6,11,16,21\}$.
Frames \#1 to \#5 demonstrate the circular subcarrier shift (\secref{broadclass}). In this example, each signal dimension is mapped to all resource units across 5 iterations.
\vspace{-2mm}}
\label{fig:ofdmframe}
\end{figure}

\subsection{Broad class of channel models}
\label{broadclass}
 The unbiasedness of $\tilde{\mathbf d}_{i}$  in \eqref{Ed} is key to proving convergence in \secref{convanalysis},
 but it may not hold under more general channel models than those considered in \secref{sysmo}. For instance, if $\mathbf h_{ij}$ remains fixed over iterations $k$ (static channels),
\begin{align}
\label{biasissue}
&\mathbb E[r_{im}]=\sum_{j\neq i}\frac{1}{R_m}\sum_{q\in\mathcal R_m}|[\mathbf h_{ij}]_q|^2[\mathbf p_{j}]_m\\
&+2p_{\mathrm{tx}}\sum_{j,j'\neq i:j'<j}\frac{1}{R_m}\sum_{q\in\mathcal R_m}\mathrm{re}([\mathbf h_{ij}]_q [\mathbf h_{ij'}]_q^*)\sqrt{[\mathbf p_{j}]_m[\mathbf p_{j'}]_m},
\nonumber
\end{align}
so that \eqref{Erm} no longer holds, violating the unbiasedness of the estimate $\tilde{\mathbf d}_i$.
Compared to \eqref{Erm}: 1) each component of $\mathbf p_{j}$ experiences a different average channel gain, $\frac{1}{R_m}\sum_{q\in\mathcal R_m}|[\mathbf h_{ij}]_q|^2$ for the $m$th component, instead of a common $\Lambda_{ij}$; and 2) the cross-product terms
$\mathrm{re}([\mathbf h_{ij}]_q [\mathbf h_{ij'}]_q^*)$ do not vanish.
A bias term thus appears, that accumulates over time and leads to loss of the convergence properties.
To recover the unbiasedness of $\tilde{\mathbf d}_{i}$, we introduce
two mechanisms under a broad class of channels, including Rayleigh fading and static as special cases.
\begin{assumption}
\label{ch2}
The channels $\mathbf h_{ij}$ are i.i.d. over iterations $k$.\footnote{Static channels can be interpreted as a degenerate i.i.d. process in which
   $\mathbf h_{ij}$ takes a deterministic value, with probability 1.}
We define the average gain across $\mathcal R_m$  as
\\
\centerline{$\Lambda_{ij}^{(m)}\triangleq\mathbb E\Big[\frac{1}{R_m}\sum_{q\in\mathcal R_m}|[\mathbf h_{ij}]_{q}|^2\Big]$,}
\\
and further across $m$ as
  $\Lambda_{ij}{\triangleq}\frac{1}{M}\sum_{m=1}^{M}\Lambda_{ij}^{(m)}$. 
\end{assumption}

\noindent\emph{Random phase shift}: each transmitter applies a random phase shift to the transmit signal, uniform
in $[0,2\pi]$, i.i.d. over time and across nodes, and i.i.d. across resource units within the same set $\mathcal R_m$.
Let $[\bstheta_j]_q$ be the phase shift applied to the $q$th resource unit by node $j$.
Then, the combined channel  $\tilde{\mathbf h}_{ij}{=}e^{\mathrm{j}\bstheta_j}{\odot}{\mathbf h}_{ij}$ between transmitter $j$ and receiver $i$
($[\tilde{\mathbf h}_{ij}]_q{=}e^{\mathrm{j}[\bstheta_j]_q}{\cdot}[{\mathbf h}_{ij}]_q$ on the $q$th resource unit)
 becomes
circularly symmetric,
yielding $\mathbb E[\mathrm{re}([\tilde{\mathbf h}_{ij}]_q [\tilde{\mathbf h}_{ij'}]_q^*)]=0,\forall j\neq j'$, upon taking the expectation on $[\bstheta_j]_q$.
Therefore, the cross product terms in \eqref{biasissue} vanish under a random phase shift.
\\\noindent\emph{Coordinated circular subcarrier shift}: each transmitter applies a shift $\varsigma$ to the mapping of signal components to resource units,
i.e., $[\mathbf p_{i}]_m$ is mapped to the resource units in the set  $\mathcal R_{m\oplus \varsigma}$;\footnote{Here, $m\oplus\varsigma{=}1{+}\mathrm{mod}(m{-}1{+}\varsigma,M)$ is the sum modulo $M$.}
 $\varsigma$ is chosen uniformly at random in $\mathcal M{\equiv}\{0,{\dots}, M{-}1\}$, i.i.d. over time, and identical across devices (e.g., a pseudo-random sequence with a common seed).
 For instance, in Fig. \ref{fig:ofdmframe}, $m$ is mapped to $\mathcal R_{m}$ at iteration 1 ($\varsigma{=}0$),
 to $\mathcal R_{m\oplus 2}$ at iteration 2 ($\varsigma{=}2$), to $\mathcal R_{m\oplus 1}$ at iteration 3 ($\varsigma{=}1$), 
 to $\mathcal R_{m\oplus 4}$ at iteration 4 ($\varsigma{=}4$),
 and to $\mathcal R_{m\oplus 3}$ at iteration 5 ($\varsigma{=}3$).
Thanks to this mechanism, each signal component goes through all resource units over multiple iterations, and experiences the same channel gain $\Lambda_{ij}$, on average.

Combining these two mechanisms, 
the transmit signal is
\begin{align}
\label{xi}
\mathbf x_{i}=\sqrt{E}\cdot
e^{\mathrm{j}\bstheta_i}\odot
\sum_{m=1}^{M}\sqrt{[\mathbf p_{i}]_m}\cdot\mathbf u_{m\oplus\varsigma}.
\end{align}
Upon receiving $\mathbf y_{i}$ as in \eqref{yi}, node $i$ aggregates the energies received on $\mathcal R_{m\oplus \varsigma}$
(the resource units allocated to $m$, accounting for the shift $\varsigma$)
 to compute
$r_{im}$, hence we modify \eqref{fdgnhdsf2}~as
\begin{align}
\label{fdgnhdsf3}
&
r_{im}
=
(1-\chi_i)\sum_{q\in\mathcal R_{m\oplus \varsigma}}
\frac{|[\mathbf y_{i}]_{q}|^2-N_0}{p_{\mathrm{tx}}(1-p_{\mathrm{tx}})EQ}. 
\end{align}
$\tilde{\mathbf d}_i$ is then computed as in  \eqref{dik},
and the local model update follows \eqref{updateeq}.
To show that $r_{im}$ yields an unbiased estimate of \eqref{sump}, note that in the $q$th resource unit, $q\in\mathcal R_{m\oplus\varsigma}$, we have
\begin{align*}
[\mathbf y_{i}]_q{=}\sum_{j\neq i}\chi_j[\mathbf h_{ij}]_q[\mathbf x_{j}]_q{+}[\mathbf n_{i}]_q,\ \forall q\in\mathcal R_{m\oplus\varsigma},
\end{align*}
with $[\mathbf x_{j}]_q=e^{\mathrm{j}[\bstheta_j]_q}\sqrt{EQ/R_{m\oplus\varsigma}}\sqrt{[\mathbf p_{j}]_m}$.
Then, we compute the expectation of $|[\mathbf y_{i}]_q|^2$
 with respect to the phase shifts, random transmissions of $j\neq i$, noise and channels, conditional on the subcarrier shift $\varsigma$. 
Thanks to the  random phase shift $[\bstheta_j]_q$, i.i.d. across $j$,
 $\chi_j[\mathbf h_{ij}]_q[\mathbf x_{j}]_q$ are zero mean and uncorrelated across $j$.
We can then use similar steps as in Lemma \ref{ESP} to find
$$
\mathbb E[|[\mathbf y_{i}]_q|^2|\varsigma]{=}\frac{EQp_{\mathrm{tx}}}{R_{m\oplus\varsigma}}\sum_{j\neq i}\mathbb E[|[\mathbf h_{ij}]_q|^2][\mathbf p_{j}]_m{+}N_0,
\forall q{\in}\mathcal R_{m\oplus\varsigma}.
$$
Using this expression in the expectation of \eqref{fdgnhdsf3}, we then find
$$
\mathbb E[r_{im}|\chi_i,\varsigma]{=}
\frac{1-\chi_i}{1-p_{\mathrm{tx}}}
\sum_{j\neq i}\frac{1}{R_{m\oplus\varsigma}}\sum_{q\in\mathcal R_{m\oplus \varsigma}}\mathbb E[|[\mathbf h_{ij}]_q|^2][\mathbf p_{j}]_m
$$$$
=
\frac{1-\chi_i}{1-p_{\mathrm{tx}}}\sum_{j\neq i}\Lambda_{ij}^{(m\oplus\varsigma)}[\mathbf p_{j}]_m,
$$
with $\Lambda_{ij}^{(m')}$ defined in  Assumption \ref{ch2}. Therefore,
 the cross-product term in \eqref{biasissue} vanishes, thanks to the circular symmetry of the equivalent channels $\tilde{\mathbf h}_{ij}{=}e^{\mathrm{j}\bstheta_j}{\odot}{\mathbf h}_{ij}$.
Finally, we compute the expectation with respect to $\chi_i$ and $\varsigma\sim\mathrm{uniform}(\mathcal M)$ (hence, $m'=m\oplus \varsigma$ is uniform in $\{1,2,\dots,M\}$), yielding
$$
\mathbb E[r_{im}]
=
\frac{1}{M}\sum_{m'=1}^{M}\sum_{j\neq i}\Lambda_{ij}^{(m')}[\mathbf p_{j}]_m
=
\sum_{j\neq i}\Lambda_{ij}[\mathbf p_{j}]_m,
$$
with $\Lambda_{ij}$ defined in  Assumption \ref{ch2}.
We conclude that $r_{im}$ satisfies \eqref{Erm},
so that $\tilde{\mathbf d}_i$ is an unbiased estimate of $\mathbf d_i$, 
as in \eqref{Ed}, and the convergence analysis of \secref{convanalysis} readily applies.

Therein, the convergence
is dictated by the variance of the estimation error $\tilde{\mathbf d}_i{-}\mathbb E[\tilde{\mathbf d}_i]$, bounded
in the next lemma.
It shows the impact of:
the variation of $|[\mathbf h_{ij}]_{q}|^2$ around the average channel gain (averaged across resource units) $\Lambda_{ij}$, 
 captured by the term $\vartheta$;
the variation of $\hat\lambda_{ij}^{(m)}$ (sample average channel gain across $\mathcal R_{m}$) around $\Lambda_{ij}$, captured by the term $\varpi$;
 the maximum sum gain ($\Lambda^*$).
 The lemma also provides  a variance-minimizing design of the transmission probability $p_{\mathrm{tx}}$ (proof in the  supplemental document).
\begin{lemma}
\label{L0}
Consider channels satisfying 
Assumption \ref{ch2}. Assume that the resource units are evenly allocated among $m=1,\dots, M$, i.e., $|R_m{-}Q/M|{<}1,\forall m$ (e.g., Fig. \ref{fig:ofdmframe}). 
Let:
\begin{align}
\label{vartheta}
&
\vartheta\triangleq\max_{i,j\neq i}\frac{1}{\Lambda_{ij}}\Big(\frac{1}{Q}\sum_{q=1}^Q\mathbb E[(|[\mathbf h_{ij}]_{q}|^2-\Lambda_{ij})^2]\Big)^{1/2},\\&
\label{varsigma}
\varpi\triangleq\max_{i,j\neq i}
\frac{1}{\Lambda_{ij}}\Big(\frac{1}{M}\sum_{m=1}^M\mathbb E[(\hat\lambda_{ij}^{(m)}-\Lambda_{ij})^2]\Big)^{1/2},
\\
\label{Lambda}
&\Lambda^*\triangleq\max_{i}\sum_{j\neq i}\Lambda_{ij},
\end{align}
where $\hat\lambda_{ij}^{(m)}=\frac{1}{{R_{m}}}\sum_{q\in\mathcal R_{m}}|[{\mathbf h}_{ij}]_{q}|^2$
is the sample average channel gain across the resource units in the set $\mathcal R_{m}$.
Then, $\frac{1}{N}\sum_{i=1}^N\mathrm{var}(\tilde{\mathbf d}_{i}){\leq}\Sigma^{(1)}
\triangleq\max\limits_{\mathbf z,\mathbf z'\in\mathcal Z}\Vert\mathbf z-\mathbf z'\Vert^2$
\\\centerline{$
\times \frac{1}{1-p_{\mathrm{tx}}}\Big[
 \frac{\sqrt{M}}{\sqrt{Q}}\sqrt{2(1+2\vartheta^2)}\Lambda^*
+\frac{\sqrt{1+\varpi^2}}{\sqrt{p_{\mathrm{tx}}}}\Lambda^*
+ \frac{\sqrt{M}}{\sqrt{Q}}\frac{N_0}{Ep_{\mathrm{tx}}}
\Big]^2.
$}
The value of the transmission probability $p_{\mathrm{tx}}$ minimizing this bound is the unique solution in $(0,1)$ of
\\\centerline{$
\sqrt{2(1{+}2\vartheta^2)}p_{\mathrm{tx}}^{3/2}{+}
\frac{\sqrt{Q}}{\sqrt{M}}\sqrt{1{+}\varpi^2}(2p_{\mathrm{tx}}{-}1){+}
\frac{N_0}{\Lambda^* E}\frac{3p_{\mathrm{tx}}{-}2}{\sqrt{p_{\mathrm{tx}}}}
=0.
$}
\end{lemma}
 For Rayleigh fading channels discussed in \secref{sysmo},
 and assuming the channel gains are independent across $q{\in}\mathcal R_{m}$ (typically achieved by spacing apart  
 the resource units belonging to the same $\mathcal R_m$ by more than the channel's coherence bandwidth)
  we find $\vartheta{=}1$ and $\varpi{\approx}\sqrt{M/Q}$.
Algorithm~\ref{A1} outlines the steps of \proposed\ described in this section.
Each OFDM symbol, with a communication bandwidth $W_{\mathrm{tot}}$, has a duration $T_{\mathrm{ofdm}} = (\text{SC} + \text{CP})/W_{\mathrm{tot}}$ ($\text{SC}$ subcarriers; cyclic prefix length $\text{CP}$). With $O$ OFDM symbols required to transmit the $Q = O \cdot \text{SC}$-dimensional signal $\mathbf{x}_i$, the frame duration is
\ba{
T=O\cdot T_{\mathrm{ofdm}}=W_{\mathrm{tot}}^{-1}\cdot O\cdot(\mathrm{SC}+\mathrm{CP}),
}{frameduration}
irrespective of network size $N$. This highlights the device-scalability of \proposed\ over large wireless networks, unlike a TDMA scheme where $T$ grows linearly with $N$. 

\begin{algorithm} [t]
\caption{\proposed\ iteration}\label{A1}
    \begin{algorithmic}[1]
        \small
         \Procedure{}{at node $i$, given $\mathbf w_{i}$, circular shift $\varsigma$, stepsizes $\eta,\gamma$}
        \State \underline{\bf Computation} (in parallel with communication): compute an unbiased stochastic gradient $\mathbf g_{i}$ based on the local $f_i$ at $\mathbf w_i$;
         \State \underline{\bf Random transmission decision} select $\chi_i=1$ with probability $p_{\mathrm{tx}}$, $\chi_i=0$ otherwise;
        \State \underline{\bf Communication} {\bf -- if transmission ($\chi_i=1$):}
         \State $\rightarrow$ Compute the convex combination $\mathbf p_{i}$ (see Example \ref{ex1});
         \State $\rightarrow$ Generate random phase $\bstheta_i$ and transmit signal $\mathbf x_{i}$ via \eqref{xi}; 
         \State $\rightarrow$ Map  $\mathbf x_{i}$ to
         $O$ OFDM symbols, each with $\mathrm{SC}$ subcarriers;
         \State $\rightarrow$ Set $\tilde {\mathbf d}_{i}=\mathbf 0$; 
          \State \underline{\bf Communication} {\bf -- if reception ($\chi_i=0$):}
         \State $\rightarrow$ Receive $\mathbf y_{i}$ (see \eqref{yi});
         \State $\rightarrow$ Compute the $M$ energy signals
via \eqref{fdgnhdsf3}; 
\State $\rightarrow$ Estimate the disagreement signal $\tilde {\mathbf d}_{i}$ via \eqref{dik}; 
\State  \underline{\bf DGD update:}
$
\mathbf w_{i}\gets\Pi[\mathbf w_{i}+\gamma\tilde {\mathbf d}_{i}-\eta\mathbf g_i]
$ as in \eqref{updateeq}.
\EndProcedure\ repeat in the next iteration.
    \end{algorithmic}
\end{algorithm}

\section{Convergence Analysis}
\label{convanalysis}
In this section, we study the convergence of Algorithm \ref{A1} as it progresses over multiple iterations. Hence, we express the dependence of the
signals and stepsizes on the iteration index $k$.
Due to randomness of
noise, fading, random transmissions, phase and  circular subcarrier shifts at transmitters, Algorithm~\ref{A1}
 induces a stochastic process.
 We denote by $\mathcal F_k$ the $\sigma$-algebra generated by
 the signals up to frame $k-1$,
 including $\mathbf w_{ik},\ \forall i$.
In essence, when taking expectations conditioned on $\mathcal F_k$, we consider the randomness generated during the $k$th iteration, conditional on the signals available at the start of that iteration.

To analyze \proposed,
we first stack its updates over the network, and make the error terms explicit (\secref{equivrep}).
We then introduce  assumptions and definitions used in the analysis, commonly adopted in prior work (\secref{assdef}).
We present the main convergence result in \secref{mainconv},  then specialize it to constant  (\secref{conststep}).
and decreasing  (\secref{decstep}) stepsizes.

\subsection{Equivalent representation of \proposed}
\label{equivrep}
We use lowercase variables for node-specific signals ($\mathbf{a}_{i}$ for node $i$) and uppercase variables for their concatenation over the network ($Nd$-dimensional vector $\mathbf{A}=[\mathbf{a}_1^\top,\mathbf{a}_2^\top\dots,\mathbf{a}_N^\top]^\top$). Thus, $\mathbf{W}_k$, $\mathbf{G}_k$, and $\tilde{\mathbf{D}}_{k}$ concatenate the $\mathbf{w}_{ik}$, $\mathbf{g}_{ik}$, and $\tilde{\mathbf{d}}_{ik}$ signals at iteration $k$. Additionally, we define $\mathbf{W}^*=\mathbf{1}_N\otimes\mathbf{w}^*$ as the concatenation of $\mathbf w^*{=}\smash{\argmin_{\mathbf w\in\mathbb R^d}}
F(\mathbf{w})$. Let
\begin{align}
f(\mathbf W)=\sum_{i=1}^Nf_i(\mathbf w_i).
\label{globalf}
\end{align}
Noting that $\nabla f(\mathbf W)=[\nabla f_1(\mathbf w_1)^\top, f_2(\mathbf w_2)^\top,\dots,f_N(\mathbf w_N)^\top]^\top$,
we can then rewrite the update \eqref{updateeq} as 
\begin{align}\mathbf W_{k+1}=\Pi^N[\mathbf W_{k}+\gamma_k\tilde{\mathbf D}_{k}-\eta_k\mathbf G_{k}],\label{stackedeq}\end{align}
where $\mathbb E[\mathbf G_{k}|\mathcal F_k]{=}\nabla f(\mathbf W_k)$ (unbiased gradient estimate)
and $\Pi^N[\cdot]$ is a projection onto $\mathcal W^N$.
To facilitate the analysis, we define the errors
$\e_k^{(1)}$
in the disagreement signal estimation, due to fading, noise, random transmissions, phase and circular subcarrier shifts, and
$\e_k^{(2)}$ in the stochastic gradients. Namely,
$$
\e_k^{(1)}=\tilde{\mathbf D}_{k}
-\mathbb E[\tilde{\mathbf D}_{k}|\mathcal F_k],\ \ \ 
\e_k^{(2)}=\mathbf G_{k}-\nabla f(\mathbf W_{k}).
$$ 
By definition, $\mathbb E[\e_k^{(1)}|\mathcal F_k]{=}\mathbb E[\e_k^{(2)}|\mathcal F_k]{=}\mathbf 0$.
Letting $\LL=\boldsymbol{L}\otimes\mathbf I_d$,
we then rewrite \eqref{stackedeq} in the equivalent form $\mathbf W_{k+1}$ 
\ba{
{=}\Pi^N\Big[&\underbrace{(\mathbf I{-}\gamma_k\LL)\mathbf W_{k}{-}\eta_k\nabla f(\mathbf W_{k})}_{(a)}
+\underbrace{\gamma_k\e_k^{(1)}{-}\eta_k\e_k^{(2)}}_{(b)}\Big],
}{Wupdate}
where 
$\mathbb E[\tilde{\mathbf D}_k|\mathcal F_k]{=}{-}\LL\mathbf W_{k}$ is obtained by concatenating
$\mathbb E[\tilde{\mathbf d}_{ik}|\mathcal F_k]$, given by \eqref{Ed}, and using the fact that
$${\mathbf d}_{i}{=}\sum_{j=1}^N\ell_{ij}(\mathbf w_{i}-{\mathbf w}_{j})
{=}
\sum_{j=1}^N\ell_{ij}\mathbf w_{i}-\sum_{j=1}^N\ell_{ij}{\mathbf w}_{j}
{=}-\sum_{j=1}^N[\boldsymbol{L}]_{ij}{\mathbf w}_{j}
$$
since $\sum_{j=1}^N\ell_{ij}{=}0$ for Laplacian weights.
The term (a) in \eqref{Wupdate} specializes to the DGD update shown in \eqref{ddd} when  $\gamma_k{=}\gamma,\forall k$; the term (b) (referred to as "DGD noise" in the remainder), instead, captures the various sources of randomness.
Following an approach similar to \cite{Yuan2016} for the analysis of DGD (therein, with \emph{fixed} stepsize),
we interpret the term (a) as a gradient descent, with stepsize $\eta_k$,  based on the Lyapunov function
$$G_k(\mathbf W)\triangleq f(\mathbf W)+\frac{\gamma_k}{2\eta_k}{\mathbf W}^\top\cdot\LL\cdot{\mathbf W}.$$
Notably,  the quadratic term in $G_k$ enforces consensus in the network (in fact, it equals zero when $\mathbf w_i=\mathbf w_j,\forall i,j$).
Note that $G_k$ is time-varying due to $k$-dependent stepsizes, thus generalizing \cite{Yuan2016,8786146}
 that employ a fixed stepsize (hence a time-invariant $G$).
We can then express \eqref{Wupdate} compactly as
\ba{
\mathbf W_{k+1}=\Pi^N[\mathbf W_{k}
-\eta_k\nabla G_k(\mathbf W_{k})
+\gamma_k\e_k^{(1)}-\eta_k\e_k^{(2)}
].
}{Lyapnoisy}

\subsection{Assumptions and Definitions}
\label{assdef}
We study the convergence of \eqref{Lyapnoisy} under standard assumptions,
commonly adopted in prior work, e.g., \cite{9563232,9517780,8786146}.
\begin{assumption}
\label{fiassumption}
All $f_i(\mathbf w)$ are  $\mu$-strongly convex, $L$-smooth.
\end{assumption}
A direct consequence is that
 the global $F$  in \ref{global} and $f$ in \eqref{globalf} are also $\mu$-strongly convex and $L$-smooth.
\begin{assumption}
\label{distance}
$\mathbf w^*\in\mathrm{int}(\mathcal W)$, hence its distance from the boundary of $\mathcal W$ satisfies
$
\zeta\triangleq \min_{\mathbf w\in\mathrm{bd}(\mathcal W)}\Vert\mathbf w-\mathbf w^*\Vert>0.
$\footnote{The analysis may be relaxed to the case when $\mathbf w^*$ lies on the boundary of $\mathcal W$ ($\zeta=0$),
but it yields looser convergence results.}
\end{assumption}
 $\boldsymbol{L}$ is a Laplacian matrix (Definition \ref{lapdef}), hence
 it is semidefinite positive with eigenvalues 
$0{=}\rho_1{\leq}\dots\leq \rho_N$.
We  assume that the average channel gains define a connected graph, 
so that the  algebraic connectivity $\rho_2$ is strictly positive~\cite{10.5555/2613412}.
 \begin{assumption}[Algebraic connectivity of $\boldsymbol{L}$] 
 $
\rho_2>0.
$
\label{algconnofL}
\end{assumption}

Finally,  we make the following assumption on $\e_k^{(1)}$ and $\e_k^{(2)}$.
\begin{assumption}
\label{SGDassumption}
There exist $\Sigma^{(1)},\Sigma^{(2)}\geq 0$ such that
$$
\frac{1}{N}\mathbb E[\Vert\e_k^{(1)}\Vert^2|\mathcal F_k]\leq \Sigma^{(1)},\ \frac{1}{N}\mathbb E[\Vert\e_k^{(2)}\Vert^2|\mathcal F_k]\leq \Sigma^{(2)}.
$$
Closed-form expressions are given in Lemma \ref{L0} for $\Sigma^{(1)}$,\footnote{Note that
$\frac{1}{N}\mathbb E[\Vert\e_k^{(1)}\Vert^2|\mathcal F_k]=
\frac{1}{N}\sum_{i=1}^N\mathrm{var}(\tilde{\mathbf d}_{i}|\mathcal F_k)$.}
and in Appendix A for $\Sigma^{(2)}$, under minibatch gradients. 
\end{assumption}

\noindent Since $\mathcal W$ is bounded, its diameter, defined below, is finite.
\begin{defi}[diameter of $\mathcal W$]
\label{diameter}
$
\mathrm{dm}(\mathcal W)\triangleq\max\limits_{\mathbf w,\mathbf w'\in\mathcal W}\Vert\mathbf w-\mathbf w'\Vert.
$
\end{defi}
\noindent 
In Example \ref{ex1}, $\mathrm{dm}(\mathcal W){=}2r$.
Lastly, we define the gradient divergence at the global optimum $\mathbf w^*$ 
(similar to \cite{8664630}).
\begin{defi}[gradient divergence]
\label{gradbo}
$
\nabla^*{\triangleq}\max_i\Vert\nabla f_{i}(\mathbf w^*)\Vert.
$
\end{defi}

\subsection{Main convergence result}
\label{mainconv}
The paper \cite{8786146} considers \emph{fixed} stepsizes and lacks a projection operator. 
The paper \cite{9224135} considers a different class of functions (convex, non-smooth with bounded subgradients) and uses a projection, but its proof leverages the structure of the quantization error, under which $\mathbf W_{k}+\gamma_k\tilde{\mathbf D}_{k}\in\mathcal W^N$. This property does not hold in this paper, due to the potentially unbounded consensus error. 
Hence, their analyses cannot be readily extended to the setup of this paper, necessitating a new line of analysis, presented next.

We assume $k\geq \bar{\kappa}$ under a suitable iteration $\bar{\kappa}\geq 0$ (to be defined later). It represents
 a regime in which stepsizes become sufficiently small to satisfy the conditions stated in Theorem~\ref{Tmain}.\footnote{For $k{<}\bar{\kappa}$, one can simply bound
$\Vert\mathbf w_{ik}{-}\mathbf w^*\Vert{\leq}\mathrm{dm}(\mathcal W)$.}
 The main idea is to decompose the error between $\mathbf W_k$ and the global optimum $\mathbf W^*$ into: \{1\} the error between $\mathbf W_k$ and the \emph{noise-free dynamics} $\overline{\mathbf W}_k$ (obtained by setting  the term (b) in \eqref{Wupdate} equal to $\mathbf 0$), defined as
$\overline{\mathbf W}_{\bar{\kappa}}=\mathbf W_{\bar{\kappa}}$,
\ba{
&
\overline{\mathbf W}_{k+1}
=\Pi^N[\overline{\mathbf W}_{k}-\eta_k\nabla G_k(\overline{\mathbf W}_{k})],\ \forall k\geq\bar{\kappa};
}{errfree}
\{2\} the error between  $\overline{\mathbf W}_k$ and the  minimizer of the \emph{time-varying} Lyapunov function $G_k$, defined at iteration $k$ as
\ba{
\mathbf W_k^*=\argmin_{\mathbf W\in\mathcal W^N} G_k(\mathbf W);
}{penalized}
and \{3\} the error between $\mathbf W_k^*$ and $\mathbf W^*$.
Accordingly, 
\begin{align}
\label{overallerrrr}
&\Big(\smash[b]{\mathbb E\Big[\sum_i\Vert{\mathbf w}_{ik}-\mathbf w^*\Vert^2\Big]}\Big)^{1/2}
=
\Vert{\mathbf W}_{k}-\mathbf W^*\Vert_{\mathbb E}
\\&\nonumber
=
\Vert(\mathbf W_{k}-\overline{\mathbf W}_{k})+(\overline{\mathbf W}_{k}-{\mathbf W}_{k}^*)
+({\mathbf W}_{k}^*-\mathbf W^*)\Vert_{\mathbb E}\\&
\leq
\Vert{\mathbf W}_{k}-\overline{\mathbf W}_{k}\Vert_{\mathbb E}
+\Vert\overline{\mathbf W}_{k}-{\mathbf W}_{k}^*\Vert_{\mathbb E}
+\Vert{\mathbf W}_{k}^*-\mathbf W^*\Vert,
\label{int30}
\end{align}
where the last step follows from Minkowski inequality \cite[Lemma 14.10]{florescu2013handbook}.
These terms are individually bounded in Theorem~\ref{Tmain}.
A significant challenge in proving Theorem~\ref{Tmain}, in contrast to \cite{8786146}, stems from the time-varying behavior of $G_k$, caused by the use of time-varying stepsizes:
in step \{2\}, one must carefully bound the \emph{tracking error} associated with the changes in the $G_k$-minimizer $\mathbf W_k^*$ over $k$.
This issue is absent when using constant stepsizes as in \cite{8786146}.
 \begin{theorem}\label{Tmain} 
 Assume $\forall k{\geq}\bar{\kappa}$: {\bf C1}: $\eta_k(\mu{+}L){+}\gamma_k\rho_N{\leq}2$;\\ {\bf C2}: $\frac{\eta_{k}}{\gamma_{k}}{\leq}\frac{\zeta\cdot \mu\rho_2}{\sqrt{N}\nabla^* L}$;
{\bf C3}: $\frac{\gamma_k}{\eta_k}\leq \frac{\gamma_{k+1}}{\eta_{k+1}}$.
Then, 
\begin{align}
&\frac{1}{\sqrt{N}}\Vert{\mathbf W}_{k}-\overline{\mathbf W}_{k}\Vert_{\mathbb E}
{\leq}
\Big[\sum_{t=\bar{\kappa}}^{k-1}P_{tk}^2\Big(\gamma_t^2\Sigma^{(1)}+\eta_t^2\Sigma^{(2)}\Big)\Big]^{\frac{1}{2}}\!\!\!\!,\!\!\label{L1}
\\
&\label{L2}
\frac{1}{\sqrt{N}}\Vert\overline{\mathbf W}_{k}-{\mathbf W}_{k}^*\Vert_{\mathbb E}{\leq}
\mathrm{dm}(\mathcal W)\cdot P_{(\bar{\kappa}-1)k}\nonumber\\&\quad
+\frac{\nabla^* L}{\mu\rho_2}
\sum_{t=\bar{\kappa}}^{k-1}P_{tk}\Big(1+\frac{L^2}{\mu\rho_2}\frac{\eta_{t}}{\gamma_{t}}\Big)\Big(\frac{\eta_{t}}{\gamma_{t}}-\frac{\eta_{t+1}}{\gamma_{t+1}}\Big)
,
\\&\label{L4}
\frac{1}{\sqrt{N}}\Vert\mathbf W_{k}^*-\mathbf W^*\Vert\leq \frac{\nabla^* L}{\mu\rho_2}
\frac{\eta_{k}}{\gamma_{k}}.
\\&
\text{where  }P_{tk}\triangleq\Pi_{j=t+1}^{k-1}(1-\mu\eta_j),\ \forall t\leq k-1.
\end{align}
 \end{theorem}
 \begin{proof}
 See Appendix B.
\end{proof}
The term \eqref{L1} captures the impact of the DGD noise through the variance terms $\Sigma^{(1)}$ and $\Sigma^{(2)}$. Eq. \eqref{L2} is composed of two terms: the first ($\mathrm{dm}(\mathcal W)P_{(\bar{\kappa}-1)k}$) accounts for the initial error 
 at time $\bar{\kappa}$; the second
is a direct consequence of time-varying stepsizes (in fact, it is zero for constant ones); as highlighted previously, it
represents the error accumulation due to tracking the changes of  the
$G_k$-minimizer ${\mathbf W}_{k}^*$ over $k$.
Eq. \eqref{L4} shows that the minimizer ${\mathbf W}_{k}^*$ of $G_k$ approximates $\mathbf W^*$  well when $\eta_k/\gamma_k$ is small.
We now specialize Theorem~\ref{Tmain} to the case of constant stepsizes (\secref{conststep}),
to gain insights on the design of decreasing stepsizes in \secref{decstep}.

\subsection{Constant stepsizes}
\label{conststep}
Let $\eta_k=\eta>0$, $\gamma_k=\gamma>0,\forall k$. For the sake of exposition, further assume  $\Sigma^{(2)}=0$,
and consider a target timeframe $K$ when the algorithm stops.
To satisfy {\bf C1}-{\bf C2} of Theorem \ref{Tmain}, we need
$\eta(\mu+L)+\gamma\rho_N\leq 2$ 
and $\frac{\eta}{\gamma}\leq \frac{\zeta\cdot \mu\rho_2}{\sqrt{N}\nabla^* L}$  
 with $\bar{\kappa}=0$. 
Then, 
at algorithm termination
 \eqref{L1}-\eqref{L4}  specialize as\begin{align}
\label{const1}
&\frac{1}{\sqrt{N}}\Vert{\mathbf W}_{K}-\overline{\mathbf W}_{K}\Vert_{\mathbb E}
\leq
\sqrt{\Sigma^{(1)}}\frac{\gamma}{\sqrt{\eta\mu}}
,
\\
\label{const2}
&\frac{1}{\sqrt{N}}\Vert\overline{\mathbf W}_{K}{-}\mathbf W_{K}^*\Vert_{\mathbb E}
{\leq}
\mathrm{dm}(\mathcal W)(1{-}\mu\eta)^k
{\leq}\mathrm{dm}(\mathcal W)e^{-\mu\eta K},\!\!\!
\\
\label{const3}
&\frac{1}{\sqrt{N}}\Vert\mathbf W_{k}^*{-}\mathbf W^*\Vert\leq\frac{ \nabla^*L}{\mu\rho_2}
\frac{\eta}{\gamma}.
\end{align}
To make these errors small, $\gamma/\sqrt{\eta}$, $\eta/\gamma$ and $e^{-\mu\eta K}$ need all be small,
yielding a trade-off between $\gamma$ and $\eta$.
 To make $e^{-\mu\eta K}{\to}0$ as $K{\to}\infty$ we require $\eta{\propto}K^{-(1-\epsilon)}$ for small $\epsilon>0$. 
 With such choice of $\eta$, the trade-off between
 \eqref{const1} and \eqref{const3} is optimized by
 making $\gamma/\sqrt{\eta}\propto \eta/\gamma$, yielding
 $\gamma{\propto}K^{-3/4(1-\epsilon)}$.
Under this choice, the three error terms
 \eqref{const1}-\eqref{const3}, hence the overall error \eqref{overallerrrr}, become of order $\mathcal O(K^{-1/4(1-\epsilon)})$
 for target $K$ iterations.
 A similar convergence rate is shown in \cite{8786146} under a different setup:\footnote{Note that $\alpha$, $\epsilon$, $\delta$, $T$ used in \cite[Theorem 1]{8786146} map as $\alpha\mapsto\eta/\gamma$, $\epsilon\mapsto\gamma$, $\delta\mapsto(1-\epsilon)/2$, $T\mapsto K$ in this paper} it assumes quantization, error-free communications, constant stepsizes and no projection.
 This analysis unveils some drawbacks of using constant stepsizes: 1) faster convergence is achieved with smaller $\epsilon$ (approaching $\mathcal O(K^{-1/4})$),
at the cost of larger \eqref{const2}; 
2) tuning of $\eta$ and $\gamma$ requires prior knowledge of the target timeframe $K$, as well as of $\rho_2$, $\rho_N$, $\zeta$, $\nabla^*$ to satisfy {\bf C1}-{\bf C2} of
Theorem~\ref{Tmain};
3) the stepsize conditions needed to ensure these convergence properties may be overly stringent to meet in practical scenarios.
 These limitations are overcome via decreasing stepsizes.

\subsection{Decreasing stepsizes}
\label{decstep}
The constant stepsize analysis suggests that 
$\gamma_k{\propto}k^{-3/4}$ and $\eta_k{\propto}k^{-1}$.
Accordingly, we choose $\gamma_k{=}\gamma_0(1{+}\delta k)^{-3/4}$ and $\eta_k{=}\eta_0(1{+}\delta k)^{-1}$
with initial stepsizes $\gamma_0,\eta_0$. The parameter
 $\delta>0$, termed \emph{decay rate}, controls how quickly $\eta_k$, $\gamma_k$ decrease over time.
With this choice: 1) $\gamma_k/\eta_k$ is non-decreasing in $k$ ({\bf C3} of Theorem~\ref{Tmain});
2) since $\eta_k,\gamma_k,\eta_k/\gamma_k\to 0$ for $k\to\infty$,
there exists $\bar{\kappa}{\geq}0$ satisfying {\bf C1}-{\bf C2} of Theorem~\ref{Tmain}  $\forall k{\geq}\bar{\kappa}$.
The following result specializes Theorem \ref{Tmain} to this choice.
\begin{theorem}
\label{T1}
Let $\gamma_k=\gamma_0(1+\delta k)^{-3/4}$, $\eta_k=\eta_0(1+\delta k)^{-1}$ with 
$\delta\leq\frac{4}{5}\mu\eta_0$. 
Then, 
 $\forall k\geq \bar{\kappa}$,
\begin{align}
\label{thm1_1}
&\frac{1}{\sqrt{N}}\Vert{\mathbf W}_{k}-\overline{\mathbf W}_{k}\Vert_{\mathbb E}
\leq
\frac{\sqrt{5}e}{2\sqrt{\mu}}\Big[
\frac{\gamma_0}{\sqrt{\eta_0}}\sqrt{\Sigma^{(1)}}
\nonumber\\&\qquad
+
\sqrt{\eta_0\Sigma^{(2)}}
(1+\delta k)^{-1/4}\Big](1+\delta k)^{-1/4},
\\\label{thm1_2}
&\frac{1}{\sqrt{N}}\Vert\overline{\mathbf W}_{k}-\mathbf W_{k}^*\Vert_{\mathbb E}
\leq
\mathrm{dm}(\mathcal W)\Big(1+\frac{\delta(k-\bar{\kappa})}{1+\delta\bar{\kappa}}\Big)^{-5/4}
\nonumber
\\&
+\frac{\nabla^* L}{\mu\rho_2}\frac{e}{4}
\frac{\eta_0}{\gamma_0}
\Big(
 1
+\frac{L^2}{\mu\rho_2}\frac{\eta_0}{\gamma_0}
(1+\delta k)^{-1/4}
 \Big)(1+\delta k)^{-1/4},
\\ \label{thm1_3}
&\frac{1}{\sqrt{N}}\Vert\mathbf W_{k}^*-{\mathbf W}^*\Vert\leq \frac{\nabla^* L}{\mu\rho_2}
\frac{\eta_0}{\gamma_0}(1+\delta k)^{-1/4}.
\end{align}
\end{theorem}
\begin{proof}
See Appendix C.
\end{proof}
Notably, this result holds for any initial stepsizes $\eta_0,\gamma_0{>}0$, unlike the constant stepsize case which demands strict adherence to conditions {\bf C1}-{\bf C2}. This flexibility allows for numerical optimization of $\eta_0$ and $\gamma_0$ while maintaining the convergence guarantees of Theorem \ref{T1}, highlighting a trade-off in their choice. Tuning $\gamma_0$ reflects a delicate balance between facilitating information propagation (favored by larger $\gamma_0$) and mitigating error propagation (smaller $\gamma_0$).
 As to the choice of $\delta$, as it increases, both $\eta_k$ and $\gamma_k$ get smaller, reducing $\bar{\kappa}$--the iteration index meeting the convergence conditions of Theorem~\ref{Tmain}. Thus, the regime $k{\geq}\bar{\kappa}$ is attained more rapidly. At the same time, terms like $(1{+}\delta k)^{-1/4}$ and $(1{+}\delta(k{-}\bar{\kappa})/(1+\delta\bar{\kappa}))^{-5/4}$ get smaller. 
We conclude that $\delta{=}\frac{4}{5}\mu\eta_0$ minimizes \eqref{thm1_1}-\eqref{thm1_3}, further motivating the use of decreasing vs constant ($\delta{\to}0$) stepsizes. This design insight is corroborated numerically in Fig.~\ref{fig:params} of the supplemental document.
   
Combining these results into \eqref{int30}, it readily follows that 
\begin{align}
\label{convk}
\smash[b]{\mathbb E\Big[\frac{1}{N}\sum_i\Vert{\mathbf w}_{ik}-\mathbf w^*\Vert^2\Big]}=\mathcal O(1/\sqrt{k}).
\end{align}
 To get some intuition behind this result, consider \eqref{Wupdate} with fixed consensus stepsize and no projection:
 \begin{align}
\label{eqx_NF_nogamma}
\mathbf W_{k+1}=(\mathbf I{-}\gamma\LL){\mathbf W}_{k}-\eta_k\nabla f(\mathbf W_{k})+\e_k^{(tot)},
\end{align}
where $\e_k^{(tot)}{\triangleq}\gamma\e_k^{(1)}{-}\eta_k\e_k^{(2)}$ is the overall DGD noise.
These updates resemble the DGD algorithm \eqref{ddd}, where $\mathbb E[\Vert\mathbf W_{k}^*-{\mathbf W}^*\Vert^2]{=}\mathcal O(1/k)$ with stepsize  $\eta_k{\propto}1/k$,
provided that $\mathrm{var}(\e_k^{(tot)}){\propto}k^{-2}$.
This condition holds when stochastic gradients are the sole source of DGD noise, since
 $\Vert-\eta_k\e_k^{(2)}\Vert_{\mathbb E}^2{\propto}\eta_k^2{\propto}k^{-2}$, but it fails in the presence of consensus error, which becomes dominant. 
 To mitigate this error, a decreasing consensus stepsize $\gamma_k{\propto}k^{-3/4}$ is required, at the expense of slower convergence ($\mathcal O(1/\sqrt{k})$).
 
  \begin{figure}
\floatbox[{\capbeside\thisfloatsetup{capbesideposition={left,center},capbesidewidth=4.7cm}}]{figure}[0.7\FBwidth]
{\caption{Example of \emph{spatially-dependent} deployment scenario with $N=40$ nodes.
Each node holds data from only one class, indicated by the label indices '0' to '9' (4 nodes per class). For instance,
node $\star$ holds data from class '3'.
With 3 reflectors, there are 4 signal paths, shown in the figure between a generic transmitter ($\star$) and receiver ($\blacktriangle$) pair.}\label{fig:simresdep}}
{\includegraphics[width=\linewidth]{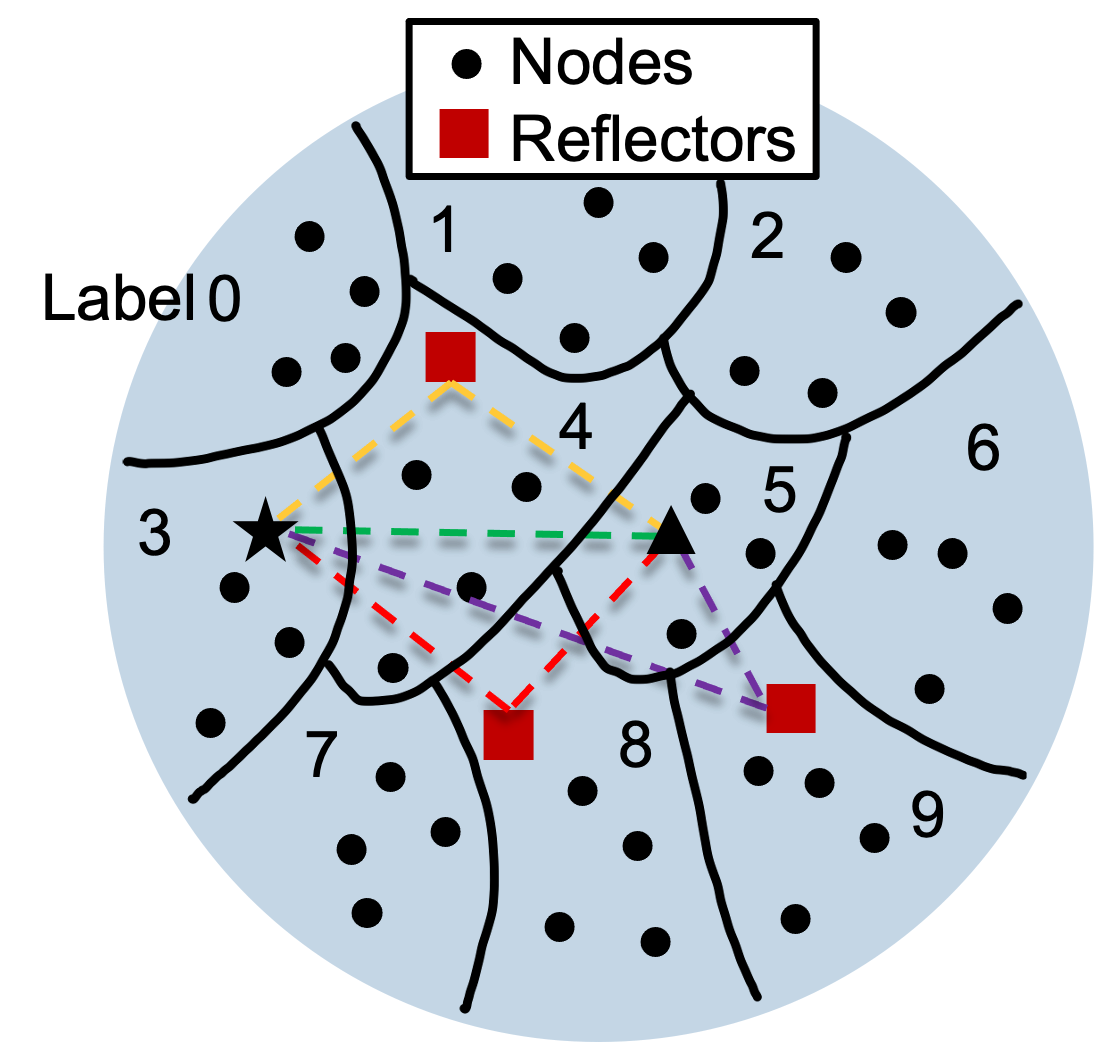}}
\end{figure}

\begin{figure*}
     \centering
               \hfill
          \begin{subfigure}[b]{0.25\linewidth}
        \includegraphics[width = \linewidth,trim=10 0 25 20, clip]{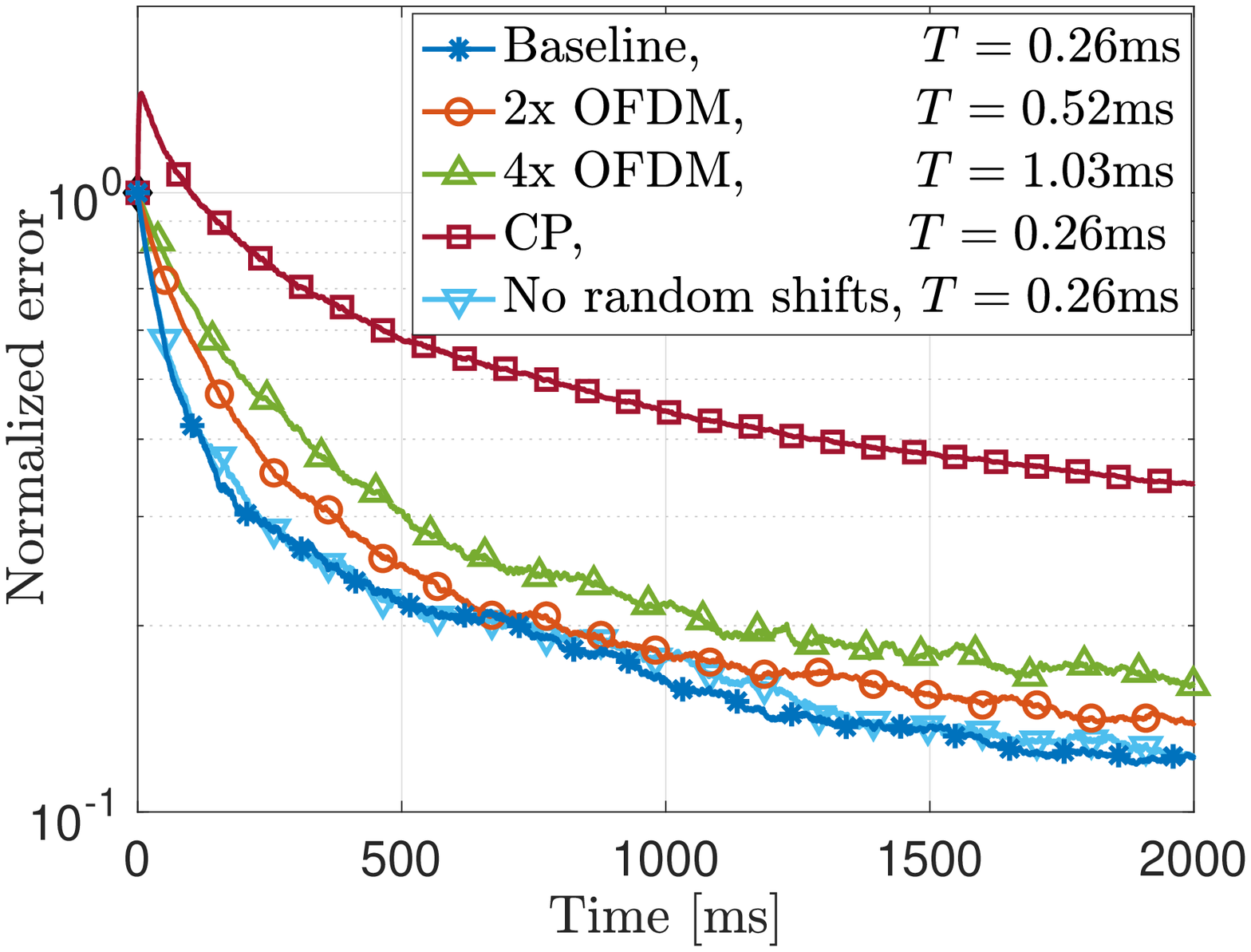}
	    \caption{} 
     \end{subfigure}
     \hfill
     \begin{subfigure}[b]{0.24\linewidth}
         \includegraphics[width = \linewidth,trim=30 0 25 20, clip]{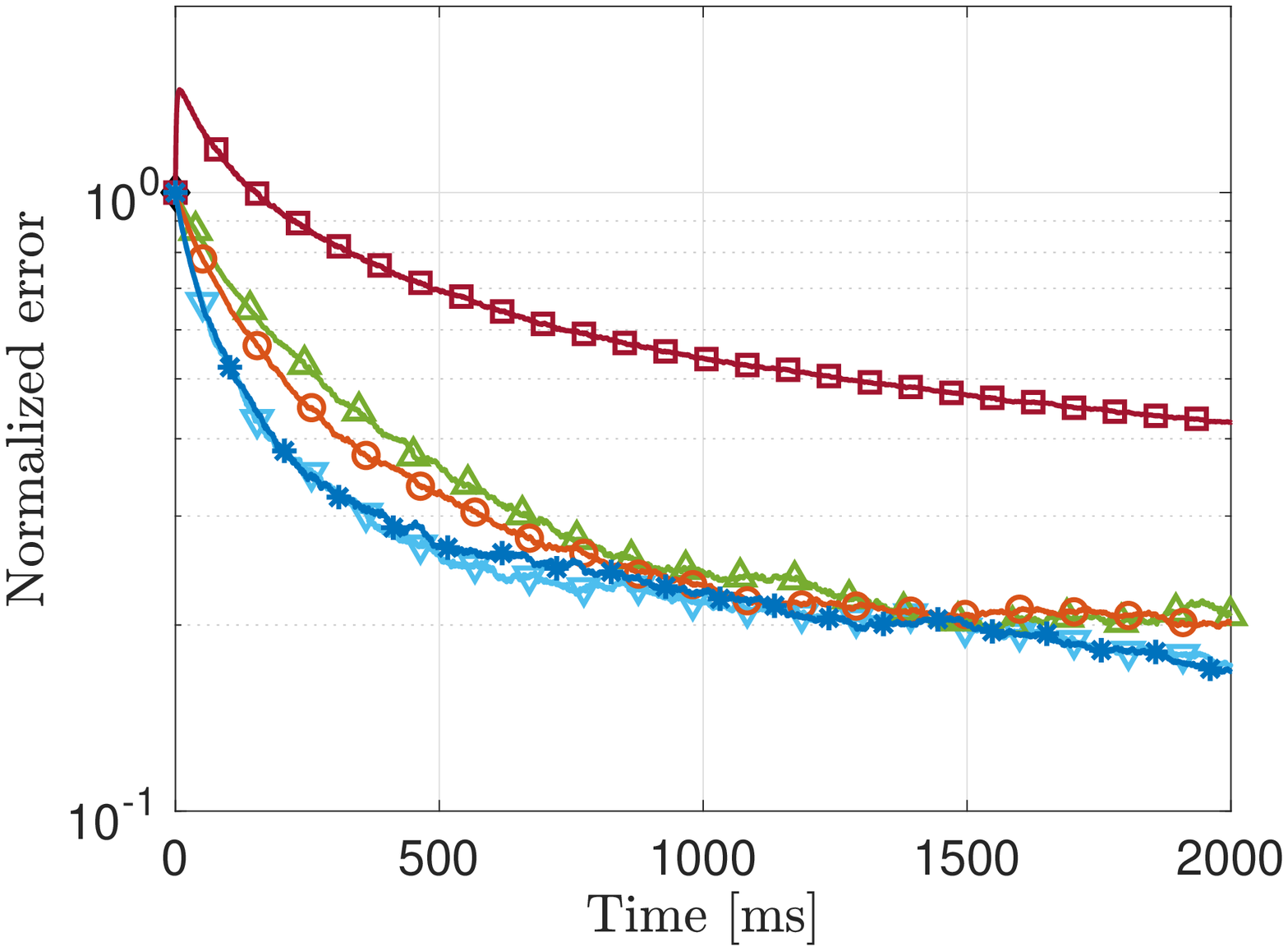}
	    \caption{} 
     \end{subfigure}
     \begin{subfigure}[b]{0.24\linewidth}
        \includegraphics[width = \linewidth,trim=30 0 25 20, clip]{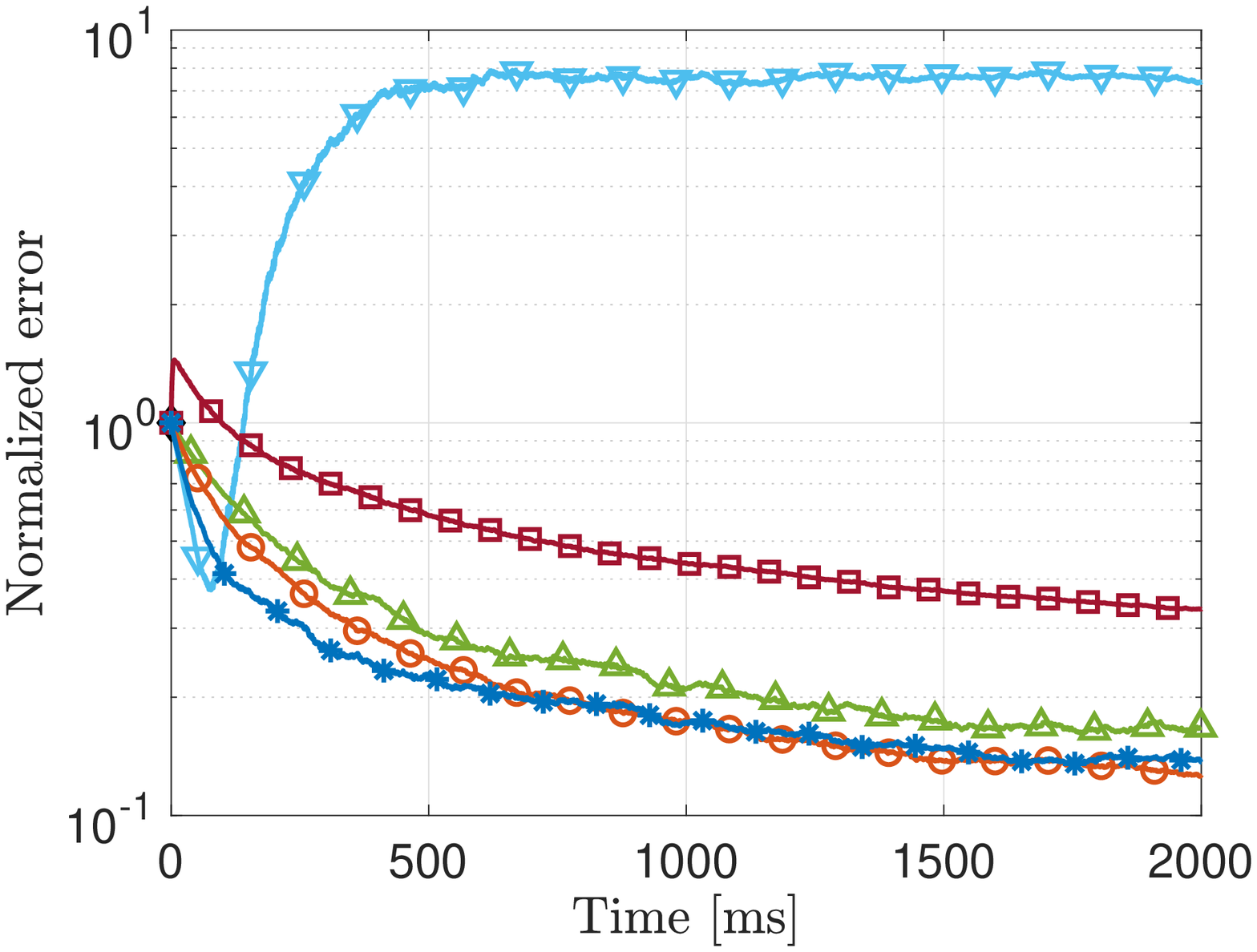}
	    \caption{} 
     \end{subfigure}
     \hfill
     \begin{subfigure}[b]{0.24\linewidth}
         \includegraphics[width = \linewidth,trim=30 0 25 20, clip]{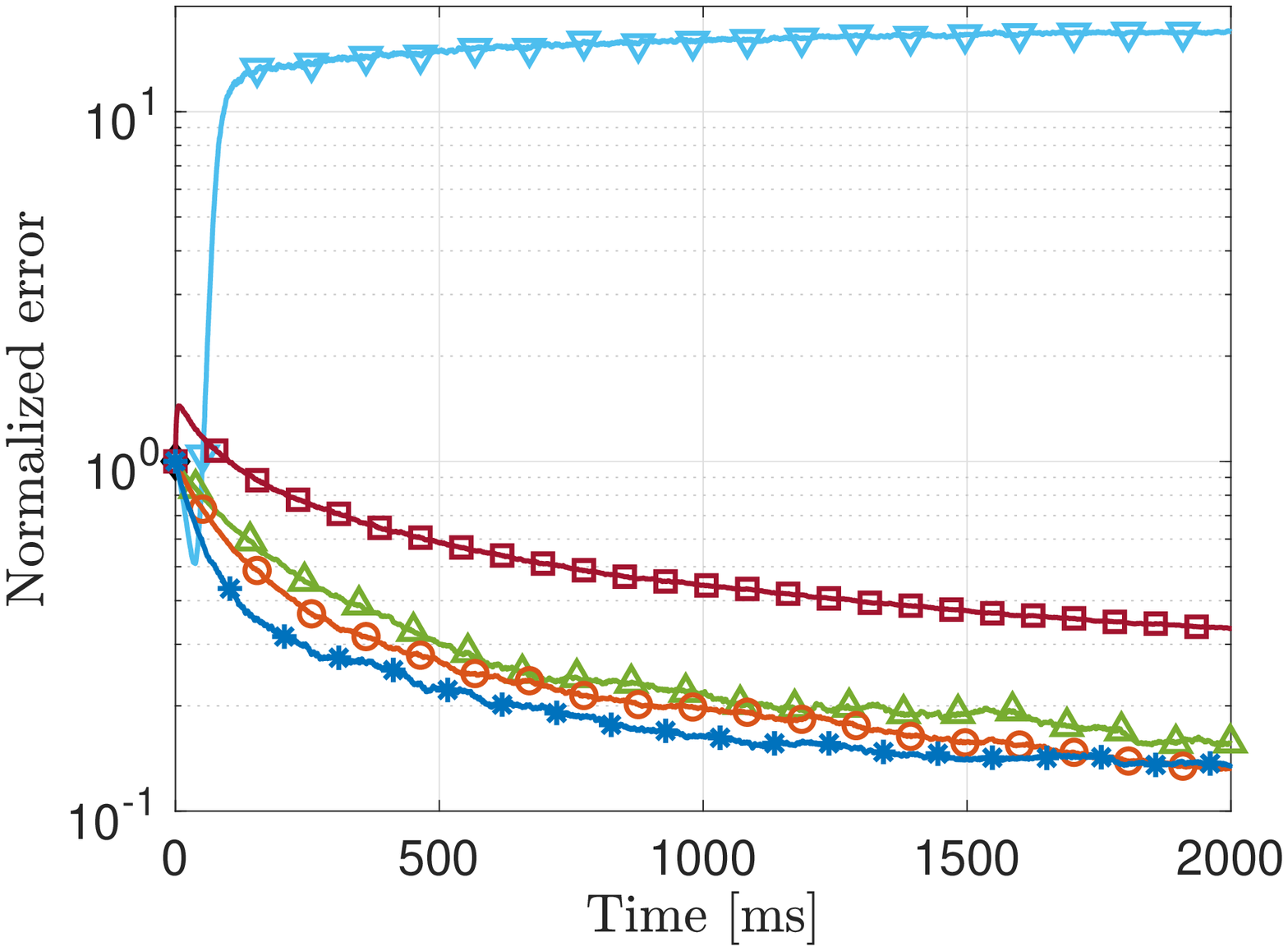}
	    \caption{} 
     \end{subfigure}
          \hfill
\caption{Normalized error vs time, for different configurations of \proposed\ and four different scenarios.
The common legend is shown in the left figure, and shows the frame duration of each configuration.
(a) Spatially-i.i.d. labels and i.i.d. channels;
(b) Spatially-dependent labels and i.i.d. channels;
(c) Spatially-i.i.d. labels and block-fading channels (2ms coherence time);
(d) Spatially-i.i.d. labels and static channels. 
\vspace{-5mm}
 } \label{fig:propcomp}
\end{figure*}
 
\section{Numerical Results}
\label{numres}
We solve a classification task on fashion-MNIST \cite{fmnist}, a dataset of grayscale images of fashion items from 10 classes. It is a variation of the popular MNIST dataset, used in recent papers such as
\cite{10360352,10384479}.

\underline{Network deployment}:
 $N$ nodes ($200$, unless otherwise stated) are spread uniformly at random over a circular area of $2$km radius. They communicate over a $W_{\mathrm{tot}}{=}$5MHz bandwidth,  $f_c{=}3$GHz carrier frequency,
 with $P_{tx}{=}20$dBm power  ($E=P_{tx}/W_{\mathrm{tot}}$).
The noise power spectral density at the receivers is $N_0{=}-173$dBmW/Hz.
We use OFDM signaling with $\mathrm{SC}{=}512$ subcarriers and a cyclic prefix of length $\mathrm{CP}{=}133$, accommodating propagation delays up to $26.6\mu$s. Each OFDM symbol has duration $T_{\mathrm{ofdm}}{=}(\mathrm{SC}{+}\mathrm{CP})/W_{\mathrm{tot}}{=}129\mu$s.

 \underline{Channel model}: 
We generate spatially consistent channels by randomly placing 3 reflectors (see Fig. \ref{fig:simresdep}).
The channel between transmitter $j$ and receiver $i$ in the $m$th subcarrier is
\\\centerline{$
[\mathbf h_{ij}]_m=\sum_{p=0}^3\sqrt{\alpha_{ijp}} \phi_{ijp} e^{-\mathrm{j}2\pi\tau_{ijp}W_{\mathrm{tot}}m/\mathrm{SC}},
$}
where $p{=}0$ is the LOS path and $p{\neq}0$ are the reflected paths;
$\tau_{ijp}{=}\frac{d_{ijp}}{c}$ is the path propagation delay, 
$d_{ijp}$ is its distance traveled, $c$ is the speed of light;
  $\alpha_{ijp}{=}(\frac{c}{4\pi f_cd_{ijp}})^2$ is the path gain, based on Friis' free space equation; 
 $\phi_{ijp}$ is the fading coefficient:
$|\phi_{ijp}|{=}1$  with
uniform phase for the LOS path,
 $\phi_{ijp}\sim\mathcal{CN}(0,1)$ (Rayleigh fading) for reflected paths.

\underline{Data deployment}:
The global dataset is constituted of 1000 low-resolution images (100 from each class), distributed across the $N$ nodes.
To emulate data heterogeneity,
each node has a local dataset with $1000/N$ images 
\emph{from a single class}, 
$c_i{\in}\{0,\dots,9\}$ for node $i$, 
so that $N/10$ nodes have label (fashion item) '0', $N/10$ have label '1', and so on.
Cooperation  is then necessary to solve the classification task, due to lack of examples of other classes at any single node.
Each 7x7 pixels image is converted into a 50-dimensional feature $\mathbf f{\in}\mathbb R^{50}$ (including a fixed offset) and normalized to $\Vert\mathbf f\Vert_2=1$.

 Furthermore, we consider two different spatial data distributions. In the \emph{spatially-i.i.d.} case, the nodes' labels are assigned in an i.i.d. fashion to each node, irrespective of their location.
In the \emph{spatially-dependent} case, shown in Fig. \ref{fig:simresdep}, the circular region is divided into
10 subregions, each containing an equal number of nodes with the same label, so that the label of a certain node depends on its position in the network.
 
 \underline{Optimization problem}:
 We solve the task via regularized cross-entropy loss minimization, with loss function
  \\\centerline{$\phi(c,\mathbf f;\mathbf w)=\frac{\mu}{2}\Vert\mathbf w\Vert_2^2
-\ln\Big(\frac{\exp\{\mathbf f^\top\mathbf w^{(c)}\}}{\sum_{j=0}^9 \exp\{\mathbf f^\top\mathbf w^{(j)}\}}\Big)$}
 for feature $\mathbf f$ with label $c$,
where: $\mathbf w^\top{=}[\mathbf w^{(1)\top},\dots,\mathbf w^{(9)\top}]{\in}\mathbb R^d$ is a $d{=}450$-dimensional parameter vector, with $\mathbf w^{(c)}{\in}\mathbb R^{50}$,  $\mathbf w^{(0)}{=}\mathbf 0$, 
 $\mu{=}0.001{>}0$. 
 Hence, $f_i(\mathbf w){=}\frac{1}{|\mathcal D_i|}\sum_{\mathbf f\in\mathcal D_i}\phi(c_i{,}\mathbf f;\mathbf w)$. 
All $f_i(\mathbf w)$ and the global function $F(\mathbf w)$
are $\mu$-strongly-convex and $L{=}\mu{+}2$-smooth.
$\mathcal W$ is the $d$-dimensional sphere centered at  $\mathbf 0$ 
 with radius
$r=\frac{1}{\mu}\Vert\nabla F(\mathbf 0)\Vert$ (Example \ref{ex1}).

\underline{Computational model}:
Each gradient computation $\nabla\phi(\cdot)$ on a single datapoint
takes $T_{\mathrm{gr}}{=}30 \mu$s.\footnote{Estimated based on a 2.4 GHz 8-Core Intel Core i9 processor.} Hence, 
computing the gradient over a minibatch $\mathcal B_{i}{\subseteq} \mathcal D_i$ takes 
$|\mathcal B_{i}|T_{\mathrm{gr}}$.
Since gradient computations are done in parallel with communications, we set the minibatch size as
$|\mathcal B_{i}|=\min\{\lfloor T/T_{\mathrm{gr}}\rfloor,| \mathcal D_i|\}$, consistent with the frame duration $T$ specific to each algorithm.

\begin{figure*}
     \centering
               \hfill
          \begin{subfigure}[b]{0.275\linewidth}
        \includegraphics[width = \linewidth,trim=10 0 30 20, clip]{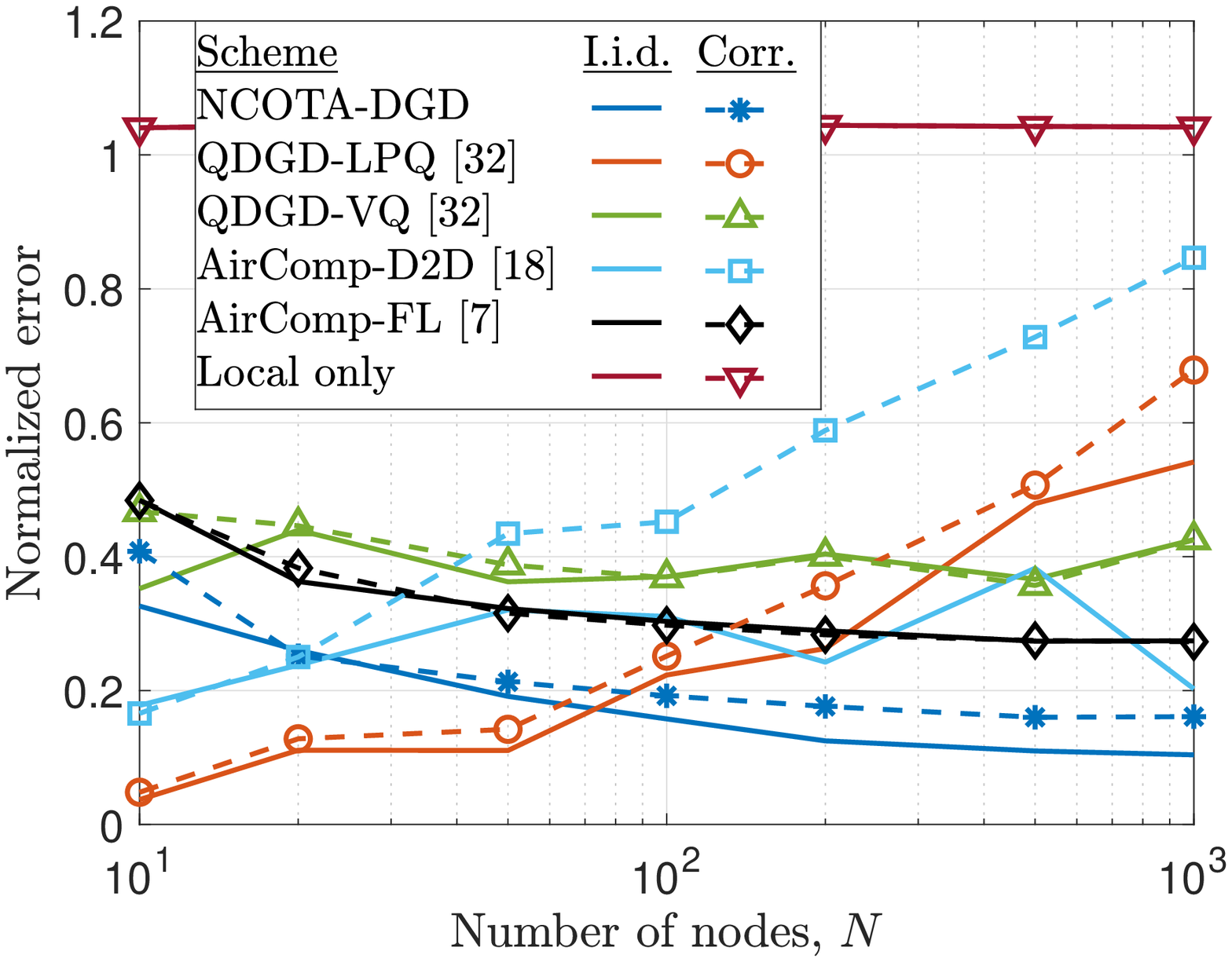}
        	    \vspace{-5mm}
	    \caption{\vspace{-5mm}} 
     \end{subfigure}
     \hfill
     \begin{subfigure}[b]{0.27\linewidth}
         \includegraphics[width = \linewidth,trim=10 0 30 20, clip]{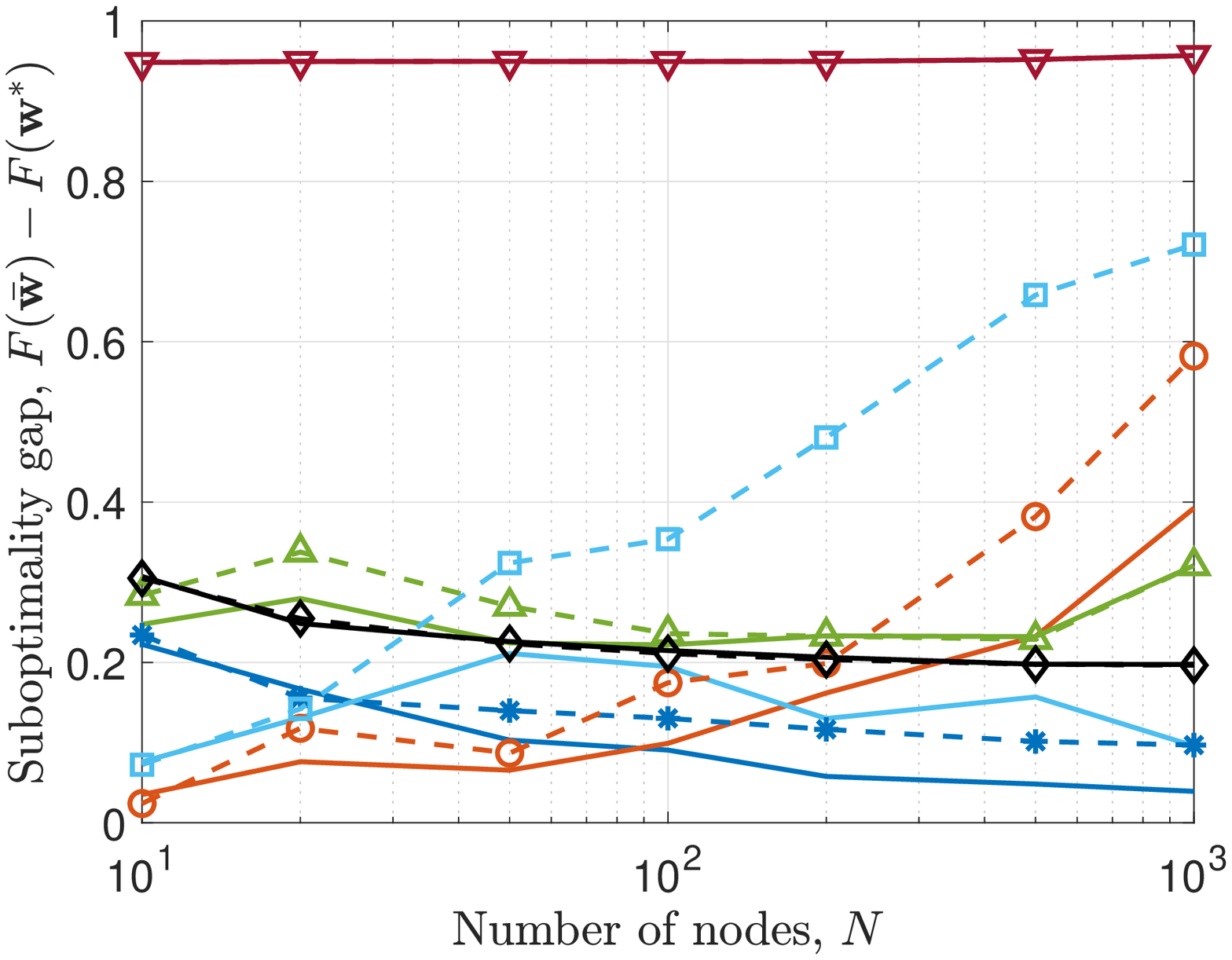}
         \vspace{-5mm}
	    \caption{\vspace{-5mm}} 
     \end{subfigure}
     \begin{subfigure}[b]{0.27\linewidth}
        \includegraphics[width = \linewidth,trim=10 0 30 20, clip]{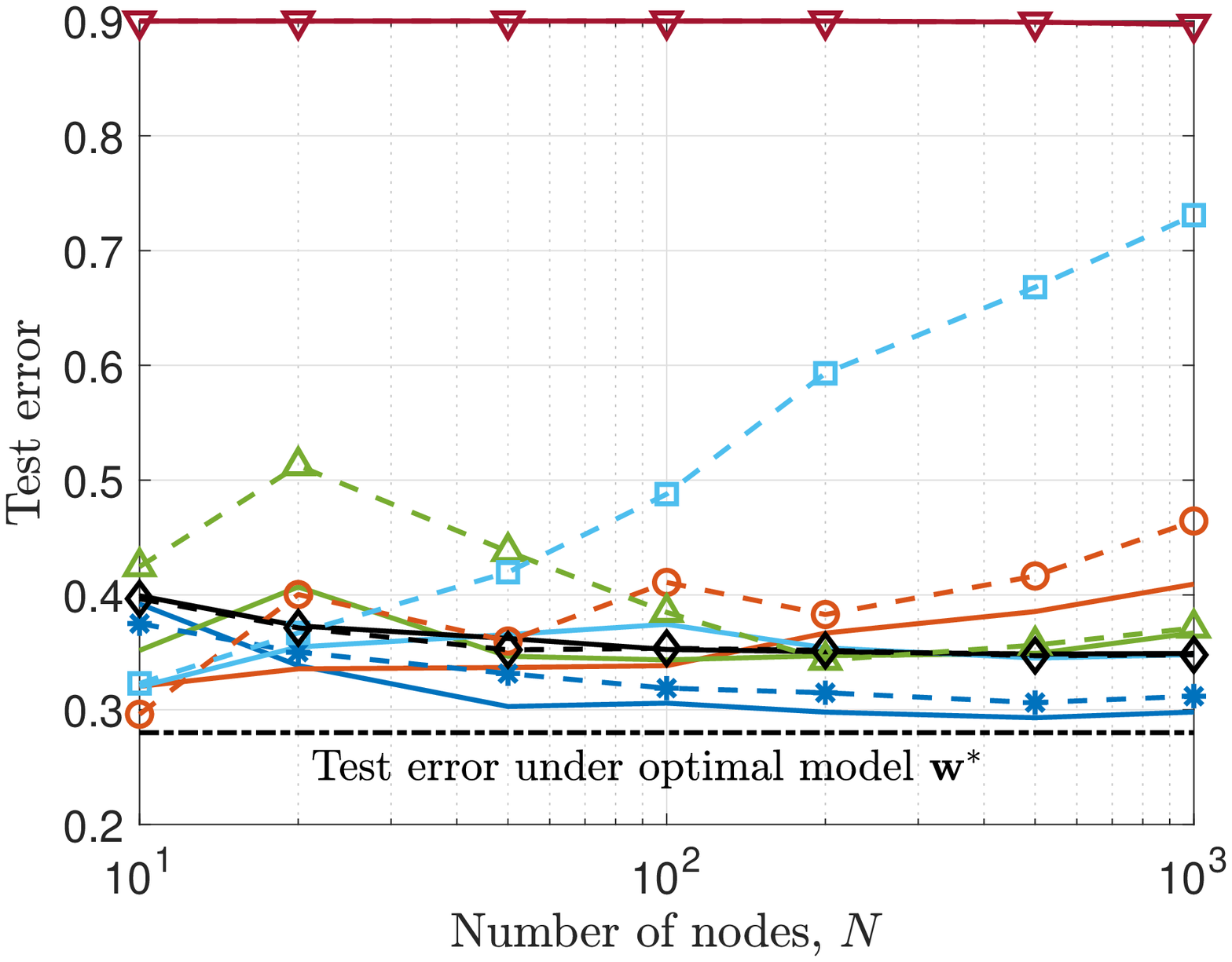}
       \vspace{-5mm}
	    \caption{\vspace{-5mm}} 
     \end{subfigure}
     \begin{subfigure}[b]{0.165\linewidth}
        \includegraphics[width = \linewidth,trim=0 0 15 20, clip]{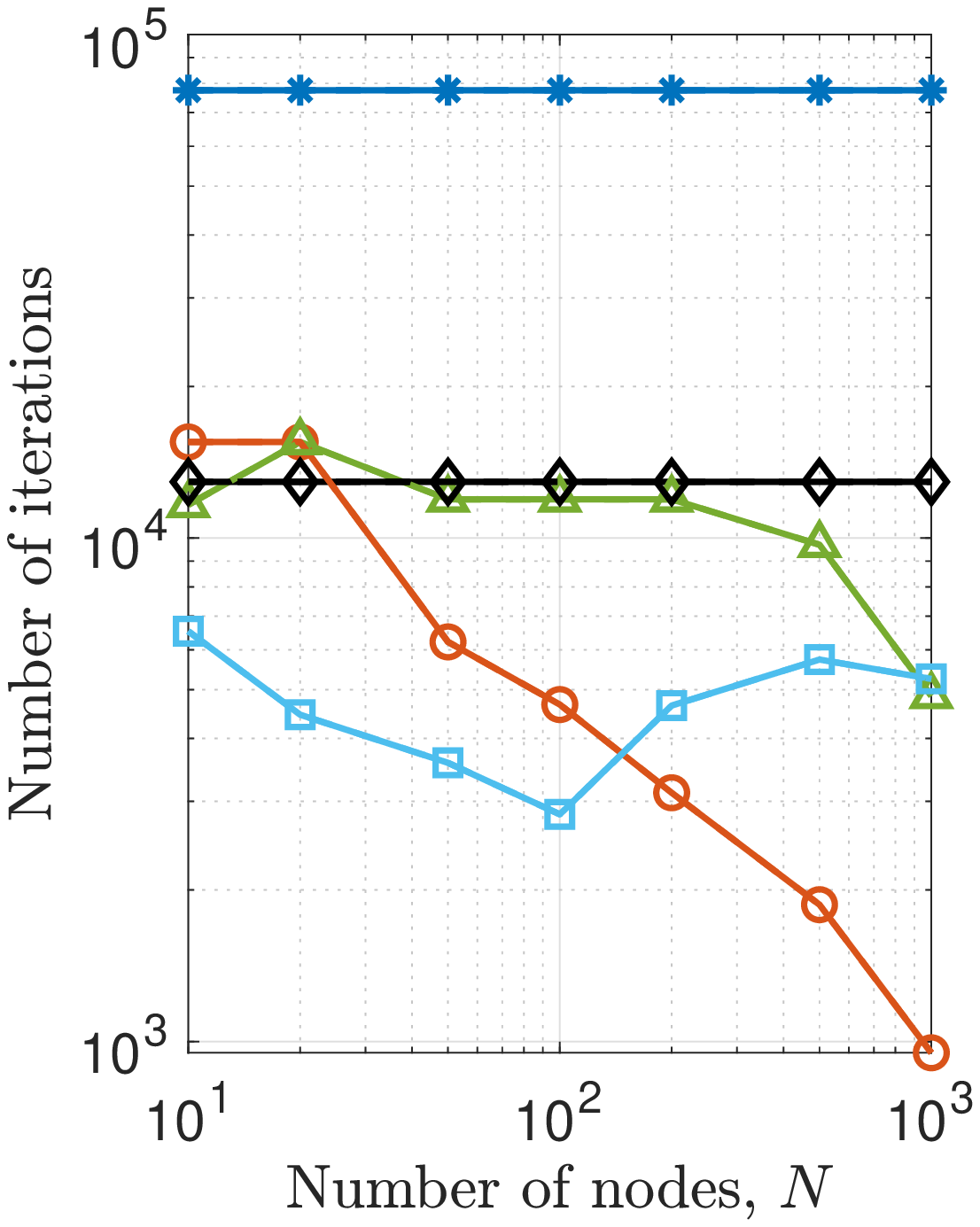}
        	 \vspace{-5mm}
	    \caption{\vspace{-5mm}} 
     \end{subfigure}     
          \hfill
\caption{Normalized error (a), suboptimality gap (b), test error (c) and number of iterations (d) vs number of nodes $N$, after 2000ms of execution time,
under spatially-i.i.d. (solid lines) and -dependent (dashed lines with markers) label scenarios, with i.i.d. channels over frames.
Common legend shown in figure (a).
\vspace{-6mm}
 } \label{fig:SoAvsN}
\end{figure*}

\subsection{Comparison of variants of \proposed}
\label{compproposed}
We compare different configurations of \proposed:
\\
$\bullet$ \underline{\emph{Baseline}} uses 
the CP0 codebook of Example \ref{ex1} with $M{=}2d{+}1{=}901$ codewords,
requiring two OFDM symbols ($Q{=}1024$), and yielding a frame of duration
$T{=}2T_{\mathrm{ofdm}}{=}258\mu$s. 
It follows Algorithm~\ref{A1} with parameters $\eta_0{=}\frac{2}{\mu+L}$, $\gamma_0{=}\frac{0.05}{\rho_2}$
 and $\delta{=}\frac{4}{5}\mu\eta_0$, in line with Theorem \ref{T1}
 and as motivated in the discussion following Theorem \ref{T1}.
The transmission probability $p_{\mathrm{tx}}$ is set as in Lemma \ref{L0}, with $\vartheta{=}1$ and $\varpi{\approx}\sqrt{M/Q}$.
 This choice of parameters is corroborated by numerical evaluations in the supplemental document (Fig. \ref{fig:params}).
\\
$\bullet$ \underline{\emph{2xOFDM}} and \underline{\emph{4xOFDM}} are the same as \emph{Baseline}, but they use 
$\times 2$ and $\times 4$ as many OFDM symbols, yielding $Q{=}2048$ ($T{=}516\mu$s) and $Q{=}4096$ ($T{=}1032\mu$s), respectively.
From Lemma~\ref{L0}, larger $Q$ yields smaller  noise in the consensus estimation, at the cost of longer frame duration $T$.
\\
$\bullet$ \underline{\emph{CP}} is the same as \emph{Baseline}, but 
it employs the conventional cross-polytope codebook of \cite{9740125} with $\phi{=}1{-}\frac{1}{\sqrt{d}r}\Vert\mathbf w_{i}\Vert_1$.
\\
$\bullet$ \underline{\emph{No random shifts}} is the same as \emph{Baseline}, but
it does not apply the random phase and coordinated circular subcarrier shifts on the transmitted signal, described in \secref{broadclass}.

 All variants are initialized as $\mathbf w_{i}{=}\mathbf 0,\forall i$.
We evaluate them on four different scenarios as described in the caption of Fig. \ref{fig:propcomp}, differing in the
data deployments (spatially-i.i.d. or -dependent) and channel fading properties: i.i.d., in which channels are i.i.d. over frames;
block-fading, in which they vary every 2ms (coherence time); and static, in which they remain fixed over the entire simulated interval.

In Fig. \ref{fig:propcomp}, we plot the \emph{normalized error}, 
$\frac{\hat{\mathbb E}[\sum_i\Vert{\mathbf w}_{i}{-}\mathbf w^*\Vert^2]}{N\Vert\mathbf w^*\Vert^2}$, versus the execution time ($kT$ after $k$ iterations). Here, $\hat{\mathbb E}$ denotes a sample average across 20
algorithm trajectories generated through independent realizations of network deployment, channels, AWGN noise,
and randomness used by the algorithm (transmission decisions, phase and circular subcarrier shifts, minibatch gradient selections).
We notice that \emph{Baseline} consistently performs the best across all different scenarios.
It has shorter frame duration but higher variance of the 
disagreement signal estimation (Lemma \ref{L0}) than
 {\emph{2xOFDM}} and {\emph{4xOFDM}}. This demonstrates that it is preferable to perform more frequent, albeit noisier, iterations of DGD.
 The CP0 codebook of \emph{Baseline} has lower variance  of the disagreement signal estimation than \emph{CP}, yielding better performance.\footnote{Analysis of this behavior is left for future work; it requires a tighter bound than Lemma \ref{L0} (valid under any codebook), by exploiting properties of CP$\phi$.}
\emph{No random shifts} diverges with non-i.i.d. channels (Figs. \ref{fig:propcomp}.c-d),
due   to a bias term in the disagreement signal estimation (see \eqref{biasissue}), which accumulates over time. In contrast, all other schemes employing the random phase and circular subcarrier shifts developed in \secref{broadclass} are not affected by different channel fading properties.
Fig. \ref{fig:propcomp}.b depicts a slight degradation in performance for all schemes in the spatially-dependent label scenario.
To explain this behavior, consider the receiving node labeled as $\blacktriangle$ in Fig. \ref{fig:simresdep}, carrying label '5': it receives strong signals from the nearby nodes carrying the same label '5' (less informative to $\blacktriangle$), but weaker signals from the nodes with label '0' (more informative to $\blacktriangle$). This results in slower propagation of information across the network.

\subsection{Comparison with state-of-the-art (SoA) schemes}
\label{compSoA}
In the remainder, we use the \emph{Baseline} configuration of \proposed, described in \secref{compproposed}, due to its superior performance observed in Fig. \ref{fig:propcomp}.
We compare it with adaptations of SoA works to the setting of this paper:
\\
$\bullet$ \underline{\emph{QDGD-LPQ}}/\underline{\emph{-VQ}}, adaptation of \cite{8786146},
which investigates the design and convergence analysis of Quantized-DGD under fixed learning and consensus stepsizes. Yet,
it assumes error-free communications, hence is not tailored to wireless systems affected by fading and interference.
To adapt \cite{8786146} to our setting,
nodes transmit over orthogonal channels via OFDMA, i.e., multiple nodes may transmit simultaneously on the same OFDM symbol, across orthogonal subcarriers. Node $i$ quantizes 
$\mathbf w_{i}$ using $B$ bits, and broadcasts the payload over the wireless channel on the assigned subcarriers $\mathcal S_i$.
Upon receiving the signal broadcast by node $j$, node $i$ decodes the payload correctly (denoted by the success indicator $\chi_{ij}{=}1$) if and only if
$
B{<}\sum_{s\in\mathcal S_j}\log_2(1{+}\frac{E}{N_0}\frac{\mathrm{SC}}{|\mathcal S_j|}|[\mathbf h_{ij}]_s|^2),
$
where $\mathcal S_j$ are the resource units allocated to node $j$. Otherwise, an outage occurs ($\chi_{ij}{=}0$).
Let $N_{rx,i}$ be the number of packets successfully received,
 and $\hat{\mathbf w}_{j}$ the reconstruction of ${\mathbf w}_{j}$ induced by the quantization operator.
 Node $i$ then follows the updates \eqref{updateeq}, with
$\tilde{\mathbf d}_{i}{=}\frac{1}{N_{rx,i}}\sum_{j}\chi_{ij}(\hat{\mathbf w}_{j}{-}{\mathbf w}_{i})$.
\\\indent
We consider two quantization schemes: the low precision quantizer (LPQ, see \cite[Example 2]{8786146}) and
the cross-polytope vector quantization (VQ) technique \cite{9740125}.
LPQ normalizes $\mathbf w_i$ by  $\Vert\mathbf w_{i}\Vert_\infty$, and then
quantizes each component of $\mathbf w_{i}/\Vert\mathbf w_{i}\Vert_\infty\in[-1,1]^d$ with $b$ bits using dithered quantization.
The overall payload is $B{=}64{+}b\cdot d$ bits,
including the magnitude $\Vert\mathbf w_{i}\Vert_\infty$ encoded with machine precision (64 bits).
VQ is used  with $\mathrm{REP}$ repetitions to reduce the variance of the quantization noise.
Since the cross-polytope codebook contains $2d$ codewords,
 the overall payload including repetitions is $B{=}\mathrm{REP}\lceil\log_2(2d)\rceil$.
 The free parameters (number of subcarriers $\mathrm{SC}_n$ per node, bits for LPQ, repetitions for VQ) are optimized numerically for best performance.
 \emph{The constant consensus and learning stepsizes are also optimized numerically,  hence may not satisfy the stringent conditions of \cite{8786146}.}
Since one OFDM symbol fits $\frac{\mathrm{SC}}{\mathrm{SC}_n}$ transmissions,
the frame duration is $T{=}N\cdot \frac{\mathrm{SC}_n}{\mathrm{SC}}\cdot T_{\mathrm{ofdm}}$.
\\$\bullet$ \underline{\emph{AirComp-D2D}}, inspired by \cite{9563232}, leverages AirComp to efficiently solve DGD over device-to-device mesh networks.
 It organizes the network into non-interfering, star-based sub-networks via graph-coloring, scheduling multiple such sub-networks simultaneously to enhance spectral efficiency. Each sub-network operates in paired slots: in the first, the central node receives coherently aggregated models from neighbors via AirComp with channel inversion; in the second, it broadcasts its own  model to its neighbors. 
   This process continues until all D2D links have been activated. 
  We generate the mesh network based on a maximum distance criterium, i.e. nodes $i,j$ are connected if and only if they are less than
$\mathrm{dist}_{\max}$ apart. We use the Metropolis-Hastings rule to design the mixing weights $\omega_{ij}$ \cite{6854643}.
While  \cite{9563232} assumes Rayleigh flat-fading channels with no interference between non-neighboring devices, our adaptation extends to frequency-selective channels and accounts for interference.
Note that \emph{AirComp-D2D} requires CSI  for power control, relies on \emph{instantaneous} channel reciprocity for channel inversion, 
and requires centralized knowledge of the graph topology 
for graph-coloring. In contrast, \proposed\ does not require any such knowledge,
 and relies only on \emph{average} (vs \emph{instantaneous}) reciprocity, $\Lambda_{ij}{=}\Lambda_{ji}$.
The free parameters ($\mathrm{dist}_{\max}$, stepsizes) are optimized numerically for best performance.
We refer to \cite{9563232} for further details. 
\\
$\bullet$ \underline{\emph{AirComp-FL}}, based on \cite{8870236}, solves \ref{global} via FL across a star topology, with one node acting as the parameter server (PS, the star center). 
The process unfolds in three stages: first, the PS sends the current model to the $N-1$ edge devices; then, each device updates this model
via \emph{local gradient descent}; finally, all devices simultaneously send their updates to the PS using AirComp with channel inversion,
so that the latter updates the model for the next round. Like \cite{9563232}, this approach relies on instantaneous channel reciprocity and CSI, 
obtained by broadcasting an OFDM pilot symbol in the downlink.
\\
$\bullet$ \underline{\emph{Local only}}: Nodes optimize independently using their datasets, without inter-node communication. This scheme illustrates the need for communication to solve the ML task: it results in a 90\% test error, as shown in Fig. \ref{fig:SoAvsN}.c.

In Fig. \ref{fig:SoAvsN}, we compare these schemes versus the number of nodes $N$, after $T_{sim}{=}$2000ms of execution time.
Note that the number of iterations completed by a certain scheme is $\sim T_{sim}/T$, function of its frame duration $T$, shown in Fig. \ref{fig:SoAvsN}.d.
In addition to the normalized error (a), we also evaluate the \emph{suboptimality gap} of the average model
$\bar{\mathbf w}=\frac{1}{N}\sum_{i}\mathbf w_{i}$, $F(\bar{\mathbf w})-F(\mathbf w^*)$ (b)
and the \emph{test error} $\mathrm{TEST}(\bar{\mathbf w})$, computed on a test set of  1000 examples (100 for each class).
For a parameter vector $\mathbf w^\top{=}[\mathbf w^{(1)\top},\dots,\mathbf w^{(9)\top}]\in\mathbb R^d$ and $\mathbf w^{(0)}{=}\mathbf 0$, we predict
 the class associated to feature vector $\mathbf d$ as $\arg\max_{c}\bar{\mathbf w}^{(c)\top}\mathbf f$.
 The results are further averaged over 20
algorithm trajectories generated through independent realizations of network deployment, channels, AWGN noise,
and randomness used by each algorithm.

Overall, we observe that smaller normalized error typically translates to smaller suboptimality gap and test error.
\proposed\ performs the best when $N{\geq}100$, under both spatially-i.i.d. and -dependent label scenarios.
In fact, it executes the most iterations (Fig. \ref{fig:SoAvsN}.d) within the execution time $T_{sim}$, revealing its scalability to dense network deployments; in contrast, all other schemes, with the exception of \emph{AirComp-FL} and \emph{AirComp-D2D}, execute less iterations as $N$ increases, due to the increased scheduling overhead. For smaller networks ($N<100$), 
\emph{QDGD-LPQ} performs the best: in this regime, it schedules all devices within a small frame duration, yet it has 
better noise control through the use of digital transmissions, and avoids interference via OFDMA scheduling.

\emph{AirComp-D2D} performs well across $N$ in the spatially-i.i.d. label scenario (between 3k and 6k iterations across all $N$, see Fig.  \ref{fig:SoAvsN}.d), thanks to its efficient scheduling and the use of AirComp allowing simultaneous transmissions \cite{9563232}. Yet, it performs poorly in the spatially-dependent label scenario, due to the slower  propagation of information across the network, noted in the evaluations of 
\secref{compproposed}. With \emph{AirComp-D2D}, information propagation is further hampered by the D2D connectivity structure, so that it takes several iterations for the information generated by the nodes at the edge to propagate through the mesh network to the rest of the network.
In contrast, \proposed\ does not rely on a predetermined D2D mesh network: it exploits the channel propagation conditions in the consensus phase,
resulting in a much smaller performance degradation.

\emph{AirComp-FL} also exhibits scalability with respect to $N$ through AirComp. Yet, its effectiveness is limited by channel estimation errors and the star topology's vulnerability to path loss, especially affecting edge devices. This results in a communication bottleneck: uplink transmissions must accommodate the worst channel conditions to satisfy power constraints during channel inversion;
in the downlink, transmission rates are reduced
to guarantee reception by the edge devices. Conversely, \proposed\ is topology-agnostic and 
leverages the local connectivity structure, achieving more robust performance for edge devices.


\vspace{-4mm}
   \section{Conclusions}
\label{conclu}
   This paper presents a novel DGD algorithm tailored to wireless systems, addressing the challenges posed by noise, fading, and limited bandwidth  without requiring inter-agent coordination, topology, or channel state information. Our approach, centered around a Non-Coherent Over-The-Air (NCOTA) consensus mechanism, leverages a noisy energy superposition property of wireless channels, allowing simultaneous, uncoordinated transmissions. This novel method achieves efficient consensus estimation without explicit mixing weights, exploiting path loss and adapting to a variety of fading and channel conditions. 
     We prove that the error of \proposed\ vanishes 
  with rate $\mathcal O(1/\sqrt{k})$ after $k$ iterations, using suitably tuned consensus and learning stepsizes.
Our results indicate a promising direction for decentralized optimization in wireless environments, showcasing faster convergence and operational efficiency,
especially in dense networks. 
There are still open challenges in designing communication-efficient decentralized learning algorithms for wireless systems. These challenges include  non-convex objectives and high model dimensionality, both inherent to deep neural network architectures. This work paves the way for further research in these areas.


\emph{Acknowledgments}: Special thanks to Prof. Gesualdo Scutari for valuable discussions  during manuscript preparation.

\appendices
\renewcommand\thesubsection{\thesection.\Roman{subsection}}

\def\thesubsectiondis{Appendix B.\Roman{subsection}:} 

\section*{Appendix A: Bound on the SGD variance}
\begin{lemma}
\label{Lsgd}
Assume $|\mathcal D_i|{=}D{\geq}2,\ \forall i$ and minibatch gradients with minibatch size $|\mathcal B_{ik}|{=}B{\in}\{1,\dots, D\},\forall i,k$.
Assume that: \{1\} the  loss function $\phi(\boldsymbol{\xi};\mathbf w)$  is $\mu$ strongly-convex and $L$ smooth with respect to $\mathbf w$ (implying Assumption \ref{fiassumption});
\{2\} $\Vert\nabla\phi(\boldsymbol{\xi};\mathbf w^*)\Vert{\leq}\nabla^*,\forall\boldsymbol{\xi}{\in}\mathcal D_i,\forall i$ (implying Def. \ref{gradbo}).
 Then,
$$
\frac{1}{N}\mathbb E[\Vert\e_k^{(2)}\Vert^2|\mathcal F_k]\leq
 \frac{D-B}{B(D-1)}(\nabla^*+L\cdot\mathrm{dm}(\mathcal W))^2\triangleq \Sigma^{(2)}.
$$
\end{lemma}
\begin{proof} It directly follows from \cite[Prop. 1]{Wu2020OnTN}, along with
$\Vert\nabla\phi(\boldsymbol{\xi};\mathbf w)\Vert
\leq \Vert\nabla\phi(\boldsymbol{\xi};\mathbf w^*)\Vert+\Vert\nabla\phi(\boldsymbol{\xi};\mathbf w)-\phi(\boldsymbol{\xi};\mathbf w^*)\Vert$\\
$
\leq
\nabla^*+L\Vert\mathbf w-\mathbf w^*\Vert
\leq \nabla^*+L\cdot\mathrm{dm}(\mathcal W).
$
\end{proof}

 \section*{Appendix B: Proof of Theorem \ref{Tmain}}
 \label{proofofmainT}
 We prove \eqref{L1}-\eqref{L4} in appendices B\ref{proofofL1}-B\ref{proofofL2},
 using the auxiliary lemmas in B\ref{aux}. 
 We denote the set of semidefinite-positive $n\times n$ matrices with eigenvalues in the interval $[\mu,L]$ as
 $\mathcal S_{\mu,L}^n$.
 We start with two important properties.
 
 {\bf P1:} 
For any twice-differentiable function $f:\mathbb R^n\mapsto\mathbb R$, $\mu$-strongly convex and $L$-smooth, 
we can express $\nabla f(\mathbf y)$ as the line integral of the Hessian matrix $\nabla^2 f$ between $\mathbf x$ and $\mathbf y$:
\begin{align}
\nonumber
 \nabla f(\mathbf y)
&=\nabla f(\mathbf x)+\int_0^1\nabla^2 f(\mathbf x+t(\mathbf y-\mathbf x))\mathrm dt(\mathbf y-\mathbf x),
\\&
\label{multiv}
=\nabla f(\mathbf x)+\mathbf A(\mathbf y-\mathbf x),
\end{align}
 where $\mathbf A$ (function of $\mathbf x$, $\mathbf y$)  is the \emph{mean Hessian} matrix between points $\mathbf x$ and $\mathbf y$.
  Since $f$ is $\mu$-strongly convex and $L$-smooth, 
it follows that   $\nabla^2 f{\in}\mathcal S_{\mu,L}^{n}$,
 hence $\mathbf A{\in}\mathcal S_{\mu,L}^{n}$.

{\bf P2:} For any function $h{:}\mathbb R^n{\mapsto}\mathbb R$,  $\mu_h$-strongly convex and $L_h$-smooth ($L_h{\geq}\mu_h$), 
any $\mathbf x,\mathbf y{\in}\mathbb R^n$,
 and stepsize $\eta\in[0,\frac{2}{\mu_h+L_h}]$,
\begin{align}
\label{nestbound}
\Vert (\mathbf x-\mathbf y)-\eta(\nabla h(\mathbf x)-\nabla h(\mathbf y))\Vert
\leq
(1-\eta\mu_h)\Vert\mathbf x-\mathbf y\Vert.
\end{align}
To show this, use \eqref{multiv} to express
$\nabla h(\mathbf y){=}\nabla h(\mathbf x){+}\mathbf A(\mathbf y{-}\mathbf x)$,
for some  $\mathbf A\in\mathcal S_{\mu_h,L_h}^{n}$,
 hence
\\\centerline{$
 (\mathbf x-\mathbf y)-\eta(\nabla h(\mathbf x)-\nabla h(\mathbf y))
=
 (\mathbf I-\eta\mathbf A)(\mathbf x-\mathbf y).
$}
Since $\mathbf A\in\mathcal S_{\mu_h,L_h}^{n}$, we further bound
$$
\Vert (\mathbf I-\eta\mathbf A)(\mathbf x-\mathbf y)\Vert^2
\leq
\max_{\rho\in[\mu_h,L_h]}(1-\eta\rho)^2
\Vert \mathbf x-\mathbf y\Vert^2.
$$
It is straightforward to show that the right-hand-side is maximized by
$\rho{=}\mu_h$ under the stepsize assumption, yielding \eqref{nestbound}.
 \subsection{Proof of  \eqref{L1}}
 \label{proofofL1}
Using the non-expansive property of projections (see, for instance, \cite{bertsekas2003convex}), we can bound
$
\Vert\mathbf W_{k+1}{-}\overline{\mathbf W}_{k+1}\Vert^2
$$$
{\leq}
\Vert\mathbf W_{k}{-}\overline{\mathbf W}_{k}
-\eta_k(\nabla G_k(\mathbf W_{k})-\nabla G_k(\overline{\mathbf W}_{k}))
+\gamma_k\e_k^{(1)}-\eta_k\e_k^{(2)}
\Vert^2.
$$
Taking the expectation conditional on $\mathcal F_k$ and using the independence of $\e_k^{(1)}$ and $\e_k^{(2)}$, it follows
$\mathbb E[\Vert\mathbf W_{k+1}{-}\overline{\mathbf W}_{k+1}\Vert^2|\mathcal F_k]$
\begin{align}
{\leq}&
\Vert(\mathbf W_{k}-\overline{\mathbf W}_{k})
{-}\eta_k(\nabla G_k(\mathbf W_{k}){-}\nabla G_k(\overline{\mathbf W}_{k}))
\Vert^2\label{dfgghs}\\&
{+}\gamma_k^2N\Sigma^{(1)}
{+}\eta_k^2N\Sigma^{(2)}.
\label{fgh}
\end{align}
Note that Assumption \ref{fiassumption} implies that $G_k$ is $\mu$-strongly convex and $L_{Gk}\triangleq L+\frac{\gamma_k}{\eta_k}\rho_N$-smooth.
Then, using \eqref{nestbound}, we find that
\eqref{dfgghs} 
$\leq(1-\mu\eta_k)^2\Vert\mathbf W_{k}-\overline{\mathbf W}_{k}\Vert^2$,
as long as $\eta_k\leq 2/(\mu+L_{Gk})$ ({\bf C1} of Theorem \ref{Tmain}).
Using this bound and computing the unconditional expectation, it follows 
$\Vert\mathbf W_{k+1}-\overline{\mathbf W}_{k+1}\Vert_{\mathbb E}^2$
$$
\leq
(1-\mu\eta_k)^2\Vert\mathbf W_{k}-\overline{\mathbf W}_{k}\Vert_{\mathbb E}^2
+\gamma_k^2N\Sigma^{(1)}+\eta_k^2N\Sigma^{(2)}.
$$
 The result follows by induction and noting that $\mathbf W_{\bar{\kappa}}=\overline{\mathbf W}_{\bar{\kappa}}$.
\vspace{-5mm}
\subsection{Proof of \eqref{L4}}
\label{proofofL4}
Let 
$
\hat {\mathbf W}_k\triangleq \argmin_{\mathbf W\in\mathbb R^{Nd}}G_k(\mathbf W)
$
be the \emph{unconstrained} minimizer of $G_k$;\footnote{Here, we assume that all $f_i$'s are defined on $\mathbb R^d$, and the extended functions are $L$-smooth and $\mu$-strongly convex on $\mathbb R^d$}
this is in contrast to $\mathbf W_k^*$, which minimizes it over $\mathcal W^N$.
In the last part of the proof, we will show that, when  $\frac{\eta_{k}}{\gamma_{k}}{\leq}\frac{\zeta\cdot \mu\rho_2}{\sqrt{N}\nabla^* L}$  ({\bf C2} of Theorem~\ref{Tmain}),
$\mathbf W_k^*$ and $\hat {\mathbf W}_k$ coincide. 
Since $G_k$ is strongly-convex and smooth on $\mathbb R^{Nd}$, $\hat {\mathbf W}_k$ is the unique solution of
$\nabla G_k(\hat {\mathbf W}_k)=\mathbf 0$ in $\mathbb R^{Nd}$. We now bound the distance of $\hat {\mathbf W}_k$ from $\mathbf W^*$.
From \eqref{multiv} with $\mathbf x=\mathbf W^*$, $\mathbf y=\hat {\mathbf W}_k$, there exists $\mathbf A\in\mathcal S_{\mu,L}^{Nd}$ such that 
\\\centerline{$
\nabla f(\hat {\mathbf W}_k)=\nabla f(\mathbf W^*)+\mathbf A(\hat {\mathbf W}_k-\mathbf W^*).
$}
  It follows that $\mathbf 0=\nabla G_k(\hat {\mathbf W}_k)-\frac{\gamma_k}{\eta_{k}}\LL\mathbf W^*$
\\\centerline{$
=\nabla f(\hat {\mathbf W}_k)+\frac{\gamma_k}{\eta_{k}}\LL(\hat {\mathbf W}_k-\mathbf W^*)
=
\nabla f(\mathbf W^*)+\mathbf B(\hat {\mathbf W}_k-\mathbf W^*),
$}
where we used 
$\LL\mathbf W^*=(\boldsymbol{L}\otimes\mathbf I_d)\cdot(\mathbf 1_N\otimes\mathbf w^*)
=
(\boldsymbol{L}\cdot\mathbf 1_N)\otimes\mathbf w^*=\mathbf 0$, and
 defined
$\mathbf B\triangleq\mathbf A+\frac{\gamma_k}{\eta_k}\LL\succeq\mu\mathbf I$.
Solving with respect to $\hat {\mathbf W}_k-\mathbf W^*$ and computing the norm of both sides, we find
\begin{align}
\Vert\hat {\mathbf W}_k-\mathbf W^*\Vert
=
\Vert\mathbf B^{-1}\nabla f(\mathbf W^*)\Vert.
\label{adsh}
\end{align}
Note that $\nabla f(\mathbf W^*)\bot(\mathbf 1_N\otimes\mathbf I_d)$ (in fact, $\sum_i\nabla f_i(\mathbf w^*)=\mathbf 0$ from the optimality condition on $\mathbf w^*$, since $\mathbf w^*\in\mathrm{int}(\mathcal W)$).
Hence, we  bound \eqref{adsh}
via Lemma  \ref{L5} in B\ref{aux} with $U=\frac{\gamma_k}{\eta_k}$ as
\ba{
\Vert\hat {\mathbf W}_k-\mathbf W^*\Vert
\leq
\frac{L}{\mu\rho_2}
\frac{\eta_k}{\gamma_k}\Vert\nabla f(\mathbf W^*)\Vert
\leq
\frac{\sqrt{N}\nabla^* L}{\mu\rho_2}
\frac{\eta_k}{\gamma_k},
}{dfg}
where in the last step we used
 $\Vert\nabla f(\mathbf W^*)\Vert{\leq}$
 $\sqrt{\sum_i\Vert\nabla f_i(\mathbf w^*)\Vert^2}{\leq}\sqrt{N}\nabla^*$
 from Definition \ref{gradbo}.

Next, we show that $\hat {\mathbf W}_k\in\mathcal W^N$, hence it coincides with the solution of the constrained problem ($\hat {\mathbf W}_k= {\mathbf W}_k^*$). Since $\mathbf w^*$ is at distance $\zeta$ from the boundary of $\mathcal W$ (Assumption \ref{distance}), it suffices to show that $\Vert\hat {\mathbf w}_{ik}-\mathbf w^*\Vert\leq\zeta,\ \forall i$ (i.e., $\mathbf w^*$ is closer to $\hat {\mathbf w}_{ik}$ than to the boundary of $\mathcal W$).
Indeed,
$
\Vert\hat {\mathbf w}_{ik}-\mathbf w^*\Vert
\leq
\Vert\hat {\mathbf W}_k-\mathbf W^*\Vert\leq 
\frac{\sqrt{N}\nabla^* L}{\mu\rho_2}
\frac{\eta_{k}}{\gamma_{k}},
$
hence, when $\frac{\eta_{k}}{\gamma_{k}}\leq \frac{\zeta \cdot \mu\rho_2}{\sqrt{N}\nabla^* L}$ ({\bf C2} of Theorem \ref{Tmain}),
 it follows that $\hat {\mathbf W}_k= {\mathbf W}_k^*$ and
\begin{align}
\label{gradeq0}
\nabla G_k({\mathbf W}_k^*)=\mathbf 0,\end{align}
and \eqref{L4} readily follows from \eqref{dfg}.
\vspace{-5mm}
\subsection{Proof of  \eqref{L2}}
\label{proofofL2}
Using the triangle inequality, we bound
\ba{
\Vert\overline{\mathbf W}_{k+1}{-}\mathbf W_{k+1}^*\Vert
{\leq}
\Vert\overline{\mathbf W}_{k+1}-\mathbf W_{k}^*\Vert
+\Vert\mathbf W_{k+1}^*-\mathbf W_{k}^*\Vert.
}{sdfghdasdf}
Using the fixed point optimality condition $\mathbf W_{k}^*=\Pi^N[\mathbf W_{k}^*
-\eta\nabla G_k(\mathbf W_{k}^*)],\forall \eta\geq 0$ and the non-expansive property of projections \cite{bertsekas2003convex},
we bound the first term of \eqref{sdfghdasdf} as
$\Vert\overline{\mathbf W}_{k+1}-\mathbf W_{k}^*\Vert$
\ba{
&\nonumber
=\Vert \Pi^N[\overline{\mathbf W}_{k}
-\eta_k\nabla G_{k}(\overline{\mathbf W}_{k})]
-\Pi^N[{\mathbf W}_{k}^*
-\eta_k\nabla G_{k}({\mathbf W}_{k}^*)]\Vert
\\&\nonumber
\leq
\Vert \overline{\mathbf W}_{k}-{\mathbf W}_{k}^*
-\eta_k(\nabla G_{k}(\overline{\mathbf W}_{k})
-\nabla G_{k}({\mathbf W}_{k}^*))\Vert\\&
\leq(1-\mu\eta_k)\Vert\overline{\mathbf W}_{k}-{\mathbf W}_{k}^*\Vert,
}{int41}
where in the last step we used
\eqref{nestbound} since $G_k$ is $\mu$-strongly convex and $L_{Gk}\triangleq L+\frac{\gamma_k}{\eta_k}\rho_N$-smooth,
 provided that $\eta_k\leq 2/(\mu+L_{Gk})$ ({\bf C1} of Theorem \ref{Tmain}).
We now bound  $\Vert\mathbf W_{k+1}^*{-}\mathbf W_{k}^*\Vert$.
First, note that,
for $t{\geq}\bar{\kappa}$,
$\mathbf W_{t}^*$ satisfies the optimality condition
 $\nabla G_t(\mathbf W_{t}^*)=\mathbf 0$ (see \eqref{gradeq0}).
 It then follows that
\ba{
\label{e1}
&\mathbf 0{=}\nabla G_{k+1}(\mathbf W_{k+1}^*){=}\nabla f(\mathbf W_{k+1}^*){+}\frac{\gamma_{k+1}}{\eta_{k+1}}\LL\mathbf W_{k+1}^*,\!\!
\\
\label{e2}
&\mathbf 0=\nabla G_{k}(\mathbf W_{k}^*)=\nabla f(\mathbf W_{k}^*)+\frac{\gamma_{k}}{\eta_{k}}\LL\mathbf W_{k}^*.
}{}
Furthermore,
from \eqref{multiv} with $\mathbf x=\mathbf W_{k}^*$ and $\mathbf y=\mathbf W_{k+1}^*$,
there exists  $\mathbf A\in\mathcal S_{\mu,L}^{Nd}$ such that 
$\nabla f(\mathbf W_{k+1}^*)=\nabla f(\mathbf W_{k}^*)+\mathbf A(\mathbf W_{k+1}^*-\mathbf W_{k}^*)$. 
We use this expression in \eqref{e1}, subtract \eqref{e2}, reorganize and solve. These steps yield
$$
\mathbf W_{k+1}^*-\mathbf W_{k}^*
=-\Big(\frac{\gamma_{k+1}}{\eta_{k+1}}-\frac{\gamma_{k}}{\eta_{k}}\Big)\mathbf B^{-1}\LL\mathbf W_{k+1}^*,
$$
where 
$\mathbf B\triangleq \mathbf A+\frac{\gamma_k}{\eta_k}\LL.$
 Therefore, 
\ba{
\Vert\mathbf W_{k+1}^*{-}\mathbf W_{k}^*\Vert
\leq
\Big(\frac{\gamma_{k+1}}{\eta_{k+1}}{-}\frac{\gamma_{k}}{\eta_{k}}\Big)
\Vert\mathbf B^{-1}\LL\mathbf W_{k+1}^*\Vert.
}{int4}
Noting that
$\LL\mathbf W_{k+1}^*\bot(\mathbf 1_N\otimes\mathbf I_d)$ (since $\boldsymbol{L}\cdot\mathbf 1=\mathbf 0$), we
use Lemma~\ref{L5} in  B\ref{aux} with $U=\frac{\gamma_{k}}{\eta_{k}}$
 to bound
$$
\Vert\mathbf W_{k+1}^*-\mathbf W_{k}^*\Vert
\leq
\frac{L}{\mu\rho_2}\Big(\frac{\gamma_{k+1}}{\eta_{k+1}}-\frac{\gamma_{k}}{\eta_{k}}\Big)\frac{\eta_{k}}{\gamma_{k}}
\Vert\LL\mathbf W_{k+1}^*\Vert.
$$
Finally, we bound $\Vert\LL\mathbf W_{k+1}^*\Vert$
via Lemma \ref{L3} in Appendix B\ref{aux}:
$$
\Vert\mathbf W_{k+1}^*-\mathbf W_{k}^*\Vert
\leq
\Big(\frac{\eta_{k}}{\gamma_{k}}-\frac{\eta_{k+1}}{\gamma_{k+1}}\Big)
\frac{\sqrt{N}\nabla^* L}{\mu\rho_2}
\Big(1+\frac{L^2}{\mu\rho_2}
\frac{\eta_{k}}{\gamma_{k}}\Big).
$$
Eq. \eqref{L2} follows by
combining this result with \eqref{int41} into \eqref{sdfghdasdf}, and solving via induction
with 
$\Vert\overline{\mathbf W}_{\bar{\kappa}}{-}\mathbf W_{\bar{\kappa}}^*\Vert
\leq\sqrt{N}\mathrm{dm}(\mathcal W)$.
\vspace{-6mm}
 \subsection{Auxiliary lemmas}
 \label{aux}
\begin{lemma}\label{L5}
Let
$
\mathbf B{=}\mathbf A{+}U\LL,
$
where $U>0$, 
$\mathbf A\in\mathcal S_{\mu,L}^{Nd}$. Then,
$$
\Vert\mathbf B^{-1}\mathbf v\Vert\leq
\frac{1}{U}\cdot\frac{L}{\mu\rho_2}\Vert\mathbf v\Vert,\ \forall \mathbf v\bot(\mathbf 1_N\otimes\mathbf I_d).
$$
\end{lemma}
\begin{proof}
Using $\mathbf A\succeq\mu\mathbf I$, we start with bounding
\begin{align}
\Vert\mathbf B^{-1}\mathbf v\Vert=
\Vert\mathbf A^{-1}\mathbf A\mathbf B^{-1}\mathbf v\Vert\leq
\frac{1}{\mu}\Vert\mathbf A\mathbf B^{-1}\mathbf v\Vert.
\label{bounsdg}
\end{align}
Let $\boldsymbol{L}{=}\mathbf V\boldsymbol{\Lambda}\mathbf V^\top$ be the eigenvalue decomposition of the Laplacian matrix, where
$\mathbf V$ is unitary and $\boldsymbol{\Lambda}$ is diagonal. 
Since $\boldsymbol{L}\cdot\mathbf 1=0$,
 the first column of  $\mathbf V$ equals $\mathbf 1_N/\sqrt{N}$ and $[\boldsymbol{\Lambda}]_{1,1}{=}0$; furthermore, 
 $[\boldsymbol{\Lambda}]_{i,i}{\geq}\rho_2>0,\forall i>1$ (Assumption \ref{algconnofL}).
Then, $\LL{=}\boldsymbol{L}\otimes\mathbf I_d{=}\hat{\mathbf V}\hat{\boldsymbol{\Lambda}}\hat{\mathbf V}^\top$,
where
$\hat{\mathbf V}{=}{\mathbf V}{\otimes}\mathbf I_d$ and $\hat{\boldsymbol{\Lambda}}{=}{\boldsymbol{\Lambda}}{\otimes}\mathbf I_d$.
Let 
$$\tilde{\mathbf A}=\hat{\mathbf V}^\top\mathbf A\hat{\mathbf V}
=
\begin{bmatrix}
\tilde{\mathbf A}_{1,1} & \tilde{\mathbf A}_{1,2}\\
\tilde{\mathbf A}_{1,2}^\top & \tilde{\mathbf A}_{2,2}
\end{bmatrix},\ 
\hat{\boldsymbol{\Lambda}}
=
\begin{bmatrix}
\mathbf 0 & \mathbf 0\\
\mathbf 0 & \hat{\boldsymbol{\Lambda}}_{2,2}
\end{bmatrix},$$
where $\tilde{\mathbf A}_{1,1}\in\mathbb R^{d\times d}$, $\tilde{\mathbf A}_{1,2}\in\mathbb R^{d\times (N-1)d}$,
$\tilde{\mathbf A}_{2,2}$
and $\hat{\boldsymbol{\Lambda}}_{2,2}\in\mathbb R^{(N-1)d\times (N-1)d}$ with $\hat{\boldsymbol{\Lambda}}_{2,2}\succeq \rho_2\mathbf I$.
Since
$\mathbf v\bot(\mathbf 1_N\otimes\mathbf I_d)$, we also define
$\hat{\mathbf V}^\top\mathbf v=[\mathbf 0^\top,\mathbf x^\top]^\top$ for suitable $\mathbf x$. Therefore,
$$
\Vert\mathbf A\mathbf B^{-1}\mathbf v\Vert
=\Big\Vert\tilde{\mathbf A}(\tilde{\mathbf A}{+}U\hat{\boldsymbol{\Lambda}})^{-1}\begin{bmatrix}
\mathbf 0 \\
\mathbf x
\end{bmatrix}\Big\Vert
.
$$
Using the inversion properties of $2\times 2$ block matrices and simplifying, we continue as $\Vert\mathbf A\mathbf B^{-1}\mathbf v\Vert$
$$
=
\Big\Vert
\begin{bmatrix}
\tilde{\mathbf A}_{1,1} & \tilde{\mathbf A}_{1,2}\\
\tilde{\mathbf A}_{1,2}^\top & \tilde{\mathbf A}_{2,2}
\end{bmatrix}
\begin{bmatrix}
\bullet & -\tilde{\mathbf A}_{1,1}^{-1}\tilde{\mathbf A}_{1,2}\mathbf S_U^{-1}\\
\bullet & \mathbf S_U^{-1}
\end{bmatrix}
\begin{bmatrix}
\mathbf 0 \\
\mathbf x
\end{bmatrix}
\Big\Vert
=\Vert\mathbf S_0\mathbf S_U^{-1}\mathbf x\Vert,
$$
where $\mathbf S_z{\triangleq}
\tilde{\mathbf A}_{2,2}{-}\tilde{\mathbf A}_{1,2}^\top\tilde{\mathbf A}_{1,1}^{-1}\tilde{\mathbf A}_{1,2}{+}z\hat{\boldsymbol{\Lambda}}
$ is the Schur complement of block $\tilde{\mathbf A}_{1,1}$ of $\tilde{\mathbf A}{+}z\hat{\boldsymbol{\Lambda}}$,
and the block elements $\bullet$ are irrelevant, being then multiplied by $\mathbf 0$.
Using the facts that $\mathbf 0{\prec}\mathbf S_0{\preceq} L\cdot\mathbf I$ ($\mathbf S_0^{-1}$ is the lower-diagonal block
of  $\tilde{\mathbf A}^{-1}$, and $\tilde{\mathbf A}^{-1}{\succeq}\frac{1}{L}\mathbf I$)
and $\mathbf S_U{=}\mathbf S_0{+}U\hat{\boldsymbol{\Lambda}}{\succeq}U\rho_2\mathbf I$, yield
$$
\Vert\mathbf A\mathbf B^{-1}\mathbf v\Vert=
\Vert\mathbf S_0\mathbf S_U^{-1}\mathbf x\Vert
\leq
L\Vert\mathbf S_U^{-1}\mathbf x\Vert
\leq
\frac{L}{U\rho_2}\Vert\mathbf x\Vert=\frac{L}{U\rho_2}\Vert\mathbf v\Vert.
$$
The final result is obtained by using this bound in \eqref{bounsdg}.
\end{proof}
\begin{lemma}\label{L3}
$\Vert\LL\mathbf W_k^*\Vert\leq\frac{\eta_k}{\gamma_k}\sqrt{N}\nabla^*\Big(1+\frac{L^2}{\mu\rho_2}
\frac{\eta_{k}}{\gamma_{k}}\Big),
\forall k\geq\bar{\kappa}.
$
\end{lemma}
\begin{proof}
From the proof of \eqref{L4} in Appendix B\ref{proofofL4}, 
$\nabla G_k(\mathbf W_k^*)=\nabla f(\mathbf W_k^*)+\frac{\gamma_k}{\eta_k}\LL\mathbf W_k^*=\mathbf 0$ for $k\geq\bar{\kappa}$
(see \eqref{gradeq0}).
It follows that 
$$\Vert\LL\mathbf W_k^*\Vert=\frac{\eta_k}{\gamma_k}\Vert\nabla f(\mathbf W_k^*)\Vert.$$
Moreover, using the triangle inequality and smoothness of $f$,
$$
\Vert\nabla f(\mathbf W_k^*)\Vert
\leq \Vert\nabla f(\mathbf W_k^*)-\nabla f(\mathbf W^*)\Vert
+\Vert\nabla f(\mathbf W^*)\Vert
$$$$
\leq
L\Vert\mathbf W_k^*-\mathbf W^*\Vert+\Vert\nabla f(\mathbf W^*)\Vert.
$$
The final result is obtained by bounding 
 $\Vert\mathbf W_k^*-\mathbf W^*\Vert$ as in  \eqref{L4},
and using $\Vert\nabla f(\mathbf W^*)\Vert\leq \sqrt{N}\nabla^*$.
\end{proof}
\vspace{-5mm}
\section*{Appendix C: Proof of Theorem \ref{T1}}
\label{proofofT1} 
\begin{proof}
First, note that the choice of stepsizes guarantees the existence of $\bar{\kappa}\geq 0$ (finite)
such that {\bf C1-C3} are satisfied for  $k\geq\bar{\kappa}$. Consider $k\geq\bar{\kappa}$.
In the proof, we will often need
$$
S_k^n(\upsilon)\triangleq \sum_{t=\bar{\kappa}}^{k-1}P_{tk}^n(1+\delta t)^{-\upsilon},\ \text{for }1\leq\upsilon\leq \frac{5}{4}n\text{ and }n\geq 1.
$$
To bound $S_k^n(\upsilon)$, we first bound $P_{tk}$ for $t\geq\bar{\kappa}-1 $ as
\begin{align}
&P_{tk}=\prod_{j=t+1}^{k-1}(1-\mu\eta_j)\leq
e^{-\mu\eta_0 \sum_{j=t+1}^{k-1}(1+\delta j)^{-1}}
\nonumber\\
\label{ptkbound}
&{\leq}
e^{-\mu\eta_0 \int_{t+1}^k(1+\delta x)^{-1}\mathrm dx}
{=}
(1{+}\delta k)^{-\frac{\mu\eta_0}{\delta}}(1{+}\delta(t{+}1))^{\frac{\mu\eta_0}{\delta}}.
\end{align}
For $t\geq\bar{\kappa}$,
we further bound $(1+\delta(t+1))^{\frac{\mu\eta_0}{\delta}}$
$$
=
(1{+}\delta t)^{\frac{\mu\eta_0}{\delta}}\Big(1{+}\frac{\delta}{1{+}\delta t}\Big)^{\frac{\mu\eta_0}{\delta}}
\leq
(1{+}\delta t)^{\frac{\mu\eta_0}{\delta}}\Big(1{+}\frac{\delta}{1{+}\delta\bar{\kappa}}\Big)^{\frac{\mu\eta_0}{\delta}};
$$
using $1+x\leq e^x$, we further bound the last term  above as
$$
\Big(1{+}\frac{\delta}{1{+}\delta\bar{\kappa}}\Big)^{\frac{\mu\eta_0}{\delta}}
\leq
 e^{\mu\eta_0/(1{+}\delta\bar{\kappa})}=e^{\mu\eta_{\bar{\kappa}}}\leq e,
$$
where in the last step we used $\eta_{\bar{\kappa}}\leq2/(\mu+L)\leq 1/\mu$ from {\bf C1} of Theorem \ref{Tmain}.
It follows that
$P_{tk}\leq e(1{+}\delta k)^{-\frac{\mu\eta_0}{\delta}}(1{+}\delta t)^{\frac{\mu\eta_0}{\delta}}
$ for $t\geq\bar{\kappa}$. Using this bound in $S_k^n(\upsilon)$, we obtain
$$
S_k^n(\upsilon)\leq e^n(1{+}\delta k)^{-n\frac{\mu\eta_0}{\delta}}
\sum_{t=\bar{\kappa}}^{k-1}(1{+}\delta t)^{n\frac{\mu\eta_0}{\delta}-\upsilon}.
$$
Since $n\frac{\mu\eta_0}{\delta}\geq5/4n\geq\upsilon$,
we continue as
\begin{align}
\nonumber
&S_k^n(\upsilon)
\leq
e^n(1+\delta k)^{-n\frac{\mu\eta_0}{\delta}}
\int_{\bar{\kappa}}^k
(1+\delta x)^{n\frac{\mu\eta_0}{\delta}-\upsilon}\mathrm dx
\\&
\!\!\!{\leq}
\frac{e^n}{n\mu\eta_0-\delta(\upsilon-1)}(1+\delta k)^{-\upsilon+1}
{\leq}
\frac{5 e^n}{4\mu\eta_0}(1+\delta k)^{-\upsilon+1},\!\label{skupbound}
\end{align}
where the last step used $\delta \leq \frac{4}{5}\mu\eta_0$ and $\upsilon\leq 5/4n$.

With $\gamma_k,\eta_k$ as in Theorem \ref{T1}, we specialize \eqref{L1} as
\begin{align*}
&\frac{1}{{N}}\Vert{\mathbf W}_{k}-\overline{\mathbf W}_{k}\Vert_{\mathbb E}^2
\leq
\gamma_0^2\Sigma^{(1)} S_k^2(3/2)
+\eta_0^2 \Sigma^{(2)} S_k^2(2).
\end{align*}
The bound in  \eqref{thm1_1} follows using 
  \eqref{skupbound} and $\sqrt{a+b}\leq\sqrt{a}+\sqrt{b}$ for $a,b\geq 0$.
Similarly,  we specialize  \eqref{L2} as
$$
\frac{1}{\sqrt{N}}\Vert\overline{\mathbf W}_{k}-{\mathbf W}_{k}^*\Vert_{\mathbb E}{\leq}
\mathrm{dm}(\mathcal W)\cdot P_{(\bar{\kappa}-1)k}
+\frac{\nabla^* L}{\mu\rho_2}
\frac{\eta_0}{\gamma_0}
\sum_{t=\bar{\kappa}}^{k-1}P_{tk}
$$$$
\times\Big(1+\frac{L^2}{\mu\rho_2}\frac{\eta_0}{\gamma_0}(1+\delta t)^{-1/4}\Big)
\Big(
(1+\delta t)^{-1/4}-(1+\delta(t+1))^{-1/4}
\Big).
$$
Since $x^{-1/4}$ is convex for $x>0$, we  further bound
$(1{+}\delta(t{+}1))^{-1/4}\geq (1{+}\delta t)^{-1/4}
-\delta/4(1{+}\delta t)^{-5/4},$
hence
$$\frac{1}{\sqrt{N}}\Vert\overline{\mathbf W}_{k}-\mathbf W_{k}^*\Vert_{\mathbb E}
\leq
\mathrm{dm}(\mathcal W)\cdot P_{(\bar{\kappa}-1)k}
$$$$
+\frac{\nabla^* L}{4\mu\rho_2}
\frac{\eta_0 \delta}{\gamma_0}
\Big(
 S_k^1(5/4)
+\frac{L^2}{\mu\rho_2}
\frac{\eta_0 }{\gamma_0}
 S_k^1(3/2)\Big).
$$
The bound in \eqref{thm1_2} readily follows using \eqref{ptkbound}, \eqref{skupbound} and $\delta{\leq}\frac{4}{5}\mu\eta_0$. 
\eqref{thm1_3} follows by direct substitution into \eqref{L4}.
\end{proof}
\vspace{-5mm}
   \bibliographystyle{IEEEtran}
\bibliography{IEEEabrv,biblio}

\begin{thebibliography}{10}
\providecommand{\url}[1]{#1}
\csname url@samestyle\endcsname
\providecommand{\newblock}{\relax}
\providecommand{\bibinfo}[2]{#2}
\providecommand{\BIBentrySTDinterwordspacing}{\spaceskip=0pt\relax}
\providecommand{\BIBentryALTinterwordstretchfactor}{4}
\providecommand{\BIBentryALTinterwordspacing}{\spaceskip=\fontdimen2\font plus
\BIBentryALTinterwordstretchfactor\fontdimen3\font minus
  \fontdimen4\font\relax}
\providecommand{\BIBforeignlanguage}[2]{{%
\expandafter\ifx\csname l@#1\endcsname\relax
\typeout{** WARNING: IEEEtran.bst: No hyphenation pattern has been}%
\typeout{** loaded for the language `#1'. Using the pattern for}%
\typeout{** the default language instead.}%
\else
\language=\csname l@#1\endcsname
\fi
#2}}
\providecommand{\BIBdecl}{\relax}
\BIBdecl

\bibitem{ICC23}
N.~Michelusi, ``Decentralized federated learning via non-coherent over-the-air
  consensus,'' in \emph{IEEE International Conference on Communications}, 2023,
  pp. 3102--3107.

\bibitem{icassp24}
------, ``{CSI-Free Over-The-Air Decentralized Learning Over Frequency
  Selective Channels},'' in \emph{IEEE International Conference on Acoustics,
  Speech and Signal Processing (ICASSP)}, 2024, pp. 13\,076--13\,080.

\bibitem{6494683}
S.~Kar and J.~M. Moura, ``{Consensus + innovations distributed inference over
  networks: cooperation and sensing in networked systems},'' \emph{IEEE Signal
  Processing Magazine}, vol.~30, no.~3, pp. 99--109, 2013.

\bibitem{9224135}
T.~T. Doan, S.~T. Maguluri, and J.~Romberg, ``{Convergence Rates of Distributed
  Gradient Methods Under Random Quantization: A Stochastic Approximation
  Approach},'' \emph{IEEE Transactions on Automatic Control}, vol.~66, no.~10,
  pp. 4469--4484, 2021.

\bibitem{Nedic2018}
A.~Nedi\'{c}, J.-S. Pang, G.~Scutari, and Y.~Sun, \emph{Multi-agent
  Optimization}, 1st~ed.\hskip 1em plus 0.5em minus 0.4em\relax Springer, Cham,
  2018.

\bibitem{YANG2019278}
T.~Yang, X.~Yi, J.~Wu, Y.~Yuan, D.~Wu, Z.~Meng, Y.~Hong, H.~Wang, Z.~Lin, and
  K.~H. Johansson, ``{A survey of distributed optimization},'' \emph{Annual
  Reviews in Control}, vol.~47, pp. 278--305, 2019.

\bibitem{10103556}
Y.~Ji, G.~Scutari, Y.~Sun, and H.~Honnappa, ``{Distributed (ATC) Gradient
  Descent for High Dimension Sparse Regression},'' \emph{IEEE Transactions on
  Information Theory}, vol.~69, no.~8, pp. 5253--5276, 2023.

\bibitem{8870236}
G.~Zhu, Y.~Wang, and K.~Huang, ``{Broadband Analog Aggregation for Low-Latency
  Federated Edge Learning},'' \emph{IEEE Transactions on Wireless
  Communications}, vol.~19, no.~1, pp. 491--506, 2020.

\bibitem{Lian17}
X.~Lian, C.~Zhang, H.~Zhang, C.-J. Hsieh, W.~Zhang, and J.~Liu., ``Can
  decentralized algorithms outperform centralized algorithms? {A} case study
  for decentralized parallel stochastic gradient descent,'' in \emph{Proc. 31st
  NeurIPS}, Dec. 2017.

\bibitem{9475989}
Y.~Xiao, Y.~Ye, S.~Huang, L.~Hao, Z.~Ma, M.~Xiao, S.~Mumtaz, and O.~A. Dobre,
  ``{Fully Decentralized Federated Learning-Based On-Board Mission for UAV
  Swarm System},'' \emph{IEEE Communications Letters}, vol.~25, no.~10, pp.
  3296--3300, 2021.

\bibitem{8950073}
S.~Savazzi, M.~Nicoli, and V.~Rampa, ``{Federated Learning With Cooperating
  Devices: A Consensus Approach for Massive IoT Networks},'' \emph{IEEE
  Internet of Things Journal}, vol.~7, no.~5, pp. 4641--4654, 2020.

\bibitem{Nedic2009grad}
A.~Nedi\'{c} and A.~Ozdaglar, ``Distributed subgradient methods for multi-agent
  optimization,'' \emph{{IEEE} Trans. Autom. Control}, vol.~54, no.~1, pp.
  48--61, Jan. 2009.

\bibitem{Yuan2016}
K.~Yuan, Q.~Ling, and W.~Yin, ``{On the Convergence of Decentralized Gradient
  Descent},'' \emph{SIAM Journal on Optimization}, vol.~26, no.~3, pp.
  1835--1854, 2016.

\bibitem{9241497}
R.~Xin, S.~Pu, A.~Nedic, and U.~A. Khan, ``A general framework for
  decentralized optimization with first-order methods,'' \emph{Proceedings of
  the IEEE}, vol. 108, no.~11, pp. 1869--1889, 2020.

\bibitem{9782148}
N.~Michelusi, G.~Scutari, and C.-S. Lee, ``{Finite-Bit Quantization for
  Distributed Algorithms With Linear Convergence},'' \emph{IEEE Transactions on
  Information Theory}, vol.~68, no.~11, pp. 7254--7280, 2022.

\bibitem{Kajiyama2020}
Y.~Kajiyama, N.~Hayashi, and S.~Takai, ``Linear convergence of consensus-based
  quantized optimization for smooth and strongly convex cost functions,''
  \emph{IEEE Transactions on Automatic Control}, vol.~66, no.~3, pp.
  1254--1261, 2021.

\bibitem{Magnusson2020}
S.~{Magn\'usson}, H.~{Shokri-Ghadikolaei}, and N.~{Li}, ``{On Maintaining
  Linear Convergence of Distributed Learning and Optimization Under Limited
  Communication},'' \emph{{IEEE} Trans. Signal Process.}, vol.~68, pp.
  6101--6116, 2020.

\bibitem{9562482}
R.~Saha, S.~Rini, M.~Rao, and A.~J. Goldsmith, ``Decentralized optimization
  over noisy, rate-constrained networks: Achieving consensus by communicating
  differences,'' \emph{IEEE Journal on Selected Areas in Communications},
  vol.~40, no.~2, pp. 449--467, 2022.

\bibitem{9563232}
H.~Xing, O.~Simeone, and S.~Bi, ``{Federated Learning Over Wireless
  Device-to-Device Networks: Algorithms and Convergence Analysis},'' \emph{IEEE
  Journal on Selected Areas in Communications}, vol.~39, no.~12, pp.
  3723--3741, 2021.

\bibitem{9322286}
E.~Ozfatura, S.~Rini, and D.~G\"und\"uz, ``{Decentralized SGD with Over-the-Air
  Computation},'' in \emph{IEEE Global Communications Conference}, 2020.

\bibitem{9517780}
Y.~Shi, Y.~Zhou, and Y.~Shi, ``{Over-the-Air Decentralized Federated
  Learning},'' in \emph{IEEE International Symposium on Information Theory
  (ISIT)}, 2021, pp. 455--460.

\bibitem{9205230}
X.~Chen, D.~W.~K. Ng, W.~Yu, E.~G. Larsson, N.~Al-Dhahir, and R.~Schober,
  ``{Massive Access for 5G and Beyond},'' \emph{IEEE Journal on Selected Areas
  in Communications}, vol.~39, no.~3, pp. 615--637, 2021.

\bibitem{ROKADE2022110322}
K.~Rokade and R.~K. Kalaimani, ``{Distributed computation of fast consensus
  weights using ADMM},'' \emph{Automatica}, vol. 142, 2022.

\bibitem{9222206}
C.-S. Lee, N.~Michelusi, and G.~Scutari, ``Finite rate distributed
  weight-balancing and average consensus over digraphs,'' \emph{IEEE
  Transactions on Automatic Control}, vol.~66, no.~10, pp. 4530--4545, 2021.

\bibitem{8952884}
K.~Yang, T.~Jiang, Y.~Shi, and Z.~Ding, ``Federated learning via over-the-air
  computation,'' \emph{IEEE Transactions on Wireless Communications}, vol.~19,
  no.~3, pp. 2022--2035, 2020.

\bibitem{9014530}
M.~M. Amiri and D.~G\"und\"uz, ``Federated learning over wireless fading
  channels,'' \emph{IEEE Transactions on Wireless Communications}, vol.~19,
  no.~5, pp. 3546--3557, 2020.

\bibitem{9515709}
M.~M. Amiri, D.~G\"und\"uz, S.~R. Kulkarni, and H.~V. Poor, ``Convergence of
  federated learning over a noisy downlink,'' \emph{IEEE Transactions on
  Wireless Communications}, vol.~21, no.~3, pp. 1422--1437, 2022.

\bibitem{9382114}
M.~M. Amiri, T.~M. Duman, D.~G\"und\"uz, S.~R. Kulkarni, and H.~V. Poor,
  ``Blind federated edge learning,'' \emph{IEEE Transactions on Wireless
  Communications}, vol.~20, no.~8, pp. 5129--5143, 2021.

\bibitem{9042352}
M.~Mohammadi~Amiri and D.~G\"und\"uz, ``Machine learning at the wireless edge:
  Distributed stochastic gradient descent over-the-air,'' \emph{IEEE
  Transactions on Signal Processing}, vol.~68, pp. 2155--2169, 2020.

\bibitem{9772390}
Z.~Jiang, G.~Yu, Y.~Cai, and Y.~Jiang, ``{Decentralized Edge Learning via
  Unreliable Device-to-Device Communications},'' \emph{IEEE Transactions on
  Wireless Communications}, vol.~21, no.~11, pp. 9041--9055, 2022.

\bibitem{9716792}
H.~Ye, L.~Liang, and G.~Y. Li, ``{Decentralized Federated Learning With
  Unreliable Communications},'' \emph{IEEE Journal of Selected Topics in Signal
  Processing}, vol.~16, no.~3, pp. 487--500, 2022.

\bibitem{9838891}
E.~Jeong, M.~Zecchin, and M.~Kountouris, ``{Asynchronous Decentralized Learning
  over Unreliable Wireless Networks},'' in \emph{IEEE International Conference
  on Communications}, 2022, pp. 607--612.

\bibitem{8786146}
A.~Reisizadeh, A.~Mokhtari, H.~Hassani, and R.~Pedarsani, ``An exact quantized
  decentralized gradient descent algorithm,'' \emph{IEEE Transactions on Signal
  Processing}, vol.~67, no.~19, pp. 4934--4947, 2019.

\bibitem{fmnist}
H.~Xiao, K.~Rasul, and R.~Vollgraf, ``{Fashion-MNIST: a Novel Image Dataset for
  Benchmarking Machine Learning Algorithms},'' \emph{CoRR}, vol.
  abs/1708.07747, 2017.

\bibitem{1104412}
J.~Tsitsiklis, D.~Bertsekas, and M.~Athans, ``Distributed asynchronous
  deterministic and stochastic gradient optimization algorithms,'' \emph{IEEE
  Transactions on Automatic Control}, vol.~31, no.~9, pp. 803--812, 1986.

\bibitem{Taheri2020}
H.~Taheri, A.~Mokhtari, H.~Hassani, and R.~Pedarsani, ``Quantized decentralized
  stochastic learning over directed graphs,'' in \emph{Proc. 37th ICML}, Jul.
  2020.

\bibitem{Kovalev2020}
D.~Kovalev, A.~Koloskova, M.~Jaggi, P.~Richt\'{a}rik, and S.~U. Stich, ``A
  linearly convergent algorithm for decentralized optimization: Sending less
  bits for free!'' in \emph{Proc. 24th AISTATS}, Apr. 2021.

\bibitem{Liao2021}
Y.~Liao, Z.~Li, K.~Huang, and S.~Pu, ``A compressed gradient tracking method
  for decentralized optimization with linear convergence,'' \emph{IEEE Trans.
  on Automatic Control}, vol.~67, no.~10, pp. 5622--5629, 2022.

\bibitem{9562559}
M.~Chen, D.~G\"und\"uz, K.~Huang, W.~Saad, M.~Bennis, A.~V. Feljan, and H.~V.
  Poor, ``Distributed learning in wireless networks: Recent progress and future
  challenges,'' \emph{IEEE Journal on Selected Areas in Communications},
  vol.~39, no.~12, pp. 3579--3605, 2021.

\bibitem{4305404}
B.~Nazer and M.~Gastpar, ``{Computation Over Multiple-Access Channels},''
  \emph{IEEE Transactions on Information Theory}, vol.~53, no.~10, pp.
  3498--3516, 2007.

\bibitem{9815298}
J.~Choi, ``{Communication-Efficient Distributed SGD Using Random Access for
  Over-the-Air Computation},'' \emph{IEEE Journal on Selected Areas in
  Information Theory}, vol.~3, no.~2, pp. 206--216, 2022.

\bibitem{Faraz}
M.~Faraz and N.~Michelusi, ``{Biased Over-the-Air Federated Learning under
  Wireless Heterogeneity},'' in \emph{IEEE International Conference on
  Communications (ICC)}, 2024.

\bibitem{8974619}
J.~Dong, Y.~Shi, and Z.~Ding, ``{Blind Over-the-Air Computation and Data Fusion
  via Provable Wirtinger Flow},'' \emph{IEEE Transactions on Signal
  Processing}, vol.~68, pp. 1136--1151, 2020.

\bibitem{9076343}
T.~Sery and K.~Cohen, ``{On Analog Gradient Descent Learning Over Multiple
  Access Fading Channels},'' \emph{IEEE Transactions on Signal Processing},
  vol.~68, pp. 2897--2911, 2020.

\bibitem{9740125}
V.~Gandikota, D.~Kane, R.~K. Maity, and A.~Mazumdar, ``{vqSGD: Vector Quantized
  Stochastic Gradient Descent},'' \emph{IEEE Transactions on Information
  Theory}, vol.~68, no.~7, pp. 4573--4587, 2022.

\bibitem{9311931}
M.~Chen, H.~V. Poor, W.~Saad, and S.~Cui, ``Wireless communications for
  collaborative federated learning,'' \emph{IEEE Communications Magazine},
  vol.~58, no.~12, pp. 48--54, 2020.

\bibitem{9562522}
F.~P.-C. Lin, S.~Hosseinalipour, S.~S. Azam, C.~G. Brinton, and N.~Michelusi,
  ``{Semi-Decentralized Federated Learning With Cooperative D2D Local Model
  Aggregations},'' \emph{IEEE Journal on Selected Areas in Communications},
  vol.~39, no.~12, pp. 3851--3869, 2021.

\bibitem{9834707}
M.~Yemini, R.~Saha, E.~Ozfatura, D.~G\"und\"uz, and A.~J. Goldsmith,
  ``{Semi-Decentralized Federated Learning with Collaborative Relaying},'' in
  \emph{IEEE Intern. Symp. on Information Theory}, 2022, pp. 1471--1476.

\bibitem{9705093}
S.~Hosseinalipour, S.~S. Azam, C.~G. Brinton, N.~Michelusi, V.~Aggarwal, D.~J.
  Love, and H.~Dai, ``{Multi-Stage Hybrid Federated Learning Over Large-Scale
  D2D-Enabled Fog Networks},'' \emph{IEEE/ACM Transactions on Networking},
  vol.~30, no.~4, pp. 1569--1584, 2022.

\bibitem{Goldsmith_2005}
A.~Goldsmith, \emph{Wireless Communications}.\hskip 1em plus 0.5em minus
  0.4em\relax Cambridge Univ. Press, 2005.

\bibitem{10.5555/2613412}
J.~L. Gross, J.~Yellen, and P.~Zhang, \emph{Handbook of Graph Theory, Second
  Edition}, 2nd~ed.\hskip 1em plus 0.5em minus 0.4em\relax Chapman \& Hall/CRC,
  2013.

\bibitem{8664630}
S.~Wang, T.~Tuor, T.~Salonidis, K.~K. Leung, C.~Makaya, T.~He, and K.~Chan,
  ``{Adaptive Federated Learning in Resource Constrained Edge Computing
  Systems},'' \emph{IEEE Journal on Selected Areas in Communications}, vol.~37,
  no.~6, pp. 1205--1221, 2019.

\bibitem{florescu2013handbook}
I.~Florescu and C.~A. Tudor, \emph{Handbook of Probability}.\hskip 1em plus
  0.5em minus 0.4em\relax Wiley, 2013.

\bibitem{10360352}
N.~K. Jha, H.~Guo, and V.~K.~N. Lau, ``{Analog Product Coding for Over-the-Air
  Aggregation Over Burst-Sparse Interference Multiple-Access Channels},''
  \emph{IEEE Trans. on Signal Process.}, vol.~72, pp. 157--172, 2024.

\bibitem{10384479}
T.~Gafni, K.~Cohen, and Y.~C. Eldar, ``Federated learning from heterogeneous
  data via controlled air aggregation with bayesian estimation,'' \emph{IEEE
  Transactions on Signal Processing}, vol.~72, pp. 1928--1943, 2024.

\bibitem{6854643}
V.~Schwarz, G.~Hannak, and G.~Matz, ``{On the convergence of average consensus
  with generalized metropolis-hasting weights},'' in \emph{IEEE International
  Conference on Acoustics, Speech and Signal Processing (ICASSP)}, 2014, pp.
  5442--5446.

\bibitem{Wu2020OnTN}
J.~Wu, W.~Hu, H.~Xiong, J.~Huan, V.~Braverman, and Z.~Zhu, ``{On the Noisy
  Gradient Descent that Generalizes as SGD},'' in \emph{ICML}, 2020.

\bibitem{bertsekas2003convex}
D.~Bertsekas, A.~Nedi\'{c}, and A.~Ozdaglar, \emph{{Convex Analysis and
  Optimization}}.\hskip 1em plus 0.5em minus 0.4em\relax Athena Scientific,
  2003.

\end{thebibliography}

\begin{IEEEbiography}[{\includegraphics[width=1in,height=1.25in,clip,keepaspectratio]{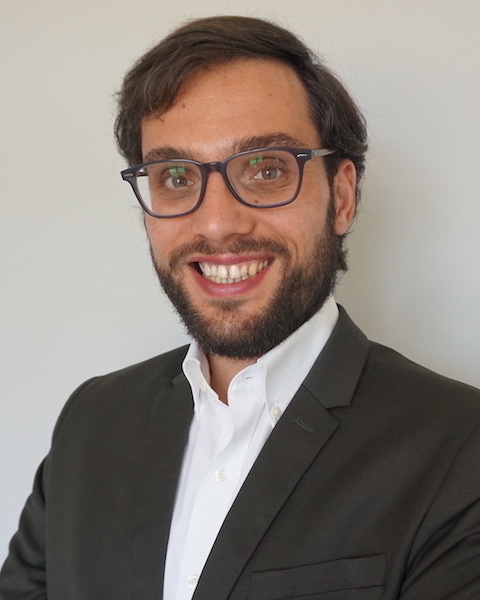}}]{Nicol\`{o} Michelusi}
(Senior Member, IEEE) received the B.Sc. (with honors), M.Sc. (with honors), and Ph.D. degrees from the University of Padova, Italy, in 2006, 2009, and 2013, respectively, and the M.Sc. degree in telecommunications engineering from the Technical University of Denmark, Denmark, in 2009, as part of the T.I.M.E. double degree program. From 2013 to 2015, he was a Postdoctoral Research Fellow with the Ming-Hsieh Department of Electrical Engineering, University of Southern California, Los Angeles, CA, USA, and from 2016 to 2020, he was an Assistant Professor with the School of Electrical and Computer Engineering, Purdue University, West Lafayette, IN, USA. He is currently an Associate Professor with the School of Electrical, Computer and Energy Engineering, Arizona State University, Tempe, AZ, USA. His research interests include 5G wireless networks, millimeter-wave communications, stochastic optimization, decentralized and federated learning over wireless systems. He served as Associate Editor for the IEEE TRANSACTIONS ON WIRELESS COMMUNICATIONS from 2016 to 2021, and currently serves as Editor for the IEEE TRANSACTIONS ON COMMUNICATIONS. He was the Co-Chair for the Distributed Machine Learning and Fog Network workshop at  IEEE INFOCOM 2021, 2023 and 2024, the Wireless Communications Symposium at  IEEE GLOBECOM 2020, the IoT, M2M, Sensor Networks, and Ad-Hoc Networking track at  IEEE VTC 2020, and the Cognitive Computing and Networking symposium at ICNC 2018. He was the Technical Area Chair for the Communication Systems track at Asilomar 2023. He received the NSF CAREER award in 2021, the IEEE Communication Theory Technical Committee (CTTC) Early Achievement
Award in 2022, and the IEEE Communications Society William R. Bennett Prize in 2024.  
\end{IEEEbiography}

\newpage

{
    \twocolumn[
    \centerline{\parbox[c][2cm][c]{10cm}{%
            \fontsize{24}{24}
            \selectfont
            {Supplemental Document}}}]
}

In this supplemental document, we provide the proof of Lemma \ref{L0} of the main manuscript. We also provide
additional numerical evaluations showcasing the parameters' selection of \proposed, as well as comparisons with state-of-the-art schemes.

\section*{Proof of Lemma \ref{L0}}
In this section, we provide the proof of Lemma \ref{L0}, restated below for ease of reference.

\noindent{\bf Lemma 2.}
Consider channels satisfying 
Assumption \ref{ch2}. Assume that the resource units are evenly allocated among $m=1,\dots, M$, i.e., $|R_m{-}Q/M|{<}1,\forall m$ (e.g., Fig. \ref{fig:ofdmframe}). 
Let:
\begin{align}
\label{varthetaproof}
&
\vartheta\triangleq\max_{i,j\neq i}\frac{1}{\Lambda_{ij}}\Big(\frac{1}{Q}\sum_{q=1}^Q\mathbb E[(|[\mathbf h_{ij}]_{q}|^2-\Lambda_{ij})^2]\Big)^{1/2},\\&
\label{varsigmaproof}
\varpi\triangleq\max_{i,j\neq i}
\frac{1}{\Lambda_{ij}}\Big(\frac{1}{M}\sum_{m=1}^M\mathbb E[(\hat\lambda_{ij}^{(m)}-\Lambda_{ij})^2]\Big)^{1/2},
\\
\label{Lambdaproof}
&\Lambda^*\triangleq\max_{i}\sum_{j\neq i}\Lambda_{ij},
\end{align}
where $\hat\lambda_{ij}^{(m)}=\frac{1}{{R_{m}}}\sum_{q\in\mathcal R_{m}}|[{\mathbf h}_{ij}]_{q}|^2$
is the sample average channel gain across the resource units in the set $\mathcal R_{m}$.
Then, $\frac{1}{N}\sum_{i=1}^N\mathrm{var}(\tilde{\mathbf d}_{i}){\leq}\Sigma^{(1)}
\triangleq\max\limits_{\mathbf z,\mathbf z'\in\mathcal Z}\Vert\mathbf z-\mathbf z'\Vert^2$
\\\centerline{$
\times \frac{1}{1-p_{\mathrm{tx}}}\Big[
 \frac{\sqrt{M}}{\sqrt{Q}}\sqrt{2(1+2\vartheta^2)}\Lambda^*
+\frac{\sqrt{1+\varpi^2}}{\sqrt{p_{\mathrm{tx}}}}\Lambda^*
+ \frac{\sqrt{M}}{\sqrt{Q}}\frac{N_0}{Ep_{\mathrm{tx}}}
\Big]^2.
$}
The value of the transmission probability $p_{\mathrm{tx}}$ minimizing this bound is the unique solution in $(0,1)$ of
\\\centerline{$
\sqrt{2(1{+}2\vartheta^2)}p_{\mathrm{tx}}^{3/2}{+}
\frac{\sqrt{Q}}{\sqrt{M}}\sqrt{1{+}\varpi^2}(2p_{\mathrm{tx}}{-}1){+}
\frac{N_0}{\Lambda^* E}\frac{3p_{\mathrm{tx}}{-}2}{\sqrt{p_{\mathrm{tx}}}}
=0.
$}
\begin{proof}
We start with
\begin{align}
\label{up1}
\frac{1}{N}\sum_{i=1}^N\mathrm{var}(\tilde{\mathbf d}_{i})
\leq
\max_i \mathrm{var}(\tilde{\mathbf d}_{i})=\max_i (\mathrm{sdv}(\tilde{\mathbf d}_{i}))^2.
\end{align}
Hence, we focus on bounding $\mathrm{sdv}(\tilde{\mathbf d}_{i})$ for a generic $i$.
Using \eqref{dik},
we can express 
$$
\tilde{\mathbf d}_{i}-\mathbb E[\tilde{\mathbf d}_{i}]=
\sum_{m=1}^M
({r_{im}}-\mathbb E[{r_{im}}])(\mathbf z_m-{\mathbf w_{i}}).
$$
Next, we use Minkowski inequality \cite[Lemma 14.10]{florescu2013handbook}, $\sqrt{\mathbb E[\Vert\sum_m\mathbf a_m\Vert^2]}\leq \sum_m\sqrt{\mathbb E[\Vert\mathbf a_m\Vert^2]}$. It
yields
\begin{align}
\nonumber
\mathrm{sdv}(\tilde{\mathbf d}_{i})&=
\sqrt{\mathbb E[\Vert\tilde{\mathbf d}_{i}-\mathbb E[\tilde{\mathbf d}_{i}]\Vert^2]}\\&\leq
\sum_{m=1}^M
\mathrm{sdv}({r_{im}})
\Vert\mathbf z_m-{\mathbf w_{i}}\Vert.
\label{up3}
\end{align}
We bound  $\Vert\mathbf z_m-{\mathbf w_{i}}\Vert$ by using the fact that  ${\mathbf w_{i}}=\sum_{m'=1}^M[\mathbf p_{i}]_{m'}\mathbf z_{m'}$
and the triangle inequality, yielding
$$
\Vert\mathbf z_m-{\mathbf w_{i}}\Vert\leq 
\sum_{m'=1}^M[\mathbf p_{i}]_{m'}\Vert\mathbf z_m-\mathbf z_{m'}\Vert
$$$$
\leq 
\max_{m',m''}\Vert\mathbf z_{m''}-\mathbf z_{m'}\Vert
=\max\limits_{\mathbf z,\mathbf z'\in\mathcal Z}\Vert\mathbf z-\mathbf z'\Vert,
$$
so that \eqref{up3} is further upper bounded as
\begin{align}
\mathrm{sdv}(\tilde{\mathbf d}_{i})\leq
\max\limits_{\mathbf z,\mathbf z'\in\mathcal Z}\Vert\mathbf z-\mathbf z'\Vert\sum_{m=1}^M
\mathrm{sdv}({r_{im}}).
\label{up2}
\end{align}
We now focus on bounding $\mathrm{sdv}({r_{im}})=\sqrt{\mathrm{var}({r_{im}})}$.
We use $\mathrm{var}(X)=\mathbb E_{A}[\mathrm{var}(X|A)]+\mathrm{var}(\mathbb E[X|A])$
to express
\begin{align}
\label{varsum}
&\mathrm{var}({r_{im}})
=
\mathbb E\Big[\mathrm{var}\Big({r_{im}}\Big|\varsigma,\chi_{-i},\chi_i,\mathbf h\Big)\Big]
\\&
+\mathrm{var}\Big(\mathbb E[{r_{im}}|\varsigma,\chi_{-i},\chi_i,\mathbf h]\Big)
.
\end{align}
The first variance term captures the impact of the AWGN noise and random phase injected at the transmitters;
the second term captures the impact of the randomness
 in the
 channels ($\mathbf h$), circular shift ($\varsigma$), random transmission of node $i$ ($\chi_i$) and
of all other nodes ($\chi_{-i}$).
Next, we bound each variance term individually.
\\
\underline{Bounding $\mathbb E[\mathrm{var}({r_{im}}|\varsigma,\chi_{-i},\chi_i,\mathbf h)]$}:
Let $m'\triangleq m\oplus\varsigma$ and 
$$A_{q}\triangleq\sum_{j\neq i}\chi_j[\mathbf h_{ij}]_{q}[\mathbf x_{j}]_{q}
$$$$
  =\sqrt{\frac{EQ}{R_{m'}}}\sum_{j\neq i}\chi_je^{\jj[\bstheta_j]_q}[{\mathbf h}_{ij}]_{q}
 \sqrt{[\mathbf p_{j}]_m},\ \forall q\in\mathcal R_{m'}.
  $$
  Then, we can express $[\mathbf y_{i}]_{q}=A_{q}+[\mathbf n_{i}]_{q}$
  and 
$$r_{im}
=
\frac{1-\chi_i}{EQp_{\mathrm{tx}}(1-p_{\mathrm{tx}})}\sum_{q\in\mathcal R_{m'}}(|A_{q}+[\mathbf n_{i}]_{q}|^2-N_0).$$
Note that, when conditioned on $\varsigma,\chi_{-i},\chi_i,\mathbf h$ (hence on $m'$), the terms
$|A_{q}+[\mathbf n_{i}]_{q}|^2$ are independent across $q$,
due to independence of $[\mathbf n_{i}]_{q}$, and  of the phases $[\bstheta_{j}]_{q}$ across $q\in\mathcal R_{m'}$ and $j$.
Therefore, 
\begin{align}
\label{varrim}
&\mathrm{var}({r_{im}}|\varsigma,\chi_{-i},\chi_i,\mathbf h)
=\frac{1-\chi_i}{E^2Q^2p_{\mathrm{tx}}^2(1-p_{\mathrm{tx}})^2}
\\&\qquad\times\sum_{q\in\mathcal R_{m'}}\mathrm{var}(|A_{q}+[\mathbf n_{i}]_{q}|^2|\varsigma,\chi_{-i},\mathbf h).
\nonumber
\end{align}
Using the independence of the noise and of the phases injected at the transmitters (across transmitters $j$),
the fact that
 $\mathbb E[e^{\jj[\bstheta_{j_1}-\bstheta_{j_2}+\bstheta_{j_3}-\bstheta_{j_4}]_{q}}]=
 \mathbbm{1}[j_3=j_4,j_1=j_2]+ \mathbbm{1}[j_3=j_2,j_4=j_1,j_1\neq j_2]
 $),
 and $\mathbb E[|[\mathbf n_{i}]_{q}|^4]=2N_0^2$, we find that
 $$
 \mathrm{var}(|A_{q}+[\mathbf n_{i}]_{q}|^2|\varsigma,\chi_{-i},\mathbf h)
=
\frac{EQ}{R_{m'}}\sum_{j\neq i}\chi_j|[\mathbf h_{ij}]_{q}|^2[\mathbf p_{j}]_m
 $$$$\times
\Big(\frac{EQ}{R_{m'}}\sum_{j'\neq i,j}\chi_{j'}|[\mathbf h_{ij'}]_{q}|^2[\mathbf p_{j'}]_m+2N_0\Big)
+N_0^2.
 $$
 Next,
we compute the expectation of this variance term with respect to the random transmissions of nodes $j\neq i$.
Since these are independent across $j\neq j'$ and occur with probability $p_{\mathrm{tx}}$, we obtain
$$
\mathbb E_{\chi_{-i}}[ \mathrm{var}(|A_{q}+[\mathbf n_{i}]_{q}|^2|\varsigma,\chi_{-i},\mathbf h)]
$$$$
{=}
\Big(\!\frac{EQp_{\mathrm{tx}}}{R_{m'}}\!\sum_{j\neq i}\!|[\mathbf h_{ij}]_{q}|^2[\mathbf p_{j}]_m+N_0\!\Big)^2
$$$$
-\frac{(EQp_{\mathrm{tx}})^2}{(R_{m'})^2}\!\!\sum_{j\neq i}|[\mathbf h_{ij}]_{q}|^4[\mathbf p_{j}]_m^2.
 $$
  Substituting in the expression of $\mathrm{var}({r_{im}}|\varsigma,\chi_{-i},\chi_i,\mathbf h)$ in \eqref{varrim},
  discarding the negative term above,
 and further taking the expectation with respect to the transmission indicator $\chi_i\in\{0,1\}$, 
   the circular shift $\varsigma$ (hence $m'=m\oplus\varsigma$ becomes uniform in $\{1,\dots,M\}$), 
   and the channels,
 yields
\begin{align}
\nonumber
&\mathbb E[\mathrm{var}({r_{im}}|\varsigma,\chi_{-i},\chi_i,\mathbf h)]
 \leq\frac{1}{E^2Q^2p_{\mathrm{tx}}^2(1-p_{\mathrm{tx}})}
 \\&
\times\frac{1}{M}\sum_{m',q}\mathbb E_{\mathbf h}
\Big[\Big(\!\frac{EQp_{\mathrm{tx}}}{R_{m'}}\!\sum_{j\neq i}\!|[\mathbf h_{ij}]_{q}|^2[\mathbf p_{j}]_m+N_0\!\Big)^2\Big],
\label{varrim2}
\end{align}
where
we defined the shorthand notation $\sum_{m',q}\equiv\sum_{m'=1}^M\sum_{q\in\mathcal R_{m'}}$.
The term above takes the form
$$
\sum_t\mathbb E[(\sum_{j}\alpha_{jt})^2]
$$
for suitable $\alpha_{jt}$. We can upper bound it as
$$
=
\sum_{j',j}\sum_t\mathbb E[\alpha_{jt}\alpha_{j't}]
\leq\sum_{j',j}\sum_t\sqrt{\mathbb E[\alpha_{jt}^2]}\sqrt{\mathbb E[\alpha_{j't}^2]}
$$$$
\leq
\sum_{j',j}\sqrt{\sum_t\mathbb E[\alpha_{jt}^2]}\sqrt{\sum_t\mathbb E[\alpha_{j't}^2]}
=
\Big(\sum_{j}\sqrt{\sum_t\mathbb E[\alpha_{jt}^2]}\Big)^2,
$$
where in the first step we used Holder's inequality,
in the second step we used \CS.
 Using this bound in \eqref{varrim2} yields
 $$\mathbb E[\mathrm{var}({r_{im}}|\varsigma,\mathcal T_{-i},\tau_i,\mathbf h)]
 $$
\begin{align}
\label{varrim3}
&
\leq\frac{1}{1-p_{\mathrm{tx}}}
\Big(
\sum_{j\neq i}
\sqrt{\sum_{m',q}\frac{1}{MR_{m'}^2}\mathbb E[|[\mathbf h_{ij}]_{q}|^4]}[\mathbf p_{j}]_m
\\&\nonumber
\qquad\qquad+\frac{N_0}{\sqrt{MQ}Ep_{\mathrm{tx}}}
\Big)^2.
\end{align}
We further bound
$$
\sum_{m',q}\frac{1}{MR_{m'}^2}
\mathbb E[|[\mathbf h_{ij}]_{q}|^4]\leq\frac{2M}{Q}\sum_{m',q}\frac{1}{MR_{m'}}
\mathbb E[|[\mathbf h_{ij}]_{q}|^4]
$$$$
=\frac{2M}{Q}\Big(\sum_{m'=1}^M\sum_{q\in\mathcal R_{m'}}\frac{1}{MR_{m'}}\mathbb E[(|[\mathbf h_{ij}]_{q}|^2-\Lambda_{ij})^2]
+\Lambda_{ij}^2\Big)
$$$$
\leq\frac{2M}{Q}\Big(\frac{2}{Q}\sum_{q=1}^Q\mathbb E[(|[\mathbf h_{ij}]_{q}|^2-\Lambda_{ij})^2]
+\Lambda_{ij}^2\Big)
$$$$
\leq
\frac{2M}{Q}(1+2\vartheta^2)\Lambda_{ij}^2,
$$
where in the first and third steps we used
$1/R_{m'}\leq 2M/Q$ (a consequence of the uniform allocation of resource units),
in the second step we used the fact that (Assumption \ref{ch2})
$$
\frac{1}{R_{m'}}\sum_{q\in\mathcal R_{m'}}\mathbb E[|[\mathbf h_{ij}]_{q}|^2]=\Lambda_{ij}^{(m')}
,\ \Lambda_{ij}=\frac{1}{M}\sum_{m'=1}^M\Lambda_{ij}^{(m')},$$
and in the last step we used \eqref{varthetaproof}.
Substituting this bound in \eqref{varrim3} yields
$$\mathbb E[\mathrm{var}({r_{im}}|\varsigma,\chi_{-i},\chi_i,\mathbf h)]$$
$$
 \leq
 \frac{M}{Q(1-p_{\mathrm{tx}})}
\Big(
\sqrt{2(1+2\vartheta^2)}\sum_{j\neq i}\Lambda_{ij}[\mathbf p_{j}]_m
+\frac{N_0}{EMp_{\mathrm{tx}}}
\Big)^2.
$$
 \noindent \underline{Bounding $
\mathrm{var}(\mathbb E[{r_{im}}|\varsigma,\chi_{-i},\chi_i,\mathbf h])
$}: Note that
\begin{align}\label{varrimx}
\mathrm{var}(\mathbb E[{r_{im}}|\varsigma,\chi_{-i},\chi_i,\mathbf h])\leq
\mathbb E[\mathbb E[{r_{im}}|\varsigma,\chi_{-i},\chi_i,\mathbf h]^2].
\end{align}
From \eqref{fdgnhdsf3} and \eqref{yi},
 \begin{align}
&
\!\!\mathbb E[{r_{im}}|\varsigma,\chi_{-i},\chi_i,\mathbf h]
{=}
(1-\chi_i)\frac{\sum_{j\neq i}\chi_j\hat\lambda_{ij}^{(m')}[\mathbf p_{j}]_m}{p_{\mathrm{tx}}(1-p_{\mathrm{tx}})}
\end{align}
where $\hat\lambda_{ij}^{(m')}=\frac{1}{{R_{m'}}}\sum_{q\in\mathcal R_{m'}}|[{\mathbf h}_{ij}]_{q}|^2$ is the sample average gain across the resource units in the set $\mathcal R_{m'}$.
We square and 
take the expectation with respect to $(\varsigma,\chi_{-i},\chi_i,\mathbf h)$, yielding
\begin{align}
&\nonumber
\mathbb E[\mathbb E[{r_{im}}|\varsigma,\chi_{-i},\chi_i,\mathbf h]^2]
=\nonumber
\frac{\mathbb E[(\sum_{j\neq i}\chi_j\hat\lambda_{ij}^{(m')}[\mathbf p_{j}]_m)^2]}{p_{\mathrm{tx}}^2(1-p_{\mathrm{tx}})}
\\&
\leq
\frac{\Big(\sum_{j\neq i}\sqrt{\mathbb E[\chi_j(\hat\lambda_{ij}^{(m')})^2]}[\mathbf p_{j}]_m\Big)^2}{p_{\mathrm{tx}}^2(1-p_{\mathrm{tx}})}
\label{fghsdsdf}
\end{align}
where in the last step we used Holder's inequality
$\mathbb E[AB]=\sqrt{\mathbb E[A^2]}\sqrt{\mathbb E[B^2]}$
with the random variables $A=\chi_j\hat\lambda_{ij}^{(m')}$
and $B=\chi_{j'}\hat\lambda_{ij'}^{(m')}$.
Furthermore,
$$
\mathbb E[\chi_j(\hat\lambda_{ij}^{(m')})^2]
=
\frac{p_{\mathrm{tx}}}{M}\sum_{m'=1}^M\mathbb E[(\hat\lambda_{ij}^{(m')})^2],
$$
where we computed the expectation with respect to the transmission indicator
 and the circular shift $\varsigma$ (hence $m'=m\oplus\varsigma$ uniform in $\{1,\dots,M\}$). We continue as
 $$
 \mathbb E[\chi_j(\hat\lambda_{ij}^{(m')})^2]
=
p_{\mathrm{tx}}\Big(\frac{1}{M}\sum_{m'=1}^M\mathbb E[(\hat\lambda_{ij}^{(m')}-\Lambda_{ij})^2]+\Lambda_{ij}^2\Big)
$$$$
\leq
p_{\mathrm{tx}}(1+\varpi^2)\Lambda_{ij}^2,
$$
where we used the fact that $\frac{1}{M}\mathbb E[\hat\lambda_{ij}^{(m')}]=\Lambda_{ij}^{(m')}$
and $\Lambda_{ij}=\frac{1}{M}\sum_{m'=1}^M\Lambda_{ij}^{(m')}$, followed by \eqref{varsigmaproof}.
Substituting this bound into \eqref{fghsdsdf}
and then in \eqref{varrimx}
 yields
 $$
\mathbb E[\mathrm{var}({r_{im}}|\varsigma,\mathcal T_{-i},\tau_i,\mathbf h)]
\leq
\frac{1+\varpi^2}{p_{\mathrm{tx}}(1-p_{\mathrm{tx}})}\Big(\sum_{j\neq i}\Lambda_{ij}[\mathbf p_{j}]_m\Big)^2.
$$
 \noindent\underline{Final results}:
Combining the results together into \eqref{varsum} yields
$$
\mathrm{var}(r_{im})\leq
\frac{1+\varpi^2}{p_{\mathrm{tx}}(1-p_{\mathrm{tx}})}\Big(\sum_{j\neq i}\Lambda_{ij}[\mathbf p_{j}]_m\Big)^2
$$$$
 +\frac{M}{Q(1-p_{\mathrm{tx}})}
\Big(
\sqrt{2(1+2\vartheta^2)}\sum_{j\neq i}\Lambda_{ij}[\mathbf p_{j}]_m
+\frac{N_0}{EMp_{\mathrm{tx}}}
\Big)^2.
$$

\begin{figure*}
     \centering
               \hfill
          \begin{subfigure}[b]{0.25\linewidth}
        \includegraphics[width = \linewidth,trim=10 0 40 20, clip]{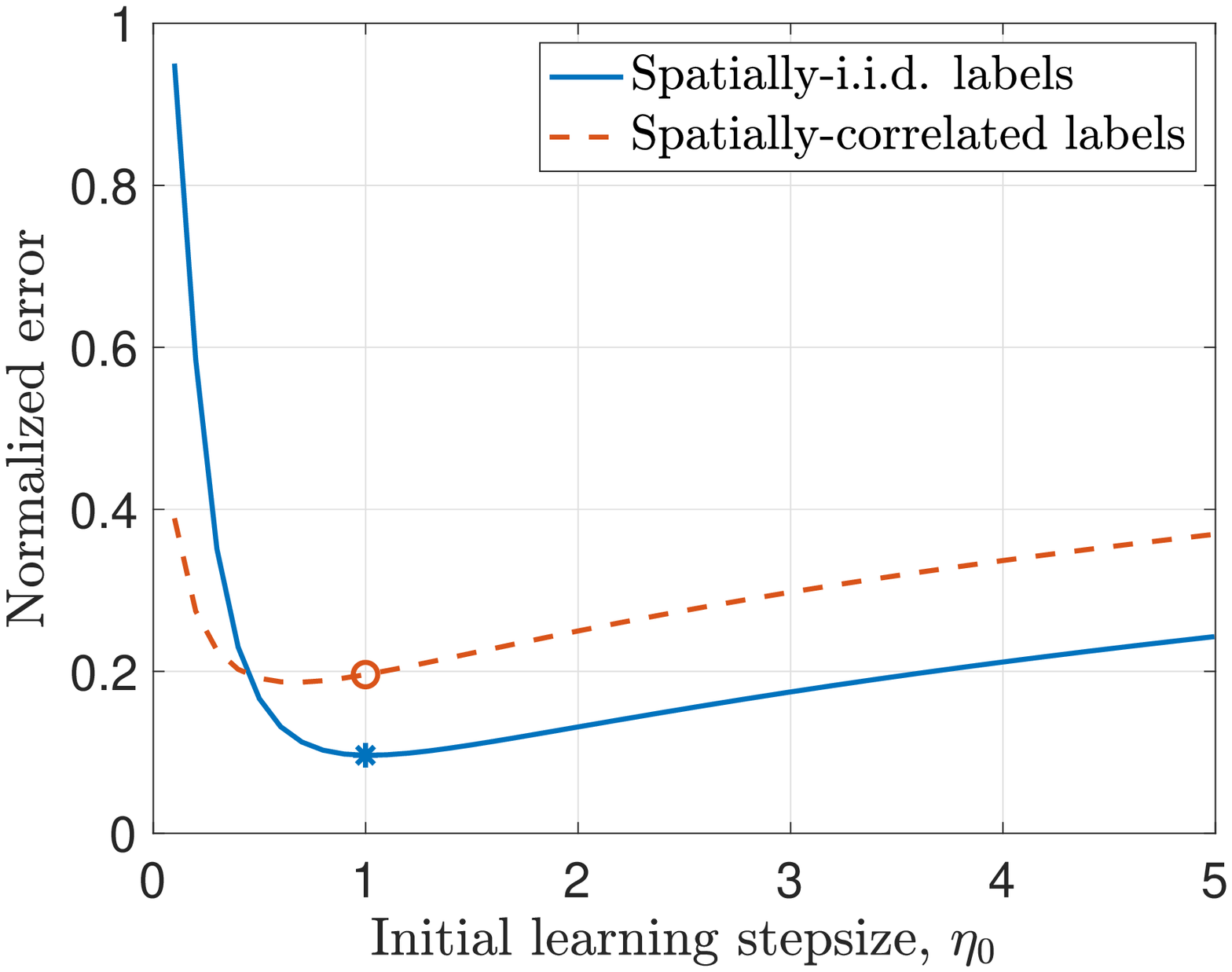}
        	    \vspace{-5mm}
	    \caption{} 
     \end{subfigure}
     \hfill
     \begin{subfigure}[b]{0.235\linewidth}
         \includegraphics[width = \linewidth,trim=40 0 40 20, clip]{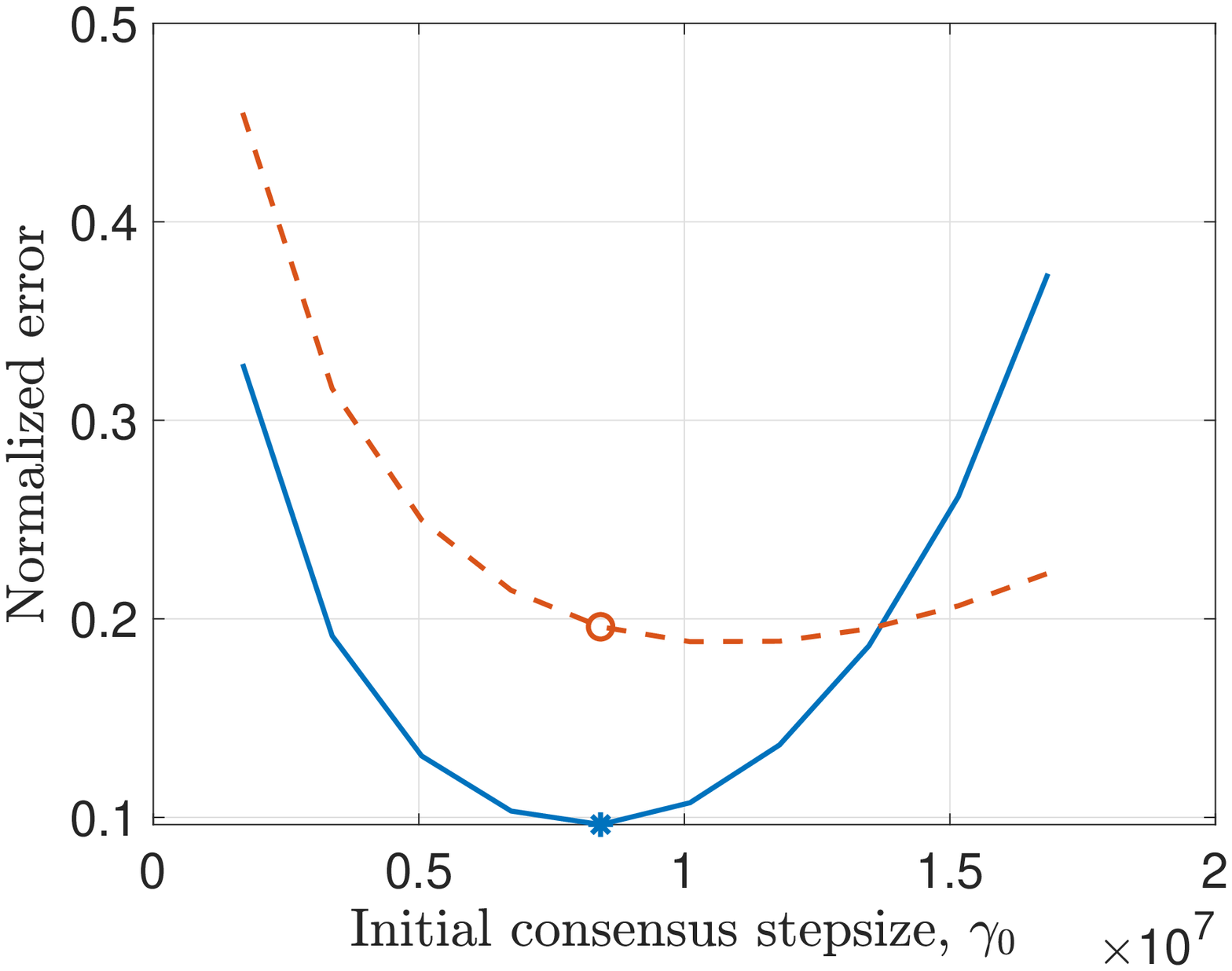}
         \vspace{-5mm}
	    \caption{} 
     \end{subfigure}
     \begin{subfigure}[b]{0.24\linewidth}
        \includegraphics[width = \linewidth,trim=28 0 40 20, clip]{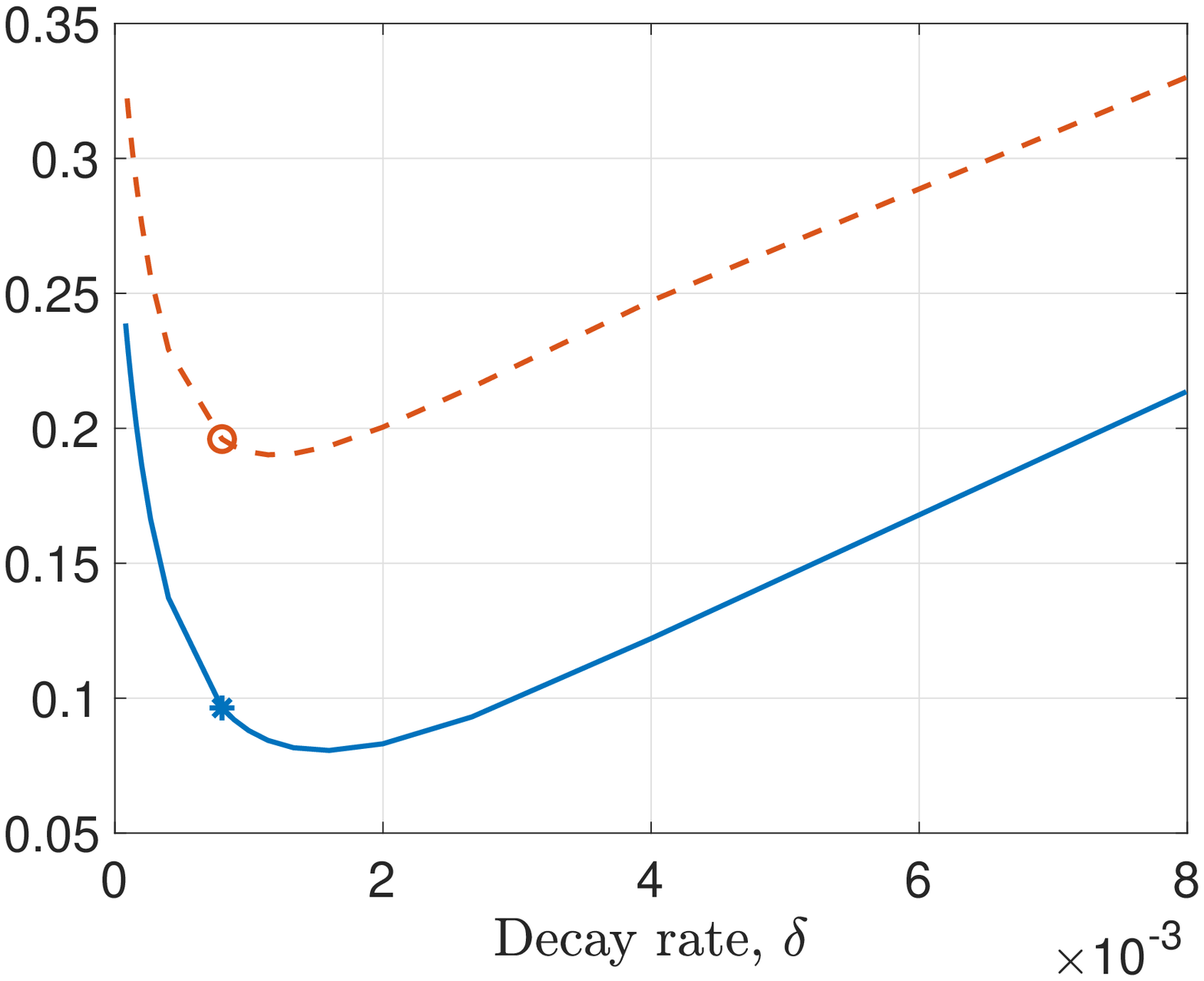}
       \vspace{-5mm}
	    \caption{} 
     \end{subfigure}
     \begin{subfigure}[b]{0.24\linewidth}
        \includegraphics[width = \linewidth,trim=30 0 40 20, clip]{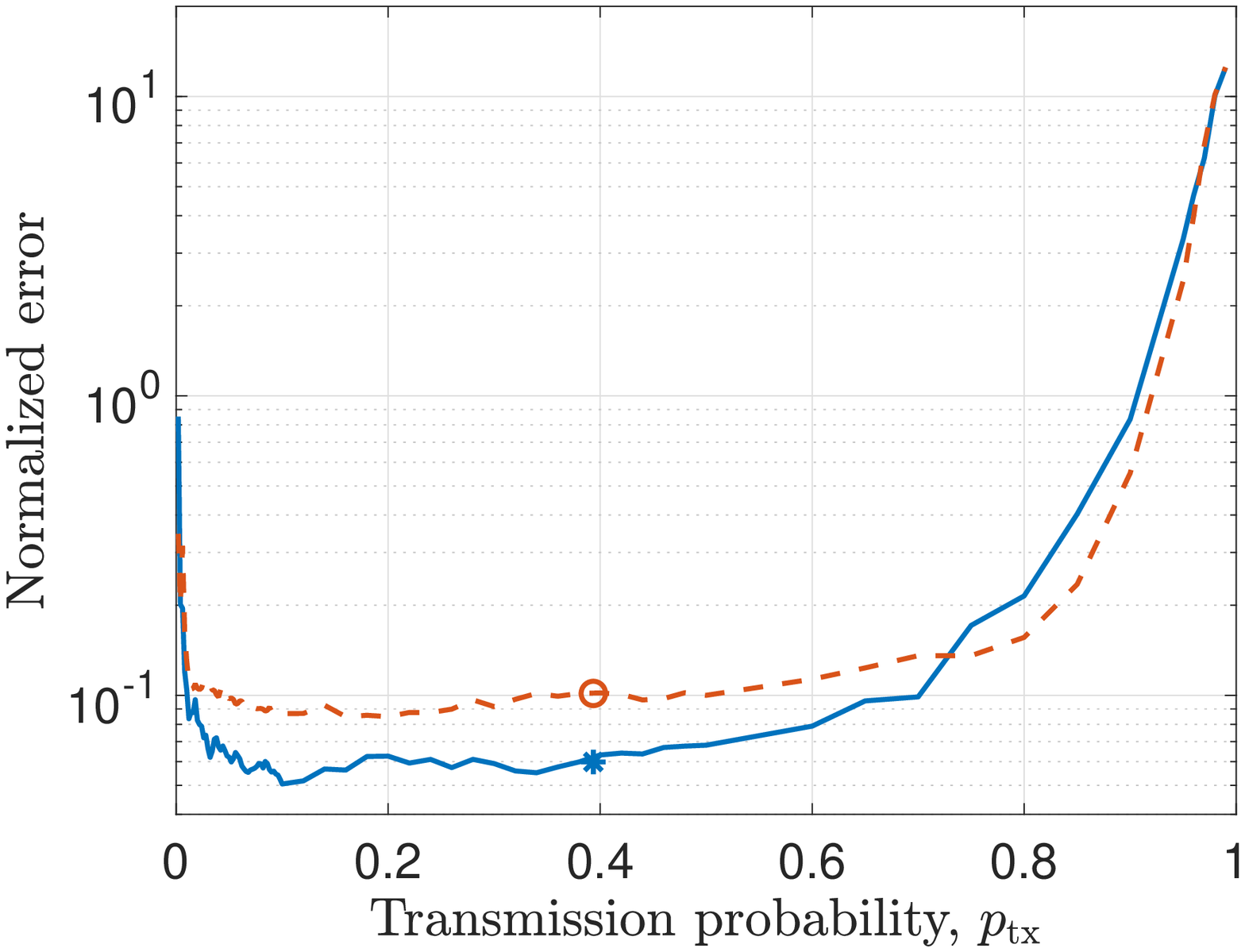}
        	 \vspace{-5mm}
	    \caption{} 
     \end{subfigure}     
          \hfill
\caption{Normalized error vs initial learning stepsize (a),
 initial consensus stepsize (b),
 decay rate  $\delta$ (c),
 and transmission probability $p_{\mathrm{tx}}$ (d), evaluated after 2000ms of execution time,
under spatially-i.i.d. (solid lines) and -dependent (dashed lines) scenarios, with i.i.d. channels over frames.
The markers denote the \emph{Baseline} configuration, adopted in \secref{numres}:
$\eta_0=\frac{2}{\mu+L}$, $\gamma_0=\frac{0.05}{\rho_2}$, $\delta=\frac{4}{5}\mu\eta_0$, $p_{\mathrm{tx}}$ as in Lemma \ref{L0}.
The common legend is shown in figure (a).
\vspace{-5mm}
 } \label{fig:params}
\end{figure*}

Finally, we take the square root of both sides,
along with $\sqrt{\sum_\ell a_\ell^2}\leq \sum_\ell a_\ell$ for $a_i\geq 0$,
add over $m$ and use the fact that $\sum_m[\mathbf p_{j}]_m=1$. These steps yield
$$
\sum_{m=1}^M\mathrm{sdv}(r_{im})
\leq
 \frac{\sqrt{M}}{\sqrt{Q}\sqrt{1-p_{\mathrm{tx}}}}\sqrt{2(1+2\vartheta^2)}\sum_{j\neq i}\Lambda_{ij}
 $$$$
+ \frac{\sqrt{M}}{\sqrt{Q}\sqrt{1-p_{\mathrm{tx}}}}\frac{N_0}{Ep_{\mathrm{tx}}}
+\frac{\sqrt{1+\varpi^2}}{\sqrt{p_{\mathrm{tx}}}\sqrt{1-p_{\mathrm{tx}}}}\sum_{j\neq i}\Lambda_{ij}\triangleq g(p_{\mathrm{tx}}).
$$
The final result is obtained by using the previous bound in \eqref{up2}, along with $\sum_{j\neq i}\Lambda_{ij}\leq \Lambda^*$  (see \eqref{Lambdaproof}).
Finally, it is straightforward to show that  $\lim_{p_{\mathrm{tx}}\to 0}g(p_{\mathrm{tx}})=\infty$,
$\lim_{p_{\mathrm{tx}}\to 1}g(p_{\mathrm{tx}})=\infty$, and
$$g^\prime(p_{\mathrm{tx}})\propto Ap_{\mathrm{tx}}^{3/2}+B(2p_{\mathrm{tx}}-1)+C\frac{3p_{\mathrm{tx}}-2}{\sqrt{p_{\mathrm{tx}}}}
\triangleq h(p_{\mathrm{tx}})
$$
where $A= \frac{\sqrt{M}}{\sqrt{Q}}\sqrt{2(1+2\vartheta^2)}$,
$B=\sqrt{1+\varpi^2}$,
$C=\frac{\sqrt{M}}{\sqrt{Q}}\frac{N_0}{\Lambda^* E}$. The function $h(p_{\mathrm{tx}})$ is 
decreasing in $p_{\mathrm{tx}}$ (with $\lim_{p_{\mathrm{tx}}\to 0}h(p_{\mathrm{tx}})=-\infty$,
$h(1)=A+B+C$,
hence there exists a unique $p_{\mathrm{tx}}^*\in(0,1)$ such that $h(p_{\mathrm{tx}}^*)=0$,
hence $g^\prime(p_{\mathrm{tx}}^*)=0$.
Such $p_{\mathrm{tx}}^*$ minimizes $g(p_{\mathrm{tx}})$, hence the upper bound $\Sigma^{(1)}$ on 
$\frac{1}{N}\sum_{i=1}^N\mathrm{var}(\tilde{\mathbf d}_{i})$.
\end{proof}

\section*{Performance of \proposed\ with respect to stepsize parameters and transmission probability}
In this and the next sections, we provide additional numerical evaluations to \secref{numres}. The simulation parameters and algorithm descriptions are provided in \secref{numres}.

We evaluate the performance of \proposed\ in terms of its normalized error after 2000ms of execution time,
for a network with $N=200$ nodes.
We use the CP0 codebook of Example \ref{ex1} with $M{=}2d{+}1{=}901$ codewords,
requiring two OFDM symbols ($Q{=}1024$), and yielding a frame of duration
$T{=}2T_{\mathrm{ofdm}}{=}258\mu$s. 
It follows Algorithm~\ref{A1}.

In Fig. \ref{fig:params}, we plot the normalized error versus some key parameters of \proposed:
the initial learning stepsize $\eta_0$ (a),
the initial consensus stepsize $\gamma_0$ (b),
the decay rate $\delta$ (c),
and the transmission probability $p_{\mathrm{tx}}$ (d).
To this end, we  vary one such parameter, while setting the other parameters equal to the \emph{Baseline} configuration.
The latter adopts $\eta_0=\frac{2}{\mu+L}$, $\gamma_0=\frac{0.05}{\rho_2}$, $\delta=\frac{4}{5}\mu\eta_0$, $p_{\mathrm{tx}}$ as in Lemma \ref{L0}.
For instance, Fig.~\ref{fig:params}.a is obtained by varying $\eta_0$, while setting the $\gamma_0$, $\delta$ and $p_{\mathrm{tx}}$ parameters as their baseline values.

In Fig. \ref{fig:params}.a, we plot the normalized error versus the initial learning stepsize $\eta_0$.
As expected, tuning of $\eta_0$ regulates the magnitude of gradient steps, hence the learning progress. A smaller $\eta_0$ slows down the progress of the algorithm, resulting in slower convergence and higher normalized error.
On the other hand, if the stepsize is too large, the optimization algorithm might oscillate around the minimizer or even diverge away from it.
Note that the baseline value $\eta_0=\frac{2}{\mu+L}$ (denoted with markers) is the maximum stepsize that guarantees convergence of gradient descent, for the class of strongly-convex and smooth loss functions. 

In Fig. \ref{fig:params}.b, we plot the normalized error versus the initial consensus stepsize $\gamma_0$.
As explained in the discussion following  Theorem \ref{T1},
tuning $\gamma_0$ reflects a delicate balance between speeding up information propagation (favored by larger $\gamma_0$) and minimizing error propagation (smaller $\gamma_0$). This trade-off is in line with the behavior observed in the figure.

In Fig. \ref{fig:params}.c, we plot the normalized error versus the decay rate $\delta$.
As explained in the discussion following Theorem \ref{T1}, the choice $\delta=\frac{4}{5}\mu\eta_0$ is the one that minimizes
 the convergence bounds stated in Theorem \ref{T1}, hence the value $\delta=\frac{4}{5}\mu\eta_0$ is employed in the numerical evaluations of \secref{numres}.
 Remarkably, the plots of \ref{fig:params}.c demonstrate the same monotonic behavior expected from Theorem \ref{T1} in the regime $\delta\leq\frac{4}{5}\mu\eta_0$. They also demonstrate that,
 if the condition $\delta\leq\frac{4}{5}\mu\eta_0$ of Theorem~\ref{T1} is violated, the normalized error quickly diverges. In this regime, 
 the stepsizes decrease too quickly, preventing the algorithm from making a tangible progress.
This numerical evaluation further motivates the use of decreasing vs constant ($\delta\to0$) stepsizes, developed in this paper.

In Fig. \ref{fig:params}.d, we plot the normalized error versus the transmission probability $p_{\mathrm{tx}}$.
We remind that the baseline value shown with markers is the one given by Lemma \ref{L0} and used in the numerical evaluations of \secref{numres}.
This value minimizes the upper bound on the variance of the disagreement signal estimate, given in Lemma \ref{L0}.
Indeed, Fig. \ref{fig:params}.d confirms that the optimal $p_{\mathrm{tx}}$ value found via Lemma \ref{L0}
 attains nearly optimal numerical performance. Furthermore, 
the normalized error exhibits a mild dependence with respect to $p_{\mathrm{tx}}$  around such optimal value (in the interval $[0.1,0.5]$), and quickly diverges as
$p_{\mathrm{tx}}\to 0$ (the transmissions are too sporadic) or $p_{\mathrm{tx}}\to 1$ (each node operates as a receiver too infrequently).
This implies that precise knowledge of propagation parameters $\vartheta$, $\varpi$ and $\Lambda^*$, affecting the design of $p_{\mathrm{tx}}$, is not critical.

\section*{Additional comparisons with\\ state-of-the-art schemes}

\begin{figure*}
     \centering
               \hfill
          \begin{subfigure}[b]{0.32\linewidth}
        \includegraphics[width = \linewidth,trim=10 0 20 20, clip]{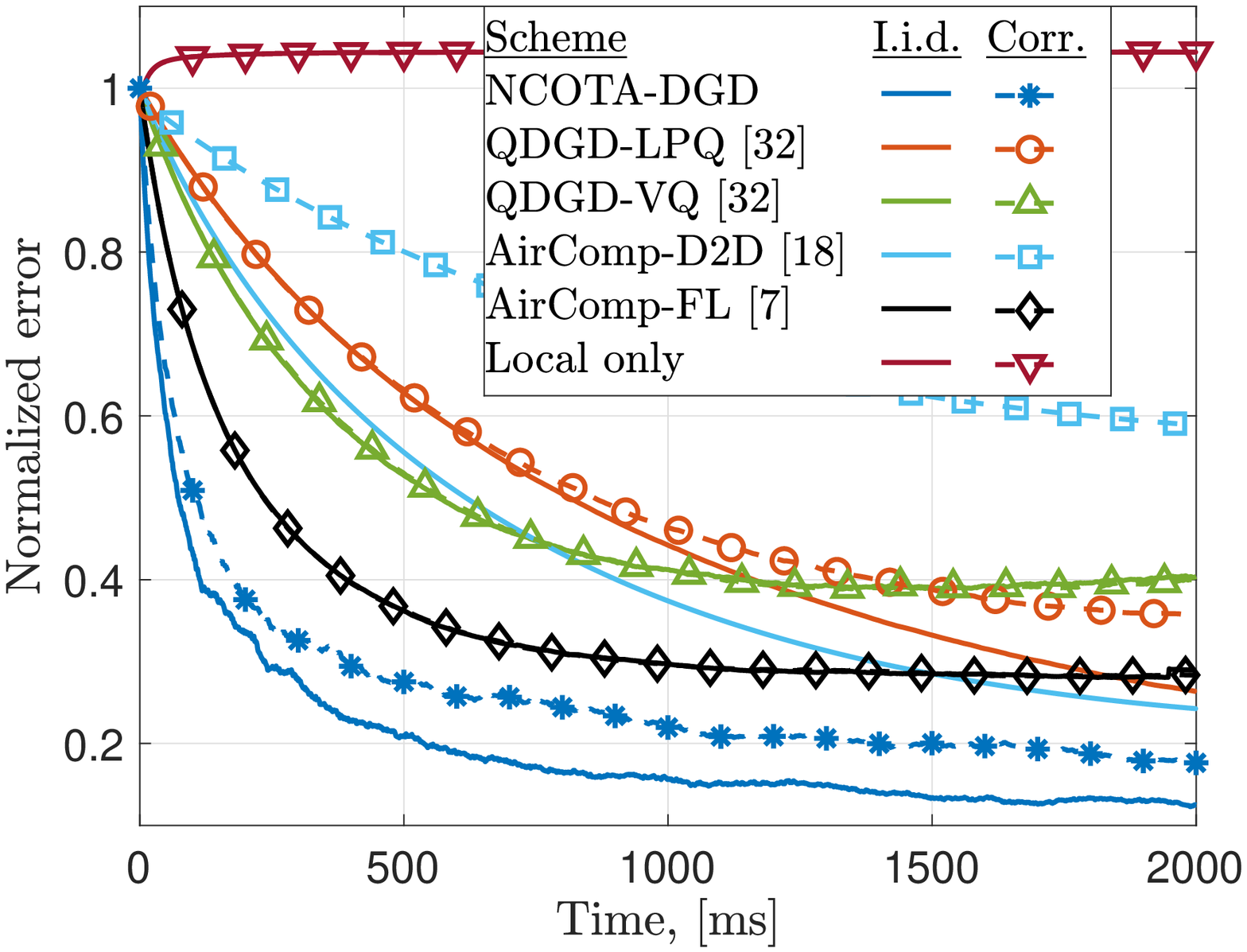}
        	    \vspace{-5mm}
	    \caption{} 
     \end{subfigure}
     \hfill
     \begin{subfigure}[b]{0.32\linewidth}
         \includegraphics[width = \linewidth,trim=10 0 20 20, clip]{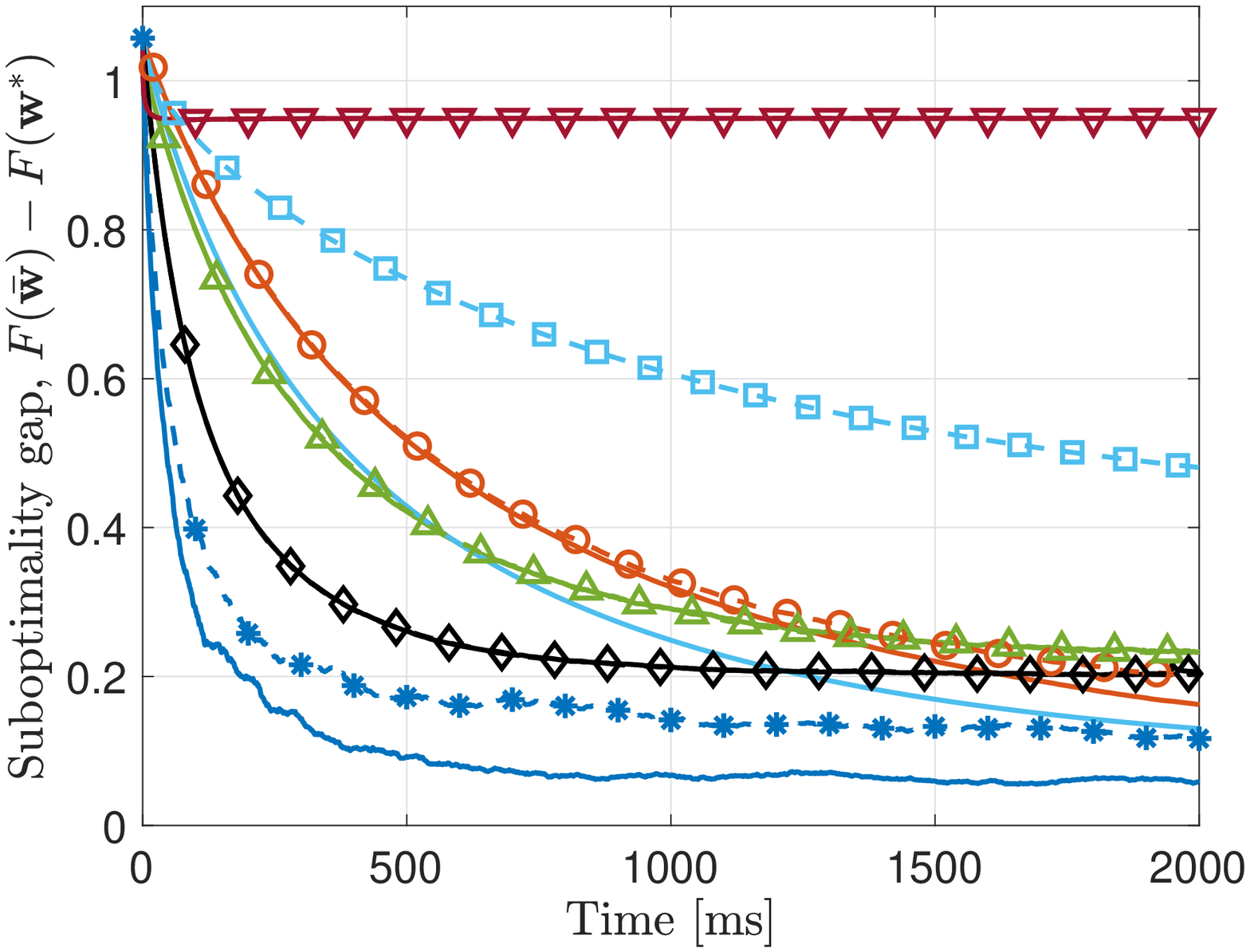}
         \vspace{-5mm}
	    \caption{} 
     \end{subfigure}
     \begin{subfigure}[b]{0.32\linewidth}
        \includegraphics[width = \linewidth,trim=10 0 20 20, clip]{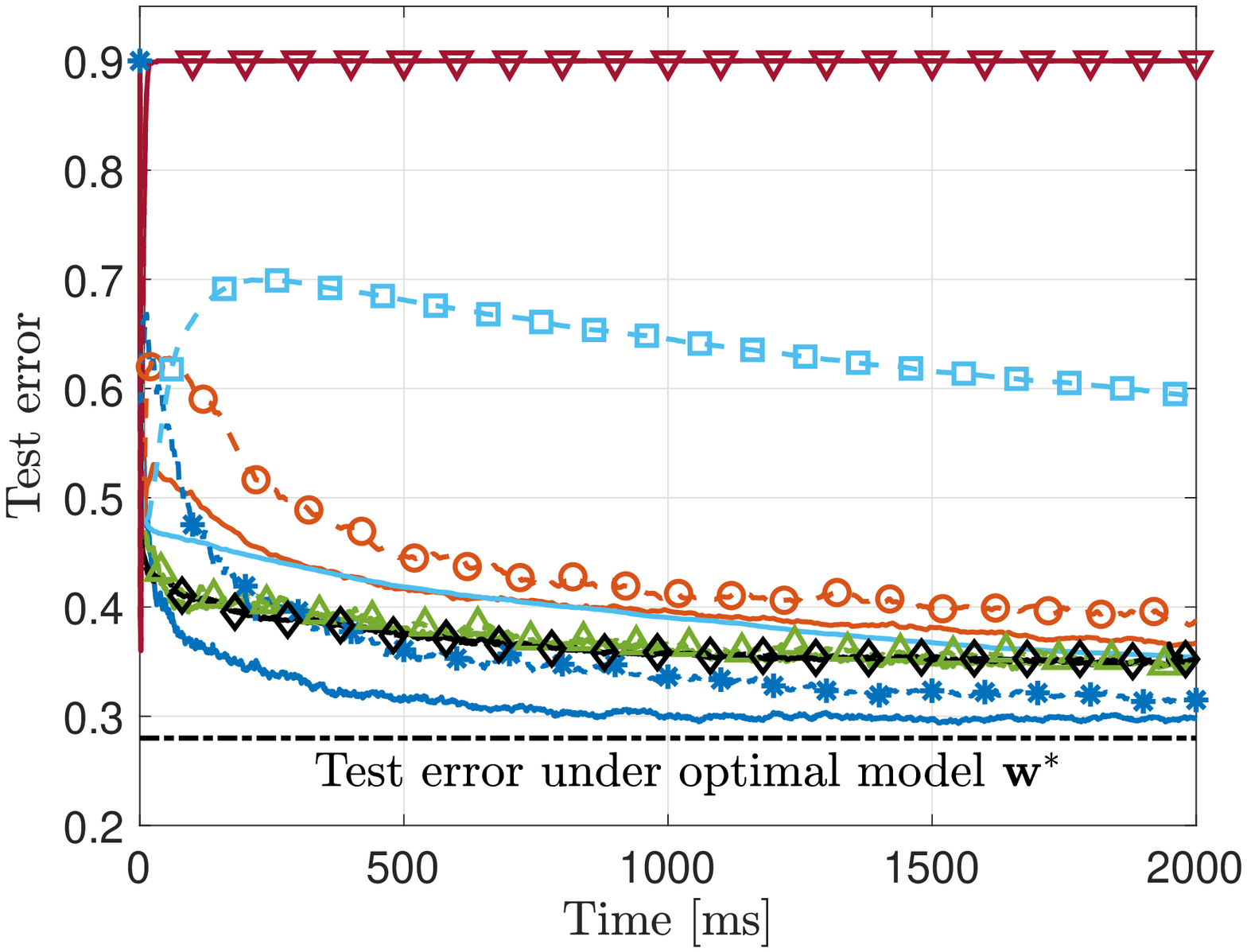}
       \vspace{-5mm}
	    \caption{} 
     \end{subfigure}
          \hfill
          \caption{Normalized error (a), suboptimality gap (b), test error (c)  vs time,
under both spatially-i.i.d. and -dependent label scenarios, with i.i.d. channels over frames.
Common legend shown in figure (a).
\vspace{-6mm}
 }
\label{fig:soa}
\end{figure*}

Fig. \ref{fig:soa} depicts the performance of the state-of-the-art schemes implemented in \secref{compSoA} versus time, for the scenario with 
$N=200$ nodes, with i.i.d. channels over frames.
We evaluate the performance under both 
spatially-i.i.d. (solid lines) and spatially-dependent (dashed with markers) label scenarios.
 As also noted in the comments of Fig. \ref{fig:SoAvsN},
smaller normalized error typically translates to smaller suboptimality gap and test error,
and the performance generally degrades in the spatially-dependent scenario.\

In both scenarios, we note that \proposed\ achieves the best performance across time, followed by \emph{AirComp-FL}.
The\newpage\noindent
 worst performance is attained by \emph{Local only}, since each node 
 optimizes its local function and does not exploit inter-agent  communications.
As noted in the comments of \secref{compSoA}, \emph{AirComp-FL} benefits from its star-based topology: all communications happen to and from the parameter server, which maintains model synchronization across the network. However, this scheme also
 suffers from the same, due to the communication bottleneck experienced by the edge nodes.

\end{document}